\documentclass[11pt]{article}

\usepackage{latexsym}
\usepackage{amssymb, amsmath, mathtools, subdepth, amsthm,statmath, dutchcal}
\mathtoolsset{showonlyrefs} 

\usepackage{graphicx}
\usepackage[latin9]{inputenc}
\usepackage{setspace}
\usepackage[margin=1.25in,nohead]{geometry}

\usepackage{titling}
\usepackage{url}
\usepackage{changepage}
\usepackage{enumerate, enumitem}

\usepackage{pdfpages}
\usepackage{relsize}
\usepackage{changepage}

\usepackage{endnotes}
\usepackage{environ}
\usepackage{wasysym}
\usepackage{subcaption}
\usepackage[multiple]{footmisc}
\usepackage{titlesec}

\usepackage{bbm}
\usepackage{diagbox}
\usepackage{bm}

\usepackage{caption, threeparttable, widetable, booktabs, multirow, tabularx, lscape, siunitx}
\usepackage[toc,title]{appendix}
\usepackage{float}
\usepackage[longnamesfirst]{natbib}
\usepackage{relsize}
\usepackage[T1]{fontenc}

\usepackage{dcolumn}
\newcolumntype{d}[1]{D{.}{.}{#1}}
\newcommand\mc[1]{\multicolumn{1}{c}{#1}} 

\theoremstyle{plain}

\newtheorem{proposition}{Proposition}
\newtheorem{corollary}{Corollary}
\newtheorem{lemma}{Lemma}
\newtheorem{lem}{Lemma}[section]

\theoremstyle{definition}
\newtheorem{assumption}{Assumption}

\newtheorem{manualtheoreminner}{Assumption}
\newenvironment{manualtheorem}[1]{%
  \renewcommand\themanualtheoreminner{#1}%
  \manualtheoreminner
}{\endmanualtheoreminner}

\theoremstyle{remark}
\newtheorem{remark}{Remark}

\newcommand\inverse{^{-1}}
\DeclarePairedDelimiter{\abs}{\lvert}{\rvert}
\DeclarePairedDelimiter{\norm}{\lVert}{\rVert}

\makeatletter
\newsavebox{\measure@tikzpicture}
\NewEnviron{scaletikzpicturetowidth}[1]{%
  \def\tikz@width{#1}%
  \def\tikzscale{1}\begin{lrbox}{\measure@tikzpicture}%
  \BODY
  \end{lrbox}%
  \pgfmathparse{#1/\wd\measure@tikzpicture}%
  \edef\tikzscale{\pgfmathresult}%
  \BODY
}
\makeatother

\usepackage{pgfplots}
\usetikzlibrary{patterns}
\usetikzlibrary{arrows.meta}
\usetikzlibrary{decorations.pathmorphing}
\usepackage{comment}

\usepackage{xcolor}
\usepackage[colorlinks = true,
            linkcolor = blue,
            urlcolor  = blue,
            citecolor = blue,
            anchorcolor = blue]{hyperref}

\RequirePackage{fix-cm}
\usepackage{multibib}
\newcites{A}{Appendix References}
\usepackage{tikz}
\usetikzlibrary{shapes.geometric, arrows}
\tikzstyle{flowbox} = [rectangle, rounded corners, text width = 3.1cm, minimum height = 1cm, text ragged, draw=black]
\tikzstyle{arrow} = [thick,->,>=stealth]

\begin{document}
\setlength{\abovedisplayskip}{6pt}
\setlength{\belowdisplayskip}{6pt}
\setlist{itemsep=0.1pt,topsep=1.0pt}
\setlength{\textfloatsep}{0.7\baselineskip plus 0.1\baselineskip minus 0.1\baselineskip}
\setlength{\intextsep}{10pt plus 2pt minus 2pt}

\title{\vspace{-2.5cm}Testing Firm Conduct}
\author{Marco Duarte\thanks{Department of Economics, University of North Carolina at Chapel Hill. Email: \href{mailto:duartema@unc.edu}{duartema@unc.edu}}
 \and Lorenzo Magnolfi\thanks{%
Department of Economics, University of Wisconsin-Madison, Madison. Email: \href{mailto:magnolfi@wisc.edu}{magnolfi@wisc.edu}}
\and Mikkel S{\o}lvsten\thanks{%
Department of Economics and Business Economics, Aarhus University. Email: \href{mailto:miso@econ.au.dk}{miso@econ.au.dk}. I acknowledge financial support from the Danish National and Aarhus University Research Foundations (DNRF Chair \#DNRF154 and AUFF Grant \#AUFF-E-2022-7-3)}
 \and Christopher Sullivan\thanks{%
Department of Economics, University of Wisconsin-Madison.  Email: \href{mailto:cjsullivan@wisc.edu}{cjsullivan@wisc.edu} \newline We thank Steve Berry, Chris Conlon, JF Houde, Sarah Johnston, Aviv Nevo, Alan Sorensen, and seminar participants at Cowles, Drexel, IO$^2$, Mannheim, Midwest IO Fest, Montreal Summer conference in IO,  Princeton, Rice, and Stanford for helpful comments.  We would like to thank IRI for making the data available. All estimates and analysis in
this paper, based on data provided by IRI, are by the authors and not by IRI. A Python package that implements the methods in this paper, \texttt{pyRVtest}, is available on GitHub \citep{dmss_code}. } 
}
\date{December 2023 
    }

\maketitle
\vspace{-.5cm}
\begin{abstract}
\noindent   Evaluating policy in imperfectly competitive markets requires understanding firm behavior. While researchers test conduct via model selection and assessment, we present the advantages of \cite{rv02} (RV) model selection under misspecification. However, degeneracy of RV invalidates inference. With a novel definition of weak instruments for testing, we connect degeneracy to instrument strength,  derive weak instrument properties of RV, and provide a diagnostic for weak instruments by extending the framework of \cite{sy05} to model selection.   We test vertical conduct \citep[][]{v07} using common instrument sets.  Some are weak, providing no power. Strong instruments support manufacturers setting retail prices.   

\vspace{11pt}

\vspace{11pt}

\end{abstract}
\pagebreak
\section{Introduction}\label{sec:intro}

Understanding the impact of policy in imperfectly competitive markets requires models of firm behavior. Additionally, studying firm conduct can be of primary interest in itself. However, the true model is often unknown, prompting researchers to test candidate models. Recent applications of conduct testing include common ownership's competitive effect \citep[][]{bcs20}, labor market monopsony power \citep[][]{rs21}, and the US government's market power in issuing safe assets \citep[][]{ckp22}. \looseness=-1

In an ideal setting, testing firm conduct involves comparing model-implied markups to true markups, which are rarely observed. To address this challenge, \cite{bh14} propose a falsifiable restriction using instruments that co-vary with markups, but are uncorrelated with true cost shocks. In practice, industrial organization (IO) researchers have implemented two types of statistical tests, whose nulls encode the falsifiable restriction: model selection (comparing the relative fit of competing models) and model assessment tests (checking the absolute fit of a given model). This distinction affects inference under misspecification. The model selection test in \cite{rv02} (RV) may conclude for the true model, whereas model assessment tests reject the true model in large samples. 
Despite these advantages, the  RV test can suffer from degeneracy, defined as zero asymptotic variance of the difference in lack of fit between models \citep[see][]{rv02}. Degeneracy invalidates the RV test's asymptotic null distribution. The economic causes and inferential effects of degeneracy remain opaque, and researchers often ignore degeneracy when testing firm conduct. \looseness=-1

In this paper, we show that, because the RV test relies on moment conditions formed with instruments, degeneracy can be understood as a problem of irrelevant instruments. In particular, degeneracy occurs if either the true model's markups are indistinguishable from those of the candidate models or the instruments are uncorrelated with markups. 
To shed light on the inferential consequences of degeneracy, we define a novel  \textit{weak instruments for testing} asymptotic framework adapted from \cite{ss97}.  Under this framework, we show that the asymptotic null distribution of the RV test statistic is skewed and has a non-zero mean.  As skewness declines in the number of instruments while the magnitude of the mean increases,  the resulting size distortions are non-monotone in the number of instruments.  With one instrument or many instruments,  large size distortions are possible. With two to nine instruments, we find that size distortions above 2.5\% are impossible.\looseness=-1

Our characterization of degeneracy allows us to develop a novel diagnostic for weak instruments, which aids researchers in drawing proper conclusions when testing conduct with RV.  In the spirit of \cite{sy05} and \cite{olea13}, our diagnostic uses an effective $F$-statistic.  However, the $F$-statistic is formed from two auxiliary regressions as opposed to a single first stage.  Like \cite{sy05}, we show that instruments can be diagnosed as weak based on worst-case size. With one or many instruments, the proposed $F$-statistic needs to be large to conclude that the instruments are strong for size; we provide the appropriate critical values.\looseness = -1

A distinguishing feature of our weak instruments framework is that power is a salient concern.
In fact, the best-case power of the RV test against either model of conduct is strictly less than one, even in large samples.  The attainable power of the test can differ across the competing models and is lowered by misspecification.  Thus,  diagnosing whether the instruments are strong for power is crucial.   The same $F$-statistic used to detect size distortions is also informative about the best-case power of the RV test.  However, the critical values to compare the $F$-statistic against are different.  For power, the critical values are monotonically declining in the number of instruments, making low power the primary concern with two to nine instruments. In sum, researchers no longer need to assume away degeneracy; instead, they can interpret the results of RV through the lens of our $F$-statistic.\looseness=-1  

Up to here, the discussion presumes researchers precommit to one instrument set. \cite{bh14} show that many sources of exogenous variation exist for testing conduct, although -- reflecting the economics of different models -- some may fail to be relevant in specific applications \citep[][]{mqsw22}.  Pooling all sources of variation into one set of instruments may obscure the degree of misspecification while also diluting instrument strength.   Section \ref{sec:accumulate} discusses how researchers can separately use these sources of variation to test firm conduct without needing to precommit to any single one. We propose a conservative procedure whereby a researcher concludes for a set of models when all strong instruments support them.\looseness=-1

In an empirical application,  we revisit the setting of \cite{v07} and test five models of vertical conduct in the market for yogurt, including models of double marginalization and  two-part tariffs. The application illustrates the empirical relevance of our results for inference on conduct with misspecification and weak instruments. Inspection of the price-cost margins implied by different models only allows us to rule out one model. To obtain sharper results, we then perform model selection with RV. Commonly used sets of instruments are weak for testing as measured by our diagnostic.   When the RV test is implemented with these weak instruments, it has essentially no power. This illustrates the importance of using our diagnostic to assess instrument strength in terms of both size and power when interpreting the results of the RV test. \looseness=-1  

Using our procedure to accumulate evidence from different sources of variation,  we conclude for a model in which manufacturers set retail prices. All strong instruments reject the other models. This application speaks to an important debate over vertical conduct in consumer packaged goods industries.  Several applied papers assume a model of two-part tariffs  where manufacturers set retail prices \citep[e.g.,][]{n01,mw17}.  Our results support this assumption. \looseness=-1 

This paper develops tools relevant to a broad literature seeking to understand firm conduct in the context of structural models of demand and supply. Focusing on articles that pursue a testing approach, collusion is a prominent application \citep[e.g.,][]{p83,s85,b87,glv92,v96,gm98,n01,s20}.  Other important applications include common ownership \citep{bcs20}, vertical conduct \citep[e.g.,][]{v07,bd10,g13, z21, lwy21}, price discrimination \citep{ddf19}, price versus quantity setting \citep{fl95}, and non-profit behavior \citep{dmr20}. Outside of IO, recent applications include labor market conduct  \citep{rs21} and the market power of the US government in issuing safe assets \citep[][]{ckp22}. \looseness=-1


This paper is also related to econometric work on the testing of non-nested hypotheses \citep[e.g., ][]{pw01}.  We build on the insights of the econometrics literature that performs inference under misspecification and highlights the importance of model selection procedures \citep[e.g.,][]{w82,v89, hi03, mo12,ls20}. Two recent contributions, \cite{s15} and \cite{sw17}, modify the likelihood based test in \cite{v89} to correct size distortions under degeneracy.  In our GMM setting, we show that power, not size, is the salient concern. Furthermore, by connecting degeneracy to instrument strength,   our work is related to the econometrics literature on inference under weak instruments \citep[surveyed in][]{ass19}. Contemporaneous work by \cite{bcs20} shares our goal of enhancing inference on firm conduct. They explore the complementary question of which functional form the researcher should use in constructing the instruments, and propose running the RV test with a single instrument formed by pooling all exogenous variation with a random forest.\looseness=-1

The paper proceeds as follows. Section \ref{sec:themodel} describes the environment: a general model of firm conduct. Section \ref{sec:nevo} formalizes our notion of falsifiability when true markups are unobserved. Section \ref{sec:hypo_formulation} explores the effect of hypothesis formulation on inference, contrasting  model selection and assessment approaches under misspecification. Section \ref{sec:instruments} connects  degeneracy of RV to instrument strength, characterizes the effect of weak instruments on inference, and introduces a diagnostic for weak instruments to report alongside the RV test. Section \ref{sec:accumulate} provides a procedure to accumulate evidence across different sets of instruments. Section \ref{sec:empiricaleg} develops our empirical application: testing models of vertical conduct in the retail market for yogurt. Section \ref{sec:conclusion} concludes. Proofs are found in Appendix \ref{sect:Proofs}. \looseness=-1

\section{Testing Environment}\label{sec:themodel}

We consider testing models of firm conduct using data on a set of products $\mathcal{J}$ offered by firms across a set of markets $\mathcal{T}$. For each product and market combination $(j,t)$, the researcher observes the price $\boldsymbol{p}_{jt}$, market share $\boldsymbol{s}_{jt}$, a vector of product characteristics $\boldsymbol{x}_{jt}$ that affects demand, and a vector of cost shifters $\textbf{w}_{jt}$ that affects the product's marginal cost. For any variable $\boldsymbol{y}_{jt}$, denote  $\boldsymbol{y}_t$ as the vector of values in market $t$. We assume that, for all markets $t$, the demand system is $\boldsymbol{s}_{t} = \mathbcal{s}\big(\boldsymbol{p}_t,\boldsymbol{x}_t,\boldsymbol\xi_t,\theta^D_0\big)$, where $\boldsymbol{\xi}_t$ is a vector of unobserved product characteristics, and $\theta^D_0$ is the true vector of demand parameters.\looseness=-1

The equilibrium in market $t$ is characterized by a system of first order conditions arising from the firms' profit maximization problems: 
\begin{align}\label{eq:model}
    \boldsymbol{p}_t &= \boldsymbol{\Delta}_{0t} + \boldsymbol{c}_{0t},
\end{align}
where $\boldsymbol{\Delta}_{0t} = \boldsymbol\Delta_0\big(\boldsymbol{p}_t,\boldsymbol{s}_t,\theta^D_0\big)$ is the true vector of markups in market $t$ and  $\boldsymbol{c}_{0t}$ is the true vector of marginal costs. Following \cite{bh14} we assume cost has the separable form $    \boldsymbol{c}_{0jt} = \bar{\boldsymbol{c}}(\boldsymbol{q}_{jt},\textbf{w}_{jt}) + \omega_{0jt},$ where $\boldsymbol{q}_{jt}$ is quantity and $\omega_{0jt}$ is an unobserved shock. To speak directly to the leading case in the applied literature, we assume marginal costs are constant in  $\boldsymbol{q}_{jt}$, and maintain $E[\bar{\boldsymbol{c}}(\textbf{w}_{jt}) \omega_{0jt}]=0$. However, the results in the paper apply to the important case of non-constant marginal cost, as shown in Appendix \ref{sect:AltCost}.\looseness=-1

The researcher can obtain an estimate $\hat \theta^D$ of the demand parameters, formulate alternative models of conduct, and then compute estimates of markups $\hat {\boldsymbol{\Delta}}_{mt} = {\boldsymbol{\Delta}}_{m}\big(\boldsymbol p_t, \boldsymbol s_t, \hat \theta^D\big)$ under each model $m$ with the estimated demand parameters. When discussing large sample results, we abstract away from the demand estimation step and treat $\boldsymbol{\Delta}_{mt} = {\boldsymbol{\Delta}}_{m}\big(\boldsymbol p_t, \boldsymbol s_t, \theta^D\big)$ as data, where $\theta^D =\plim \hat \theta^D $.\footnote{When demand is estimated in a preliminary step, the variance of the test statistics presented in Section \ref{sec:hypo_formulation} needs to be adjusted. The necessary adjustments are in Appendix \ref{sect:TwoStep}.} We focus on the case of two candidate models, $m=1,2$, and defer a discussion of more than two models to Section \ref{sec:accumulate}. To simplify notation, we replace the $jt$ index with $i$ for a generic observation.  We suppress the $i$ index when referring to a vector or matrix that stacks all $n$ observations in the sample.\footnote{We consider asymptotic analysis with $n$ diverging, thereby allowing for large $\mathcal J$ and/or $\mathcal T$.}  Our framework is general, and depending on the choice of  $\boldsymbol \Delta_1$ and $\boldsymbol \Delta_2$  allows us to test many models of conduct found in the literature. Canonical examples include the nature of vertical relationships, whether firms compete in prices or quantities, collusion, intra-firm internalization, common ownership and nonprofit conduct.\footnote{In important applications  \citep[e.g.,][]{mw17,bcs20}, markups are functions of a parameter $\kappa$  ($\boldsymbol{\Delta}_m=\boldsymbol{\Delta}(\kappa_m)$). Researchers may investigate conduct by either estimating $\kappa$ or testing. \cite{ms21} provide a comparison of testing and estimation approaches in this setting.}

Throughout the paper, we consider the possibility that the researcher may misspecify demand or cost, or specify two models of conduct (e.g., Nash price setting or collusion) which do not match the truth (e.g., Nash quantity setting).  In these cases, $\boldsymbol{\Delta}_0$ does not coincide with the markups implied by either candidate model.  We show that misspecification along any of these dimensions has consequences for testing, contrasting it to the case where $\boldsymbol{\Delta}_0 = \boldsymbol{\Delta}_1$.

Another important consideration for testing conduct is whether markups for the true model $ \boldsymbol{\Delta}_{0}$ are observed. In an ideal testing environment, the researcher observes not only markups implied by the two candidate models, but also the true markups $\boldsymbol{\Delta}_{0}$ (or equivalently marginal costs).  However, $\boldsymbol\Delta_0$ is unobserved in most empirical applications, and we focus on this case in what follows. Testing models thus requires instruments for the endogenous markups $\boldsymbol{\Delta}_1$ and $\boldsymbol{\Delta}_2$.   We maintain that the researcher constructs instruments $\boldsymbol z$, such that the following exclusion restriction holds:

\begin{assumption}\label{as:momcond}
    $\boldsymbol{z}_i$ is a vector of $d_z$ excluded instruments, so that $E[ \boldsymbol{z}_i \omega_{0i}]=0$.
\end{assumption}

\noindent 
This assumption requires that the instruments are {exogenous for testing}, and therefore uncorrelated with the unobserved cost shifters for the true model. 
Any source of exogenous variation which moves the residual marginal revenue curve for at least one firm can serve as a valid instrument \citep{bh14}.  These include variation in the set of rival firms and rival products, own and rival product characteristics, rival cost, and market demographics. While for now we do not distinguish different sets of instruments, in Section \ref{sec:accumulate}, we discuss how researchers can separately use these sources of variation to test firm conduct, without needing to precommit to any single one.     

The following assumption introduces regularity conditions that are maintained throughout the paper and used to derive the properties of the tests discussed in Section \ref{sec:hypo_formulation}.

\begin{assumption}\label{as:regcond}
    (i) $\{\boldsymbol \Delta_{0i},\boldsymbol \Delta_{1i},\boldsymbol \Delta_{2i},\boldsymbol z_i, {\textbf{w}}_i, \omega_{0i}\}_{i=1}^n$ are jointly iid;
    (ii) $E\big[(\boldsymbol \Delta_{1i}-\boldsymbol \Delta_{2i})^2\big]$ is positive and $E\big[(\boldsymbol z_i',{\textbf{w}}_i')'(\boldsymbol z_i',{\textbf{w}}_i')\big]$ is positive definite;
    (iii) the entries of $\boldsymbol\Delta_{0i}$, $\boldsymbol\Delta_{1i}$, $\boldsymbol \Delta_{2i}$, $\boldsymbol z_i$, ${\textbf{w}}_i$, and $ \omega_{0i}$ have finite fourth moments.
\end{assumption}

\noindent 
Part (i) is a standard assumption for cross-sectional data. Correlated market-level shocks to cost and demand primitives will typically lead to markups that are correlated within a market, but iid across markets. Extensions to
such settings are straightforward and discussed in Appendix \ref{sect:TwoStep}.  Part (ii) excludes cases where the two competing models of conduct map cost and demand primitives to identical markups and cases where the instruments $\boldsymbol z$ are linearly dependent with the cost shifters $\textbf{w}$. Part (iii) is a standard regularity condition allowing us to establish asymptotic approximations as $n \rightarrow \infty$.\looseness=-1

We further maintain that marginal costs are a linear function of observable cost shifters $\textbf{w}$ and $\omega_0$, so that $\boldsymbol{c}_{0} = \textbf{w} \tau + \omega_{0}$, where $\tau$ is defined by the orthogonality condition $E[\textbf{w}_i  \omega_{0i}]=0$. 
This restriction allows us to eliminate the cost shifters $\textbf{w}$ from the model, which is akin to the thought experiment of keeping the observable part of marginal cost constant across markets and products. For any variable $\boldsymbol y$, we therefore define the residualized variable $y = \boldsymbol y - \textbf{w}E[\textbf{w}'\textbf{w}]^{-1}E[\textbf{w}'\boldsymbol y]$ and its sample analog as $\hat y = \boldsymbol y - \textbf{w}(\textbf{w}'\textbf{w})^{-1}\textbf{w}'\boldsymbol y$.

The following section discusses the essential role of the instruments $\boldsymbol z$ in distinguishing between different models of conduct. For this discussion, a key role is played by the part of residualized markups $\Delta_m$ that are predicted by $z$:\looseness=-1 
\begin{align}\label{eq:predMarkups}
    \Delta^z_m &= z \Gamma_m, && \text{where } \Gamma_m = E\!\left[z'z\right]^{-1}E\!\left[z'\Delta_{m}\right]
\end{align}
and its sample analog $\hat \Delta^z_m = \hat z \hat \Gamma_m$ where $\hat \Gamma_m = (\hat z'\hat z)^{-1} \hat z'\hat \Delta_{m}$. \cite{bcs20} highlight the importance of modeling non-linearities both in the cost function and the predicted markups. Our linearity assumptions are not restrictive insofar as $\textbf{w}$ and $\boldsymbol{z}$ are constructed flexibly from exogenous variables in the data. When stating theoretical results, the distinction between population and sample counterparts matters, but for building intuition there is no need to separate the two. We refer to $\Delta^z_m$ as \textit{predicted markups} for model $m$. \looseness=-1 

\section{Falsifiability of Models of Conduct}\label{sec:nevo}

We begin by reexamining the conditions under which models of conduct are falsified. Models are characterized by their markups $\Delta_m$.  In the ideal setting where true markups are observed, a model is falsified if the markups implied by model $m$ differ from the true markups with positive probability, or $E\big[(\Delta_{0i} - \Delta_{mi})^2\big] \neq 0$.  Instead, when true markups are unobserved, researchers need to rely on a set of excluded instruments when attempting to distinguish a wrong model from the true one. \looseness=-1

In our setting,  instruments provide a benchmark for distinguishing models through the moment condition in Assumption \ref{as:momcond}, $E[z_i\omega_{0i}] = 0$.  For each model $m$, the analog of this condition is $E[z_i(p_i-\Delta_{mi})]= 0$, where $p_i-\Delta_{mi}$ is the  residualized marginal revenue under model $m$.  Thus, to falsify model $m$, the correlation between the instruments and the residualized marginal revenue implied by model $m$ must be different from zero. This is in line with the result in \cite{bh14} that valid instruments need to alter the  marginal revenue faced by at least one firm to distinguish conduct. 

However, it is not apparent from the  restriction $E[z_i(p_i-\Delta_{mi})]= 0$ how instruments distinguish model $m$ from the truth based on their key economic feature, markups.\footnote{For a thorough discussion of the economic determinants of falsification, see \cite{mqsw22}.} Notice that, under Assumption \ref{as:momcond}, the covariance between residualized price and the instrument is equal to the covariance between the  residualized unobserved true markup and the instrument, or $E[z_ip_i] = E[z_i\Delta_{0i}]$. This equation highlights the role of instruments: after we control for $\textbf{w}$ by residualizing prices, the instruments recover from residualized prices a feature of the unobserved true markups. Thus testing relies on the comparison between $E[z_i\Delta_{0i}]$ and $E[z_i\Delta_{mi}]$. If we rescale the moments by the variance in $z$, we can restate the falsifiable restriction for model $m$ in terms of the mean squared error (MSE) in predicted markups, which we formally connect to \cite{bh14} in the following lemma.\footnote{Our environment is an example of Case 2 discussed in Section 6 of \cite{bh14}.} \looseness=-1

\begin{lemma}\label{as:falsify}
Under Assumptions \ref{as:momcond} and \ref{as:regcond}, the falsifiable restriction in Equation (28) of \cite{bh14} implies 
\begin{align}\label{eq:testability}
    E\!\left[\!\left(\Delta^z_{0i} - \Delta^z_{mi}\right)^2\right] = 0.
\end{align}
If Equation \eqref{eq:testability} is violated, we say  model $m$ is falsified by the instruments $z$.
\end{lemma}

The lemma establishes an analog to testing with observed true markups. When $\Delta_0$ is observed,  a model $m$ is falsified if its markups differ from the truth with positive probability. Here, $\Delta_0$ is unobserved and falsifying model $m$ requires the markups predicted by the instruments to differ, i.e., $\Delta_{0i}^z \neq \Delta_{mi}^z$ with positive probability. Therefore testing when markups are unobserved still relies on differences in economic features between model $m$ and the true model, insofar as these differences result in different correlations with the instruments. Thus, falsifiability in our environment is a joint feature of a pair of models and a set of instruments. \looseness=-1

Moreover, a consequence of the lemma is that the sources of exogenous variation discussed in \cite{bh14} permit testing in our context.  These sources of variation include demand rotators, own and rival product characteristics, rival cost shifters, and market demographics.\footnote{We discuss how to separately use these sources of variation in Section 6.}  In addition to being exogenous, Lemma \ref{as:falsify} shows that instruments need to be  relevant for testing in order to falsify a wrong model of conduct.  In particular, a model $m$ can be falsified only  if the instruments are correlated with, and therefore generate non-zero predicted markups for, at least one of $\Delta_0$ and $\Delta_m$.  We illustrate this point in an example.\looseness=-1

\vspace{0.5em}
\noindent \emph{Example 1:} Consider an example, in the spirit of \cite{b82},  of distinguishing monopoly and perfect competition.\footnote{\cite{b82} also allows for non-constant marginal cost, which we consider in Appendix \ref{sect:AltCost}.\looseness=-1 }  Notice that, under perfect competition, both markups and predicted markups  are zero.  Thus, falsifying perfect competition when data are generated under monopoly (or vice versa) requires that the instruments generate non-zero monopoly predicted markups.  This occurs whenever variation in the instruments induces variation in the monopoly markups.  Given that these markups are a function of market shares and prices,  the sources of variation in \cite{bh14} typically suffice.  
\vspace{0.5em}

While Equation \eqref{eq:testability} is a falsifiable restriction in the population, performing valid inference on conduct in a finite sample requires two steps. First, we need to encode the falsifiable restriction into  hypotheses. Second, we need strong instruments to falsify the wrong model. We turn to these problems in the next two sections.\looseness=-1

\section{Hypothesis Formulation for Testing Conduct}\label{sec:hypo_formulation}

To test amongst alternative models of firm conduct in a finite sample, researchers need to choose a testing procedure, four of which have been used in the IO literature.\footnote{E.g., \cite{bcs20} use an RV test, \cite{bd19} use an Anderson-Rubin test to supplement an estimation exercise, \cite{mw17} use an estimation based test, and \cite{v07} uses a Cox test.  All these procedures accommodate instruments and do not require specifying the full likelihood as was done in earlier literature \citep[e.g., ][]{b87,glv92}. } As discussed below, these can be classified as model assessment or model selection tests based on how each formalizes the null hypothesis. In this section, we present the standard formulation of RV, a model selection test, and the \cite{ar49} test (AR), a model assessment test. We focus on AR as its properties in our environment are representative of the three model assessment tests used in IO to test conduct.\footnote{In Appendix \ref{sect:EBCox}, we show that the other model assessment procedures have similar properties to AR.}  We relate the hypotheses of these tests to our falsifiable restriction in Lemma \ref{as:falsify}.  Then, we contrast the statistical properties of RV and AR, allowing us to characterize the performance of RV under misspecification.\looseness=-1

\subsection{Definition of the Tests}\label{sec:thetests}

\noindent
\textbf{Rivers-Vuong Test (RV):}
A prominent approach to testing non-nested hypotheses was developed in \cite{v89} and then extended to models defined by moment conditions in \cite{rv02}. The null hypothesis for the test is that the two competing models of conduct have the same fit,
\begin{align}
    H_0^\text{RV}: \,\,  {Q}_1 = {Q}_2,
\end{align}
where ${Q}_m$ is a population measure for lack of fit in model $m$.
Relative to this null, we define two alternative hypotheses corresponding to cases of better fit of one of the two models:
\begin{align}
    H_1^\text{RV} \, : \, {Q}_1<{Q}_2 \qquad \text{and} \qquad H_2^\text{RV}\, : \, {Q}_2 < {Q}_1.
\end{align} 
With this formulation of the null and alternative hypotheses, the statistical problem is to determine which of the two models has the best fit, or equivalently, the smallest lack of fit.

We define lack of fit via a GMM objective function,  a standard choice for models with endogeneity. Thus, ${Q}_m  =  {g}_m' W {g}_m$ where ${g}_m= E[z_i(p_i-\Delta_{mi})]$ and $ W=E[z_i z_i']^{-1}$ is a positive definite weight matrix.\footnote{This weight matrix allows us to interpret ${Q}_m$ in terms of Euclidean distance between predicted markups for model $m$ and the truth, directly implementing the MSE of predicted markups in Lemma \ref{as:falsify}.\looseness=-1}   The sample analog of ${Q}_m$ is $\hat Q_m=\hat g_m'\hat W \hat g_m$ where $\hat g_m = n^{-1} \hat z'(\hat p - \hat \Delta_m)$ and $\hat W = n (\hat z'\hat z)^{-1}$.

For the GMM measure of fit, the RV test statistic is then
\begin{align}\label{RVt}
T^\text{RV} = \frac{\sqrt{n}(\hat Q_1 - \hat Q_2)}{\hat \sigma_\text{RV}} , 
\end{align}
where $\hat\sigma^2_\text{RV}$ is an estimator for the asymptotic variance of  the scaled difference in the measures of fit appearing in the numerator of the test statistic. We denote this asymptotic variance by $\sigma^2_\text{RV}$. Throughout, we let $\hat \sigma^2_\text{RV}$ be a delta-method variance estimator that takes into account the randomness in both  $\hat W$ and $\hat g_m$. Specifically, this variance estimator takes the form
\begin{align}
    \hat \sigma^2_\text{RV} = 4\!\left[ \hat g_1'\hat W^{1/2} \hat V_{11}^\text{RV} \hat W^{1/2} \hat g_1 + \hat g_2'\hat W^{1/2} \hat V_{22}^\text{RV} \hat W^{1/2} \hat g_2- 2\hat g_1'\hat W^{1/2} \hat V_{12}^\text{RV} \hat W^{1/2} \hat g_2  \right]
\end{align}
where $\hat V_{\ell k}^\text{RV}$ is an estimator of the covariance between $\sqrt{n}\hat W^{1/2}\hat g_\ell$ and  $\sqrt{n}\hat W^{1/2}\hat g_k$. Our proposed $\hat V^\text{RV}_{\ell k}$ is given by $\hat V^\text{RV}_{\ell k} = n^{-1} \sum_{i=1}^n \hat \psi_{\ell i} \hat \psi_{ki}'$ where
\begin{align}
    \hat \psi_{mi} = \hat W^{1/2} \!\left( \hat z_i \big(\hat p_i - \hat \Delta_{mi}\big) - \hat g_m \right)- \tfrac{1}{2} \hat W^{3/4} \!\left( \hat z_i \hat z_i' -\hat W^{-1}\right)\!\hat W^{3/4} \hat g_m.
\end{align}
This variance estimator is transparent and easy to implement.  Adjustments to $\hat\psi_{mi}$ and/or $\hat V^\text{RV}_{\ell k}$ can also accommodate initial demand estimation and clustering; see Appendix \ref{sect:TwoStep}.\footnote{An alternative way of estimating this variance would be by bootstrapping, which can be costly especially when demand has to be re-estimated in each bootstrap sample.}\looseness=-1

The test statistic $T^\text{RV}$ is standard normal under the null as long as $\sigma^2_\text{RV}>0$. The RV test therefore rejects the null of equal fit at level $\alpha \in (0,1)$ whenever $\abs{ T^\text{RV} }$ exceeds the $(1-\alpha/2)$-th quantile of a standard normal distribution. If instead $\sigma^2_\text{RV} = 0$, the RV test is said to be degenerate. In the rest of this section, we maintain non-degeneracy.

\begin{manualtheorem}{ND}\label{as:nodegen}
The RV test is not degenerate, i.e., $\sigma^2_\text{RV} > 0$.
\end{manualtheorem}

While Assumption \ref{as:nodegen} is often maintained in practice, severe inferential problems may occur when $\sigma^2_\text{RV} = 0$. These problems include large size distortions and little to no power throughout the parameter space. Thus, it is essential to understand degeneracy and diagnose the inferential problems it can cause. We return to these issues in Section \ref{sec:instruments}.

\vspace{0.5em}
\noindent
\textbf{Anderson-Rubin Test (AR):} 
In this approach, the researcher writes down the following equation for each of the two models $m$:
\begin{align}\label{ref:eq_AR}
    p -\Delta_m  =  z\pi_m + e_m
\end{align}
where $\pi_m$ is defined by the orthogonality condition $E[z e_m]=0$. She then performs the test of the null hypothesis that $\pi_m=0$ with a Wald test. This procedure is similar to an \cite{ar49} testing procedure. For this reason, we refer to this procedure as AR. Formally, for each model $m$, we define the null and alternative hypotheses:
\begin{align}
    H_{0,m}^\text{AR} \,: \, \pi_m =0 \qquad \text{and} \qquad H_{A,m}^\text{AR} \, : \, \pi_m \neq 0.
\end{align} 
For the true model, $\pi_m$ is equal to zero since the dependent variable in Equation \eqref{ref:eq_AR} is equal to $\omega_0$ which is uncorrelated with $z$ under Assumption \ref{as:momcond}. \looseness=-1

We define the AR test statistic for model $m$ as:
\begin{align}\label{eq:ARtest_Stat}
    T^\text{AR}_m &= n\hat \pi_m' \big( \hat V^\text{AR}_{mm}  \big)^{-1} \hat \pi_m 
\end{align}
where $\hat \pi_m$ is the OLS estimator of $\pi_m$ in Equation  \eqref{ref:eq_AR} and $\hat V^\text{AR}_{mm}$ is White's heteroskedasticity-robust variance estimator. This variance estimator is $\hat V^\text{AR}_{\ell k} = n^{-1} \sum_{i=1}^n \hat \phi_{\ell i} \hat \phi_{ki}'$ where $\hat \phi_{mi} = \hat W \hat z_i\big( \hat p_i - \hat \Delta_{mi} - \hat z_i'\hat \pi_m\big)$. Under the null corresponding to model $m$, the large sample distribution of the test statistic $T^\text{AR}_m$ is a (central) $\chi^2_{d_z}$ distribution and the AR test rejects the corresponding null at level $\alpha$ when $T^\text{AR}_m$ exceeds the $(1-\alpha)$-th quantile of this distribution.\looseness=-1

\subsection{Hypotheses Formulation, Falsifiability, and Misspecification}\label{sec:hypo}

We now show that the null hypotheses of both tests can be reexpressed in terms of our falsifiable restriction in Lemma \ref{as:falsify}.\looseness=-1

\begin{proposition}\label{prop:nulls}
Suppose that Assumptions \ref{as:momcond} and \ref{as:regcond} are satisfied.  Then 
    \begin{enumerate}
        
        \item[(i)] the null hypothesis $H_0^\emph{RV}$ holds if and only if  $E\big[(\Delta^z_{0i}-\Delta^z_{1i})^2\big] = E\big[(\Delta^z_{0i}-\Delta^z_{2i})^2\big]$,
        
        \item[(ii)] the null hypothesis $H_{0,m}^\emph{AR}$ holds if and only if $E\big[(\Delta^z_{0i}-\Delta^z_{mi})^2\big] = 0$.
        
    \end{enumerate}
\end{proposition}

The formulation of the hypotheses in Proposition \ref{prop:nulls} shows that AR and RV implement Equation \eqref{eq:testability} through their null hypotheses, but they do so in distinct ways.\footnote{While it may seem puzzling that the hypotheses depend on a feature of the unobservable $\Delta_0,$ recall that $\Delta_0^z$ is identified by the observable covariance between $p$ and $z$, which have both been residualized with respect to the observed cost shifters $\textbf{w}$.} The null of AR asserts that the MSE of predicted markups for model $m$ is zero, so the test separately evaluates whether each model is falsified by the instruments. We show in Appendix \ref{sect:EBCox} that other \textit{model assessment} procedures used in the IO literature to test conduct share the same null as AR, and may reject both models if they both have poor absolute fit. Instead, the RV test pursues a \textit{model selection} approach by positing a null under which the two models have equal fit, as measured by the MSE of predicted markups. If the RV test rejects, it will never reject both models, but only the one whose predicted markups are farther from the true predicted markups.  

Importantly, the two formulations of the null have implications for inference when markups are misspecified.\footnote{Misspecification of $\Delta^z_m$ may arise by the researcher misspecifying cost.  We show in Appendix \ref{sect:MCMisp} that this case has similar consequences.} To fully understand the role of misspecification on inference of conduct, we would like to contrast the performance of AR and RV in finite samples. However, it is not feasible to characterize the exact finite sample distribution of AR and RV under our maintained assumptions. Instead, we can approximate the finite sample distribution of each test by considering local misspecification, i.e., a sequence of candidate models that converge to the null space at an appropriate rate. As model assessment and model selection procedures have different nulls, we define distinct local alternatives for RV and AR based on $\Gamma_m$. For model assessment, local misspecification is characterized in terms of the absolute degrees of misspecification for each model:\looseness=-1
\begin{align}\label{eq:local assessment}
    \Gamma_0 -  \Gamma_m &= {q_m}/{\sqrt{n}} & \text{ for } m \in \{1,2\}
\end{align}
By contrast, local alternatives for model selection are in terms of the relative degree of misspecification between the two models:
\begin{align}\label{eq:local selection}
    (\Gamma_0 -  \Gamma_1) - (\Gamma_0 -  \Gamma_2) &= {q}/{\sqrt{n}}.
\end{align}

Under the local alternatives in Equations \eqref{eq:local assessment} and \eqref{eq:local selection}, we approximate the finite sample distribution of AR and RV with misspecification in the following proposition.  To facilitate a characterization in terms of predicted markups, we define stable versions of predicted markups under either of the two local alternatives considered:  $\Delta_{mi}^{\text{RV},z} = n^{1/4}\Delta_{mi}^{z}$ and $\Delta_{mi}^{\text{AR},z} = n^{1/2}\Delta_{mi}^{z}$. We also introduce an assumption of homoskedastic errors, which in this section serves to simplify the distribution of the AR statistic:\looseness=-1

\begin{assumption}\label{as:homsked}
    The error term in Equation \eqref{ref:eq_AR}, $e_{mi}$, is homoskedastic, i.e., $E[e_{m i}^2 z_iz_i']=\sigma_{m}^2 E[z_iz_i']$ with $\sigma_{m}^2>0$ for $m \in \{1,2\}$ and $E[e_{1 i} e_{2 i} z_iz_i']=\sigma_{12} E[z_iz_i']$ with $\sigma_{12}^2 <\sigma_{1}^2 \sigma_{2}^2$.
\end{assumption}

\noindent 
The intuition developed in this section does not otherwise rely on Assumption \ref{as:homsked}.\looseness=-1

\begin{proposition}\label{prop:distrib_TE2}
Suppose that Assumptions \ref{as:momcond}--\ref{as:homsked} and \ref{as:nodegen} are satisfied.
Then
\begin{align}
    (i) && 
    T^{\emph{RV}} & \xrightarrow{d} 
    N\!\left( \!\frac{E\big[(\Delta^{\emph{RV},z}_{0i}-\Delta^{\emph{RV},z}_{1i}) ^2\big] - E\big[(\Delta^{\emph{RV},z}_{0i}-\Delta^{\emph{RV},z}_{2i}) ^2\big]}{\sigma_\text{\emph{RV}}},1\!\right)\! && 
    \text{under } \eqref{eq:local selection}, 
    \\
    (ii) && 
    T^\text{\emph{AR}}_m &\xrightarrow{d} \chi^2_{d_z}\!\left(\frac{E\big[(\Delta^{\emph{AR},z}_{0i}-\Delta^{\emph{AR},z}_{mi}) ^2\big] }{\sigma^2_{m}} \right)\! && 
    \text{under } \eqref{eq:local assessment},
\end{align}
where $\chi^2_{df}(nc)$ denotes a non-central $\chi^2$ distribution with $df$ degrees of freedom and non-centrality $nc$.
\end{proposition}

From Proposition \ref{prop:distrib_TE2}, both test statistics follow a non-central distribution. However, the non-centrality term differs for the two tests because of the formulation of their null hypotheses.  For AR, the non-centrality for model $m$ is the ratio of the MSE of predicted markups to the noise given by $\sigma^2_{m}$.  Alternatively, the noncentrality term for RV depends on the ratio of the difference in MSE for the two models to the noise.  
Intuitively, whether the AR or RV test conclude for a model of conduct depends not only on the MSE of predicted markups but also on the noise in the testing environments. If the degree of misspecification is low, the probability RV concludes in favor of the true model increases as the noise decreases. Instead, AR only concludes in favor of the true model with sufficient noise; as the noise declines, AR is increasingly likely to reject both models. 

\vspace{0.5em}
\noindent \emph{Example 1 -  continued:}  Consider again the case of distinguishing perfect competition from the true model of monopoly, now using market demographics as instruments. Suppose that the researcher misspecifies the demand model, for instance by estimating a mixed logit model that omits a significant interaction between demographics and product characteristics. Let $\Delta_0$ be monopoly markups with the true demand system, and $\Delta_1$ and $\Delta_2$ be the monopoly and perfect competition markups, respectively, with the misspecified demand system.  Thus, $\Delta_2$ and $\Delta^z_2$ are both zero. Because substitution patterns are misspecified, the degree to which market demographics affect  $\Delta_0$ and  $\Delta_1$ is different. In sufficiently small samples, such as when the noise is high, neither test rejects the null - in particular, AR does not reject either the monopoly or perfect competition models. Instead, in larger samples, and therefore less noise, AR rejects both models. As long as the MSE of $\Delta_1^z$ is smaller than the variance of $\Delta^z_0$, or $E\big[(\Delta^z_{0i}-\Delta^z_{1i})^2\big]< E\big[(\Delta^z_{0i})^2\big]$, RV concludes in favor of monopoly in large enough samples. 
\vspace{0.5em}


An analogy may be useful to summarize our discussion in this section. Model selection compares the relative fit of two candidate models and asks whether a ``preponderance of the evidence'' suggests that one model fits better than the other. Meanwhile, model assessment uses a higher standard of evidence, asking whether a model can be falsified ``beyond any reasonable doubt.'' While we may want to be able to conclude in favor of a model of conduct beyond any reasonable doubt,  this is not a realistic goal in the presence of misspecification. If we lower the evidentiary standard, we can still learn about the true nature of firm conduct. Hence, in the next section we focus on the RV test.  However, to this point we have assumed $\sigma^2_\text{RV} > 0$ and thereby assumed away degeneracy.  We address this threat to inference with the RV test in the next section.

\section{Degeneracy of RV and Weak Instruments}\label{sec:instruments}

Having established the desirable properties of RV under misspecification, we now revisit Assumption \ref{as:nodegen}. First, we connect degeneracy to our falsifiable restriction in Lemma \ref{as:falsify} and show that maintaining Assumption \ref{as:nodegen} is equivalent to ex ante imposing that at least one of the models is falsified by the instruments. To explore the consequences of such an assumption, we define a novel  weak instruments for testing asymptotic framework adapted from \cite{ss97} and for which degeneracy occurs.  We use the weak instrument asymptotics to  show that  degeneracy can cause size distortions and low power in finite samples.  To help researchers interpret the frequency with which the RV test makes errors, we propose a diagnostic in the spirit of \cite{sy05}. This proposed diagnostic is a scaled $F$-statistic computed from two first stage regressions and researchers can use it to gauge the extent to which inferential problems are a concern.\looseness=-1

\subsection{Degeneracy and Falsifiability}\label{sec:degeneracy}

We first characterize when the RV test is degenerate in our setting. Since $\sigma^2_\text{RV}$ is the asymptotic variance of $\sqrt{n}(\hat Q_1-\hat Q_2)$, it follows that Assumption \ref{as:nodegen} fails to be satisfied whenever $\hat Q_1-\hat Q_2=o_p\big( 1/\sqrt{n} \big)$ \citep[see also][]{rv02}. In the following proposition, we reinterpret this condition through the lens of our falsifiable restriction.

\begin{proposition}\label{prop:degen}
Suppose Assumptions \ref{as:momcond}--\ref{as:homsked} hold. Then $\sigma^2_\emph{RV} = 0$ if and only if $E\big[(\Delta^z_{0i}-\Delta^z_{mi})^2\big]=0$ for $m=1$ and $m=2$. 
\end{proposition}

\noindent 
The proposition shows that when $\sigma^2_\text{RV}=0$, neither model is falsified by the instruments.  Thus, Assumption \ref{as:nodegen} is equivalent to assuming the falsifiable restriction in Equation \eqref{eq:testability} is violated for at least one model.

Such a characterization  permits us to better understand degeneracy. Consider two extreme cases where instruments are weak: (i) the instruments are uncorrelated with $\Delta_0$, $\Delta_1$, and $\Delta_2$ such that $z$ is irrelevant for testing of either model, and (ii) models 0, 1, and 2 imply similar markups such that $\Delta_1$ and $\Delta_2$ overlap with $\Delta_0$.  Much of the econometrics literature focuses on (ii) as it considers degeneracy in the maximum likelihood framework  of \cite{v89}.  As RV generalizes the \cite{v89} test to a GMM framework, degeneracy is a broader problem that encompasses instrument strength. We illustrate these ideas in the following two examples that correspond respectively to cases (i) and (ii).\looseness=-1

\vspace{0.5em}
\noindent \textit{Example 2:}  Consider an industry where firms compete across many local markets, but charge uniform prices across all markets. Suppose a researcher wants to distinguish a model of uniform Nash price setting ($m=1$) and a model of uniform monopoly pricing ($m=2$). Let $m=1$ be the true model and assume demand and cost are correctly specified so that $\Delta_0 = \Delta_1$. The researcher forms instruments from local variation in rival cost shifters. If the number of markets is large, the contribution of any one market to the firm-wide pricing decision is negligible.  Thus, the local variation leveraged by the instruments becomes weakly correlated with $\Delta_0$, $\Delta_1$, and $\Delta_2$, resulting in degeneracy.\looseness=-1 
\vspace{0.5em}

\noindent \textit{Example 3:}\footnote{We thank an anonymous referee for suggesting this example.}   Consider three models of simple ``rule of thumb'' pricing, where markups are a fixed fraction of cost. Suppose that the true model implies markups $\boldsymbol \Delta_0 = \boldsymbol c_0,$ and models 1 and 2 correspond to $\boldsymbol \Delta_1 = 0.5 \boldsymbol c_0$ and $\boldsymbol \Delta_2 = 2 \boldsymbol c_0$ respectively.  Given that $\boldsymbol c_0 = \textbf{w}\tau + \omega_0$, the residualized markups are $\Delta_0 = \omega_0$, $\Delta_1 = 0.5\omega_0$ and $\Delta_2 = 2\omega_0$. As the instruments are uncorrelated with $\omega_0$, they are also uncorrelated with   the residualized markups for all three models, and predicted markups are therefore zero. From the perspective of the instruments, both model 1 and model 2 overlap with the true model, and degeneracy obtains for any choice of $z$ satisfying Assumption \ref{as:momcond}. \looseness=-1

\vspace{0.5em}

As shown in the examples, degeneracy can occur in standard economic environments.  It is therefore important to understand the consequences of violating Assumption \ref{as:nodegen}.  To do so, we connect degeneracy to the formulation of the null hypothesis of RV.  From Proposition \ref{prop:nulls}, degeneracy occurs as a special case of the null of RV.  Intuitively, when degeneracy occurs, there is not enough information to falsify either model in the population.  Thus, both models have perfect fit.    Figure \ref{fig:degen_null} illustrates this point by representing both the null space and the space of degeneracy in the coordinate system of MSE of predicted markups $\big(E\big[(\Delta^z_{0i}-\Delta^z_{1i})^2\big],E\big[(\Delta^z_{0i}-\Delta^z_{2i})^2\big]\big)$. While the null hypothesis of RV is satisfied along the full 45-degree line, degeneracy only occurs at the origin.\footnote{If $\Delta_0 = \Delta_1$, the graph shrinks to the $y$-axis and degeneracy arises whenever the null of RV is satisfied. This special case is in line with \cite{hp11}, who note RV is degenerate if both models are true.}\looseness=-1

\begin{figure}[ht]
    \centering
    \caption{Degeneracy and Null Hypothesis}
    \begin{tikzpicture}[scale = 0.75]
\begin{axis}[
  axis lines=middle,
  axis line style={Stealth-Stealth,very thick,black},
  xmin=-1,xmax=5.5,ymin=-1.0,ymax=5.5,
  tick style={draw=none},
  xticklabels=\empty,
  yticklabels=\empty,
  xlabel={$E[(\Delta^z_0-\Delta^z_1)^2]$},
  ylabel={$E[(\Delta^z_0-\Delta^z_2)^2]$},
  axis lines=middle,
    axis line style={->},
    x label style={at={(axis description cs:0.8,0.15)},anchor=north},
    y label style={at={(axis description cs:0.15,.75)},rotate=90,anchor=south},
  xlabel style = {black},
  ylabel style = {black},
  grid=major,
  grid style={thin,densely dotted,white!20}]

\addplot [domain=0:5.5,samples=2, color = black] {x} node[right]{};

\tikzstyle{rej1}=[draw=none,line width=10pt,preaction={clip, postaction={pattern=horizontal lines, pattern color=blue, opacity=0.2}}]
\tikzstyle{rej2}=[draw=none,line width=1pt,preaction={clip, postaction={pattern=vertical lines, pattern color=red, opacity=0.2}}]


\node[text width=3.0cm,rotate = 0,color = black,fill = none] at (400,130) {{$\Gamma_2 = \Gamma_0$}};

\node[text width=5.9cm,rotate = 40,color = black,fill = none] at (410,370) {{ $\Gamma_1 = \Gamma_2$}};

\node[text width=3.0cm,rotate = 0,color = black,fill = none] at (550,250) {{${Q}_1 > {Q}_2$}};

\node[text width=5.9cm,rotate = 40,color = black,fill = none] at (500,540) {{$H^\text{RV}_0$: ${Q}_1 = {Q}_2$}};

\node[text width=2.0cm,rotate = 0,color = black,fill = none] at (300,30) {{$\sigma^2_\text{RV} = 0$}};

\node[text width=3.0cm,rotate = 90,color = black,fill = none] at (125,350) {{$\Gamma_1 = \Gamma_0$}};

\node[text width=3.0cm,rotate = 0,color = black,fill = none] at (300,450) {{${Q}_1 < {Q}_2$}};

\draw[-stealth,decorate,decoration={snake,amplitude=4pt,pre length=3pt,post length=4pt}] (100,100) -- ++(100,-70);
\draw[-stealth,decorate,decoration={snake,amplitude=4pt,pre length=3pt,post length=4pt}] (100,100) -- ++(100,-70);

\draw[red, fill=red] (100,100) circle (10) ;

\end{axis}

\end{tikzpicture}
    \label{fig:degen_null}
    \caption*{\footnotesize{This figure illustrates that the region of degeneracy is a subspace of the null space for RV.}} \vspace{-0.5cm}
\end{figure}

As degeneracy is a special case of the null, maintaining Assumption \ref{as:nodegen} has no consequences for size control if the RV test reliably fails to reject the null under degeneracy.  However, degeneracy can cause size distortions \citep{s15}, though we show that these are only meaningful with one or many instruments. Furthermore, we show that the RV test suffers from a substantial loss of power in the region around degeneracy.  To make this point, we  recast degeneracy as a problem of weak instruments.\looseness=-1

\subsection{Weak Instruments for Testing}

Proposition \ref{prop:degen} shows that degeneracy arises when the predicted markups across models 0, 1 and 2 are indistinguishable. Given the definition of predicted markups in Equation \eqref{eq:predMarkups}, this implies that the projection coefficients from the regression of markups on the instruments: $\Gamma_0$, $\Gamma_1$, and $\Gamma_2$ are also indistinguishable.  Thus, we can rewrite Proposition \ref{prop:degen} as follows:

\begin{corollary}\label{cor:degen}
Suppose Assumptions \ref{as:momcond}--\ref{as:homsked} hold. Then $\sigma^2_\emph{RV}=0$ if and only if  $\Gamma_0-\Gamma_m  = 0$  for $m = 1$ and $m=2$.
\end{corollary}

\noindent 
Degeneracy is characterized by $\Gamma_0-\Gamma_m$ being zero for \textit{both} $m=1$ and $m=2$. Thus when models are fixed and $\Gamma_0-\Gamma_m$ is constant in the sample size, degeneracy is a problem of irrelevant instruments.\looseness=-1

To better capture the finite sample performance of the test when the instruments are nearly irrelevant, it is useful to conduct analysis allowing $\Gamma_0-\Gamma_m$ to change with the sample size.\footnote{For example, in the setting of \cite{a16} with the true model being Nash price setting and the alternative being joint profit maximization, $\Gamma_0-\Gamma_m$ goes to zero. Here, $\Delta_0$ and $\Delta_m$ are nearly constant and therefore have vanishing correlation with any instruments.}
Thus, we forgo the classical approach to asymptotic analysis  where the models are fixed as the sample size goes to infinity.  Instead, we now adapt \cite{ss97}'s asymptotic framework of weak instruments in the following assumption:\looseness=-1 
\begin{assumption}\label{as:wkinstr}
For both $m=1$ and $m=2$, 
\begin{align}
\Gamma_0-\Gamma_m &= {q_m}/{\sqrt{n}} && \text{ for some finite vector } q_m. 
\end{align}
\end{assumption}
\noindent Here, the projection coefficients $\Gamma_0-\Gamma_m$ change with the sample size and are local to zero which enables the asymptotic analysis in the next subsection.  This approach is technically similar to the analysis of local misspecification conducted in Proposition \ref{prop:distrib_TE2}. However, it does not impose Assumption \ref{as:nodegen}. Instead, Assumption \ref{as:wkinstr} implies that $\sigma^2_\text{RV}$ is zero so that degeneracy obtains.  Thus, in the next subsection, we use weak instrument asymptotics to clarify the effect of degeneracy on inference.

\subsection{Effect of Weak Instruments on Inference}

We now use Assumption \ref{as:wkinstr} to show that RV has inferential problems under degeneracy and to provide a diagnostic for instrument strength in the spirit of \cite{sy05}. The diagnostic relies on formulating an $F$-statistic that can be constructed from the data. An appropriate choice is the scaled $F$-statistic for testing the joint null hypotheses of the AR model assessment approach for the two models. The motivation behind this statistic is Corollary \ref{cor:degen}. Note that $\Gamma_0-\Gamma_m = E[z_iz_i']^{-1}E[z_i(p_i-\Delta_{mi})] = \pi_m$, the parameter being tested in AR.  Thus, instruments are weak for testing if both $\pi_1$ and $\pi_2$ are near zero, and degeneracy occurs when the null hypotheses of the AR test for both models, $H_{0,1}^\text{AR}$ and $H_{0,2}^\text{AR}$, are satisfied.\looseness=-1

A benefit of relying on an $F$-statistic to construct a single diagnostic for the strength of the instruments is that its asymptotic null distribution is known. However, it is more informative to scale the $F$-statistic by $1-\hat \rho^2$ where $\hat \rho^2$ is the squared empirical correlation between $e_{1i}-e_{2i}$ and $e_{1i}+e_{2i}$, where $e_m$ is the error in the regression of $p-\Delta_m$ on $z$ used to estimate $\pi_m$.  Expressed formulaically, our proposed $F$-statistic is then
\begin{gather}\label{eq:F}
    F = \big(1-\hat \rho^2\big)\frac{n}{2d_z} \frac{\hat \sigma_2^2 \hat g_1' \hat W \hat g_1 + \hat \sigma_1^2 \hat g_2' \hat W \hat g_2 - 2\hat \sigma_{12} \hat g_1' \hat W \hat g_2}{\hat \sigma_1^2 \hat \sigma_2^2 - \hat \sigma_{12}^2 },
\shortintertext{where}
	\hat \rho^2 = \frac{ \big(\hat \sigma_1^2 - \hat \sigma_2^2\big)^2}{\big( \hat \sigma_1^2 + \hat \sigma_2^2 \big)^2 - 4\hat \sigma_{12}^2 },\,\,\hat \sigma_m^2 = \frac{\text{trace}\big( \hat V_{mm}^\text{AR}\hat W\inverse \big)}{d_z} , \,\,\hat \sigma_{12} = \frac{\text{trace}\big( \hat V_{12}^\text{AR}\hat W\inverse \big)}{d_z}.
	\end{gather}
While maintaining homoskedasticty as in Assumption \ref{as:homsked}, we will describe how $F$ can be used to diagnose the quality of inferences made based on the RV test. In the language of \cite{olea13}, ours is an effective $F$-statistic as it relies on heteroskedasticity-robust variance estimators.\footnote{$F$ is closely related to the likelihood ratio statistic for the test of $\pi_1=\pi_2=0$. However, the likelihood ratio statistic does not scale by $1-\hat \rho^2$ nor does it use heteroskedasticity-robust variance estimators as $F$ does.\looseness=-1} For this reason, we expect that $F$ remains useful in diagnosing weak instruments outside of homoskedastic settings. For simulations that support this expectation in the standard IV case, we refer to \cite{ass19}.

In the following proposition, we characterize the joint distribution of the RV statistic and our $F$. As our goal is to learn about inference and to provide a diagnostic for size and power, we only need to consider when the RV test rejects, not the specific direction.  Thus, we derive the asymptotic distribution of the absolute value of $T^\text{RV}$ in the proposition.  This result forms the foundation for interpretation of $F$ in conjunction with the RV statistic. We use the notation $\boldsymbol e_1$ to denote the first basis vector $\boldsymbol e_1 = (1,0,\dots,0)' \in \mathbb{R}^{d_z}$.\footnote{Proposition \ref{prop:wkinstrF} introduces objects with plus and minus subscripts, as these objects are sums and differences of rotated versions of $W^{1/2}g_1$ and $W^{1/2}g_2$ and their estimators. The role of these objects is discussed after the proposition, while we defer a full definition to Appendix \ref{sect:Proofs} to keep the discussion concise.}

\begin{proposition}\label{prop:wkinstrF}
Suppose Assumptions \ref{as:momcond}--\ref{as:wkinstr} hold. Then
	\begin{align}
	(i) &&
	\begin{pmatrix}
		\abs{T^\emph{RV}} \\ F
	\end{pmatrix}
		 &\xrightarrow{d} \begin{pmatrix}
 	\abs{\Psi_-' \Psi_+}/\!\left( \norm{ \Psi_-}^2 + \norm{ \Psi_+}^2 + 2\rho  \Psi_-' \Psi_+ \right)^{1/2} \\
		\left( \norm{ \Psi_-}^2 + \norm{ \Psi_+}^2 - 2\rho  \Psi_-' \Psi_+ \right)/(2d_z)
 \end{pmatrix}
		\shortintertext{where $\hat \rho^2 \xrightarrow{p} \rho^2$ and}
		&&\begin{pmatrix}
			 \Psi_- \\  \Psi_+
		\end{pmatrix} &\sim N\!\left( \begin{pmatrix}
			\mu_- \mathbf{e}_1 \\ \mu_+ \mathbf{e}_1
		\end{pmatrix}, \begin{bmatrix} 1 & \rho \\ \rho & 1 \end{bmatrix} \otimes I_{d_z} \right)\!, \\
		(ii) && H_0^\emph{RV} &\text{ holds if and only if } \mu_-=0, \\
		(iii) && H_{0,1}^\emph{AR} \text{ and } H_{0,2}^\emph{AR} &\text{ holds if and only if } \mu_+=0, \\
		(iv) && 0 &\le \mu_- \le \mu_+.
	\end{align}
\end{proposition}

The proposition shows that the asymptotic distribution of $T^\text{RV}$ and $F$ in the presence of weak instruments depends on $\rho$ and two non-negative nuisance parameters, $\mu_-$ and $\mu_+$, whose magnitudes are tied to whether $H_0^\text{RV}$ holds, and to whether $H_{0,1}^\text{AR}$ and $H_{0,2}^\text{AR}$ hold, respectively. Specifically, the null of RV corresponds to $\mu_-=0$. Furthermore, the proposition sheds light on the effects that degeneracy has on inference for RV. Unlike the standard asymptotic result, the RV test statistic converges to a non-normal distribution in the presence of weak instruments. For compact notation, let this non-normal limit distribution be described by the variable $T^\text{RV}_\infty = \Psi_-' \Psi_+/\!\left( \norm{ \Psi_-}^2 + \norm{ \Psi_+}^2 + 2\rho  \Psi_-' \Psi_+ \right)^{1/2}$.   Under the null, the numerator of $T^\text{RV}_\infty$ is the product of $\Psi_-$, a normal random variable centered at 0, and $\Psi_+$, a normal random variable centered at $\mu_+ \ge 0$.  When $\rho \neq 0$, the distribution of this product is not centered at zero and is skewed, both of which may contribute to size distortions. 

Alternatives to the RV null are characterized by $\mu_- \in (0,\mu_+]$. For a given value of $\rho$, maximal power is attained when $\mu_-=\mu_+$. This power is strictly below one for any finite $\mu_+$, so that the test is not consistent under weak instruments. Actual power will often be less than the envelope, as $\mu_{-} = \mu_{+}$ generally only occurs with no misspecification. Furthermore, the lack of symmetry in the distribution of the RV test statistic when $\rho \neq 0$ leads to different levels of power against each model.
 \looseness=-1

Ideally, one could estimate the parameters and then use the distribution of the RV statistic under weak instruments asymptotics to quantify the distortions to size and the best-case power that can be attained. However, this is not viable since  $\mu_{-}$, $\mu_{+}$, and the sign of $\rho$ are not consistently estimable.  Instead, we adapt the approach of \cite{sy05} and develop a diagnostic to determine whether $\mu_+$ is sufficiently large to ensure control of the highest possible size distortions for the given value of $\rho^2$.  Given the threat of low power, we  develop a similar diagnostic to ensure a lower bound on the best-case power, which we take to be the maximal power across $\mu_-$ and the sign of $\rho$.\looseness=-1  

One might wonder if robust methods from the IV literature would be preferable when instruments are weak.  For example, AR is commonly described as being robust to weak instruments in the context of IV estimation.  Note that while AR maintains the correct size under weak instruments, this is of limited usefulness for inference with misspecification since neither null is satisfied.  Furthermore, tests proposed in \cite{k02} and \cite{m03} do not immediately apply to our setting. The econometrics literature has also developed modifications of the \cite{v89} test statistic that seek to control size under degeneracy \citep{ s15, sw17}.
While these may be adaptable to our setting, the benefits of size control may come at the cost of lower power.  As we show in the next section, power as opposed to size is the main concern with a moderate number of instruments.\looseness=-1

\subsection{Diagnosing Weak Instruments}\label{sec:diagnosedegen}

To implement our diagnostic for weak instruments, we need to define a target for reliable inference.  Motivated by the practical considerations of size and power, we provide two such targets: a worst-case size $r^s$ exceeding the nominal level of the RV test ($\alpha = 0.05$) and a best-case power $r^p$.  Then, we construct separate critical values for each of these targets.  A researcher can choose to diagnose whether instruments are weak based on size, power, or ideally both by comparing $F$ to the appropriate critical value.   We construct the critical values based on size and power in turn.  

\vspace{0.5em}
\noindent
\textbf{Diagnostic Based On Worst-Case Size:} We first consider the case where the researcher wants to understand whether the RV test has asymptotic size no larger than $r^s$ where $r^s \in (\alpha,1)$. For each value of $\rho^2$, we then follow \cite{sy05} in denoting the values of $\mu_+$ that lead to a size above $r^s$ as corresponding to \textit{weak instruments for size}:\footnote{Because we measure instrument strength by $F$, we define $\cal S$ as a set of non-centralities, $(1-\rho^2)\inverse \mu_+^2$. This is equivalent to defining $\cal S$ in terms of $\mu_+$ only. }
\begin{align}
	{\cal S}(\rho^2,d_z,r^s) = \left\{ \tfrac{\mu_+^2}{1-\rho^2} : \mu_+^2\geq0,\  \Pr\left( \abs*{T_\infty^\text{RV}} > 1.96 \mid \rho^2, \mu_-=0, \mu_+  \right) > r^s \right\}.
\end{align}
Depending on $\rho^2$, $d_z$, and $r^s$, this set may be empty, which occurs for instance with two to nine instruments for any value of $\rho^2$ when $r^s\geq 0.075$. In this case, weak instruments for size are not a concern.

When weak instruments are a possible concern, the role of $F$, viewed through the lens of size control, is to determine whether it is exceedingly unlikely that the true value of the non-centrality $(1-\rho^2)\inverse \mu_+^2$ belongs to $ {\cal S}(\rho^2,d_z,r^s)$. Using the distributional approximation to $F$ in Proposition \ref{prop:wkinstrF} and the standard burden of a five percent probability to denote an exceedingly unlikely event, we say that the instruments are strong for size whenever $F$ exceeds\looseness=-1
\begin{align}\label{eq:cvs}
    cv^s(\rho^2,d_z,r^s) = \tfrac{1-\rho^2}{2d_z}   \chi^2_{2d_z,.95}\left(  \sup {\cal S}(\rho^2,d_z,r^s) \right)
\end{align}
where $\chi^2_{df,.95}(nc)$ denotes the upper $95$th percentile of a non-central $\chi^2$-distribution with degrees of freedom $df$ and non-centrality parameter $nc$. If ${\cal S}$ is empty, we set $cv^s$ to zero.\looseness=-1

In practice, one compares $F$ to the  critical value $cv^s(\hat\rho^2,d_z,r^s)$, which relies on the estimated $\rho^2$. The event $F < cv^s(\hat\rho^2,d_z,r^s)$ expresses that the instruments may be so weak that size is distorted above $r^s$ with high probability. In this case, the researcher should be concerned that rejections of the null may be spurious.  Our diagnostic for size is thus informative about the RV test when the null is rejected.\looseness=-1

\vspace{0.5em}
\noindent 
\textbf{Diagnostic Based on Best-Case Power:} For interpretation of the RV test, particularly when the test fails to reject, it is important to understand the best-case power that the test can attain. By considering rejection probabilities when $\mu_-= \mu_+$ and linking these probabilities to values of $F$, it is also possible to let the data inform us about the power potential of the test. To do so we consider an ex ante desired target of best-case power $r^p$. Because the potential power of the RV test depends on the sign of $\rho$, which is not estimable, we define \textit{weak instruments for power} as the values of $\mu_+$ that lead to best-case power less than $r^p$ for both positive and negative $\rho$:\footnote{When $\mu_-= \mu_+$, the non-centrality of $F$ becomes $(1+\rho)\inverse 2\mu_+ ^2$, which we use to define $\cal P$.  }\looseness=-1
\begin{align}\label{eq:cvp}
	{\cal P}(\rho^2,d_z,r^p) = \left\{ 
 \tfrac{2\mu_+ ^2}{1+\varrho} : \mu_+ ^2\geq0,\ \varrho = \pm\rho,\ \Pr\left( \abs*{T_\infty^\text{RV}} > 1.96 \mid \varrho, \mu_- = \mu_+, \mu_+ \right) < r^p \right\}.
\end{align}

We determine the strength of the instruments by considering the power envelope for the RV test for a given value of $\rho^2$.  This is to ensure that the power against both models exceeds $r^p$ for any value of $\rho$. Again using the distributional approximation to $F$ in Proposition \ref{prop:wkinstrF}, we say that the instruments are strong for power if $F$ is larger than
\begin{align}
    cv^p(\rho^2,d_z,r^p) = \tfrac{1-\rho^2}{2d_z} \chi^2_{2d_z,.95}\left(  \sup {\cal P}(\rho^2,d_z,r^p) \right).
\end{align}

The event $F < cv^p(\hat\rho^2,d_z,r^p)$ expresses that the power against either model must be below $r^p$ with high probability. Therefore, this event informs a researcher that the RV test may fail to reject, not because the two models are very similar, but because the instruments are too weak to tell them apart. In this way, our diagnostic for power is informative about the RV test when the null is not rejected.\looseness=-1

\vspace{0.5em}	
\noindent\textbf{Computing Critical Values:} To compute $cv^s$ for a given $(\rho^2,d_z,r^s)$, we numerically determine ${\cal S}(\rho^2,d_z,r^s)$ by simulating rejection probabilities across a grid of 800 equally spaced values for $\mu_+$ with range from zero to 80.  Once we obtain ${\cal S}(\rho^2,d_z,r^s)$, $cv^s$ is computed according to Equation \eqref{eq:cvs}. To compute  $cv^p$ for a given $(\rho^2,d_z,r^s)$, we use the same procedure as for size, but simulate rejection probabilities with $\mu_{-} = \mu_{+}$ instead of $\mu_{-} = 0$.\looseness=-1

To aid applied researchers, we provide as supplementary material a lookup table of critical values computed for 100 values of $\rho^2$ from 0 to 0.99 and for values of $d_z$ from 1 to 30.  Additionally, the \texttt{pyRVtest} package computes $\hat \rho^2$ and displays the appropriate critical values from this lookup table in any given application.  To further shed light on our diagnostic, we report in Table \ref{tab:Tab_StockYogo_Combined_new} critical values $cv^s$ and $cv^p$ for certain $\rho^2$ and $d_z$.  \looseness=-1

\begin{table}[htb]
\footnotesize
\caption{Critical Values to Diagnose Weak Instruments for Testing}
\label{tab:Tab_StockYogo_Combined_new}
\centering
\begin{threeparttable}

\begin{widetable}{.98\columnwidth}{lS r *{3}{S} r *{3}{S}}
\toprule
&&& \multicolumn{3}{c}{\textbf{Panel A:} Critical Values, ${cv^s}$} &&  \multicolumn{3}{c}{\textbf{Panel B:} Critical Values, ${cv^p}$}\\ 
\cmidrule(lr){4-10}
&&& \multicolumn{3}{c}{Worst-Case Size, $r^s$} &&  \multicolumn{3}{c}{Best-Case Power, $r^p$}\\ 
\cmidrule(lr){4-6} \cmidrule(lr){8-10} 
$\rho^2$ & \mc{$d_z$} && \mc{0.075} & \mc{0.10} &  \mc{0.125} & & \mc{0.95} & \mc{0.75} & \mc{0.50} \\
\midrule
\multirow{8}{*}{$0.25$} 
& 1 && 0 & 0 & 0 && 20.5 & 15.5 & 12.4 \\
&2&&0&0&0&&10.7&8.1&6.6\\
&3&&0&0&0&&7.4&5.7&4.6\\
&4&&0&0&0&&5.7&4.4&3.7\\
&5&&0&0&0&&4.8&3.7&3.1\\
&10&&0&0&0&&2.8&2.2&1.9\\
&20&&4.3&2.0&1.1&&1.7&1.4&1.3\\
&30&&10.3&5.2&3.5&&1.4&1.2&1.1\\
\midrule
\multirow{8}{*}{$0.75$} 
&1&&26.7&12.6&0&&6.5&4.9&3.9\\
&2&&0&0&0&&3.3&2.5&2.0\\
&3&&0&0&0&&2.2&1.7&1.4\\
&4&&0&0&0&&1.7&1.3&1.1\\
&5&&0&0&0&&1.4&1.1&0.9\\
&10&&0.8&0&0&&0.8&0.6&0.5\\
&20&&17.1&8.3&5.4&&0.4&0.4&0\\
&30&&34.4&17.0&11.3&&0.3&0&0\\
\bottomrule
\end{widetable}
\begin{tablenotes}[flushleft]
    \setlength\labelsep{0pt}
    \footnotesize
    \item For a given $d_z$ and $\rho^2$, each row of Panel A reports critical values $cv^s$ for a target worst-case size below $r^s\in \{ 0.075,\ 0.10,\ 0.125\}$. Each row of Panel B reports critical values $cv^p$ for a target best-case power above $r^p \in \{ 0.95,\ 0.75,\ 0.50\}$.  We diagnose the instruments as weak for size if $F \le {cv}^s$, and weak for power if $F \le {cv}^p$.\looseness=-1
\end{tablenotes}
\end{threeparttable}
\end{table}

\vspace{0.5em}
\noindent
\textbf{Discussion of the Diagnostic:} To diagnose whether instruments are weak for size or power, a researcher would compute $F$ and compare it to the relevant critical value for an estimated $\hat\rho^2$.  Table \ref{tab:Tab_StockYogo_Combined_new} reports the critical values used to diagnose whether instruments are weak in terms of size (Panel A) or power (Panel B) for two illustrative values of $\rho^2$.  These critical values explicitly depend on both the number of instruments $d_z$ and a target for reliable inference.  For size, we consider targets of worst-case size below $r^s \in \{ 0.075,\ 0.10,\ 0.125\}$.  For power, we consider targets of best-case power above $r^p \in \{ 0.95,\ 0.75,\ 0.50\}$.  

Suppose a researcher wanting to diagnose whether instruments are weak based on size has twenty instruments and measures $F = 10$ and $\hat\rho^2 = 0.75$.  Given a target worst-case size of 0.10, the critical value in Panel A is 8.3.  Since $F$ exceeds $cv^s$, the researcher concludes that instruments are strong in the sense that size is no larger than $0.10$ with at least 95 percent confidence.  Instead, for a target of 0.075, the critical value is 17.1.  In this case, $F < cv^s$ and the researcher cannot conclude that the instruments are strong for size.  Thus, the interpretation of our diagnostic for weak instruments based on size is analogous to the interpretation that one draws for standard IV when using an $F$-statistic and \cite{sy05} critical values.\looseness=-1

If the researcher also wants to diagnose whether instruments are weak based on power, she can compare $F$ to the relevant critical value in Panel B.  For two instruments, $\hat\rho^2 = 0.25$ and a target best-case power of  0.75, the critical value is again 8.1.  Since $F = 10$, the researcher can conclude that instruments are strong in the sense that the best-case power the test could obtain exceeds 0.75 with at least 95 percent confidence.  Instead, for a target best-case power of 0.95, the critical value is 10.7.  In this case, $F < cv^p$ and the researcher cannot conclude that the instruments are strong for power.

The columns of Panels A and B in Table \ref{tab:Tab_StockYogo_Combined_new} are ordered in terms of increasing maximal type I (Panel A) and type II errors (Panel B).  Unsurprisingly, for a given value of $\rho^2$, the critical values decrease across columns with the target error, as larger $F$-statistics are required to conclude for smaller type I and II errors.  Inspection within each column is useful to understand when size distortions and low power are relevant threats to inference. The RV test statistic has a skewed distribution whose mean is not zero. The effect of skewness on size is largest with one instrument, so in Panel A, the critical value may be large when $d_z=1$ depending on the value of $\rho^2$. As the effect of skewness on size decreases in $d_z$, we find that there are no size distortions exceeding 0.025 with 2-9 instruments for all values of $\rho^2 $. Meanwhile, the effect of the mean on size is increasing in $d_z$, and becomes relevant when $d_z$ exceeds 9. Thus the critical values are monotonically increasing from 10 to 30 instruments. Alternatively, for power, the critical values are monotonically decreasing in the number of instruments.  Taken together, the critical values indicate that (except for the case of one instrument) a lack of power is the main concern when testing with few instruments, while  size distortions are the main concern when testing with many instruments. These considerations interact with the measured value of $\rho^2$: for fixed $d_z$, $cv^s$ is increasing in $\rho^2$, while $cv^p$ is decreasing.\footnote{Not all patterns described in this paragraph are immediately available from Table \ref{tab:Tab_StockYogo_Combined_new}, but are learned from the lookup table in the supplement.}\looseness=-1

To illustrate the usefulness of our $F$-statistic, consider an example where the researcher has two instruments and computes an RV test statistic $T^\text{RV} = 0.54$ and $\hat \rho^2 = 0.25$.  For a target size of 0.075, the critical value is zero regardless of the value of $\rho^2$ and there are no size distortions above 0.025.  Thus,  low power is the only salient concern.  If the $F$-statistic is below 6.6 which is the critical value for having best-case power above 0.5, then the researcher can conclude rejection was very unlikely in this setting even if the null is violated. In other words, when power is the salient concern, our $F$-statistic is necessary to interpret no rejection.  Likewise, when size is a concern, our $F$-statistic is necessary to interpret rejections of the null.  

We develop in Appendix \ref{app:simul} two sets of simulations to show that the $F$-statistic appropriately diagnoses weak instruments for power and for size in finite samples. In the first simulation, we vary the power of the test by injecting noise into the instruments; as the power of the test declines, our proposed $F$-statistic detects this power reduction. In the second simulation, we consider an environment where we approach the region of degeneracy while staying in the null space. Near degeneracy, we find large size distortions for the RV test (around 0.15 using one instrument). Our diagnostic detects when instruments are weak for size.\looseness=-1

Up to this point, we have considered the case where a researcher has one set of instruments she will use for testing two candidate models.  Indeed, if a researcher chooses her instruments for testing once-and-for-all based on intuition, the procedure for testing conduct is straightforward: run the RV test and then inspect whether the instruments pass the diagnostic for strength.  In practice, several sets of instruments may be available to the researcher.  Furthermore,  in many settings including our application, a researcher wants to test more than two models. In the next section, we discuss how an applied researcher can perform RV testing on multiple models with multiple sets of instruments while using the $F$-statistic to guide inference.

\section{Testing Conduct with Multiple Sets of Instruments}\label{sec:accumulate}

\cite{bh14} show that multiple sources of exogenous variation in marginal revenue can be used to construct instruments for testing conduct. As mentioned in Section \ref{sec:nevo}, these typically include demand rotators, own and rival product characteristics, rival cost shifters, and market demographics. By connecting \cite{bh14}'s falsifiable restriction to models' pass-through, \cite{mqsw22} provide a framework that may help a researcher to rule out irrelevant instruments ex-ante. However, there are still interesting open questions about best practices for combining all available sources of exogenous variation. First,  should researchers run one RV test with a single pooled set of instruments or should they keep the sources of variation separate and run multiple RV tests?   Second, which functional form should researchers use to construct instruments from their chosen sources?  The latter point is addressed in \cite{bcs20}, who consider efficiency in the spirit of \cite{c87}. 
In this section, we focus instead on the first consideration.

Based on the results in Sections \ref{sec:hypo_formulation} and \ref{sec:instruments}, there are two main reasons a researcher may want to keep the sources of variation separate.   First, drawing inference on conduct by pooling sources of variation can conceal the severity of misspecification. As seen in Section \ref{sec:hypo_formulation}, the RV test concludes for the model with the lower MSE of predicted markups. With misspecification, strong instruments constructed from economically different sources of variation (e.g., demand shifters versus rival cost shifters) could  conclude for different models.     By keeping the sources of variation separate and running multiple RV tests, a researcher can observe such conflicting evidence.  Instead, a single RV test run with pooled instruments could conclude for one model, obscuring the severity of misspecification.   Below, Example 4 provides an economic setting where misspecifying models generates conflicting evidence.\looseness=-1

Second,  pooling variation may have adverse consequences for the strength of the resulting instrument set, which occurs in our empirical application (see Appendix \ref{sect:Robustness}).  For example,  if some sources of variation on their own yield weak instruments for power, combining these with strong instruments dilutes the power of the strong instruments, manifesting itself in a lower $F$-statistic.  Furthermore, Panel A of Table \ref{tab:Tab_StockYogo_Combined_new} shows that the combined set of instruments faces a larger critical value for size.  Thus, if pooling across sources creates many instruments, size distortions can undermine inference on firm conduct.

\vspace{0.5em}
\noindent \textit{Example 4}: Consider the case of two firms competing in a market where demand is logit, as in the examples in \cite{mqsw22}.  Suppose that the true model of conduct is Nash quantity setting, and a researcher specifies two incorrect models: perfect competition ($m=1$) and  Nash price setting ($m=2$).  The researcher constructs two sets of instruments, one from variation in rival cost and the other from variation in own product characteristics.  With cost instruments, $\Delta^z_0 = \Delta^z_1 = 0$ while $\Delta^z_2\neq 0$, and the researcher concludes for perfect competition.  Instead, variation in own product characteristics moves both Nash quantity setting and Nash price setting markups, but not markups under perfect competition.  Under some formulations of demand and cost, the researcher concludes for Nash price setting.  Because both models are misspecified, the two sets of instruments generate conflicting evidence.\looseness=-1

\vspace{0.5em}
\noindent \textbf{Accumulating Evidence:} Researchers who  want to keep their sources of variation separate need to aggregate information across multiple RV tests. We suggest that a researcher can conclude for a model insofar as   there is no conflicting evidence across sets of instruments and  all the strong instruments support it. Continuing the legal analogy made in Section \ref{sec:hypo_formulation}, we have adopted a preponderance of the evidence standard by using model selection.  However, we may not want to rely on a single piece of evidence to convict, nor would we  want to rely on weak evidence.  To achieve the two aims above, we propose a conservative approach that utilizes both the RV test and the $F$-statistic.

Suppose we want to test a set of two models $M=\{1,2\}$ using $L$ sets of instruments. We suggest researchers use sets of 2-9 instruments, so that there are no size distortions above 0.025 and power is the salient concern. In a preliminary step we run separate RV tests with each instrument set $z_\ell$ and denote the model confidence set (MCS) $M^*_\ell$ as the set of models that are not rejected.\footnote{Thus, $M^*_\ell = \{1,2\}$ if the RV null is not rejected and $M^*_\ell = \{1\}$ if the RV null is rejected in favor of a superior fit of model 1.} Our goal is to generate $M^*$, an MCS which aggregates evidence from all $M^*_\ell$. Our approach,  illustrated in Figure \ref{fig:flowchart},   proceeds in two steps.  In step 1, the researcher needs to check that the evidence coming from the $L$ sets of instruments is not in conflict.  We say that evidence arising from $L$ RV tests  is not in conflict if, for every pair of ($M^*_\ell$, $M^*_{\ell'}$), one is a weak subset of the other. In step 2 we form $M^*$ based on step 1. If the evidence is in conflict, the researcher concludes $M^* = \{1,2\}$. If the evidence is not in conflict, we first set $M^*$ equal to the smallest MCS $M^*_\ell$, and then take the union with all MCS for which the instruments are strong based on the $F$-statistic.\footnote{While difficult to establish the coverage probability of $M^*$ in general, $M^*$ is an asymptotically valid MCS when all instrument sets are strong as it is a union of $L$ valid confidence sets.}  \looseness=-1

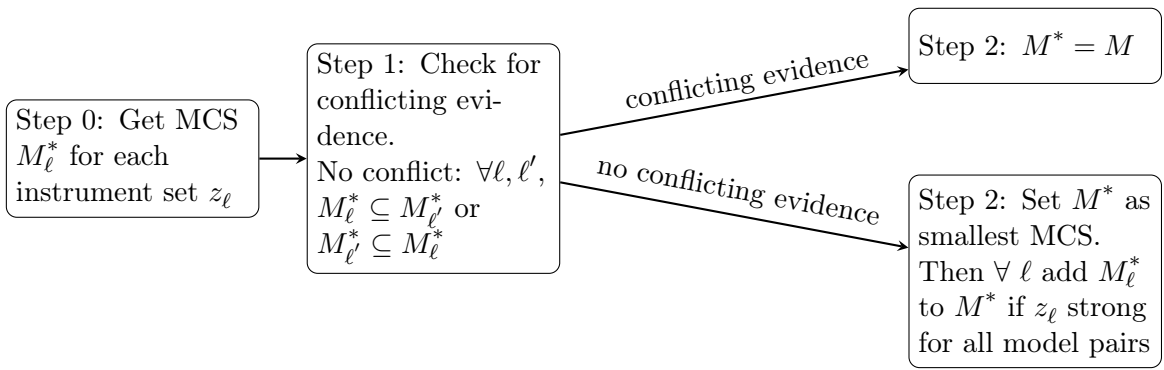
\begin{figure}[ht] 
\centering \caption{Procedure for Accumulating Evidence}
\label{fig:flowchart}
\begin{tikzpicture}[node distance=2cm]

\node (start) [flowbox] {Step 0: Get MCS $M^*_\ell$  for each\\ instrument set $z_\ell$};
\node (in1) [flowbox, right of=start, xshift = 2cm] {Step 1: Check for conflicting evidence.\\ No conflict: $\forall \ell,\ell'$, $M^*_\ell\subseteq M^*_{\ell'} \,\, \text{or}$ \\$M^*_{\ell'}\subseteq M^*_{\ell}$ };
\node (pro1) [flowbox, right of=in1, yshift = 1.5cm, xshift = 6cm] {Step 2: $M^* = M$};
\node (pro2) [flowbox, right of=in1, yshift = -1.5cm, xshift = 6cm] {Step 2: Set $M^*$ as smallest MCS.\\ Then $\forall$ $\ell$ add $M^*_\ell$ to $M^*$ if $z_\ell$ strong for all model pairs};

\draw [arrow] (start) -- (in1);
\draw [arrow] (in1) -- node[anchor=east,rotate=10,xshift=2cm,yshift = 0.25cm] {conflicting evidence} (pro1);
\draw [arrow] (in1) -- node[anchor=south, rotate=-10.5] {no conflicting evidence} (pro2);

\end{tikzpicture}
    \vspace{-0.5cm}
\end{figure}

To illustrate the rationale behind our approach, we consider a few examples in which $L = 2$. First, we illustrate the importance of step 1. If the researcher had computed $M^*_1 = \{1\}$ and $M^*_2 = \{2\}$, then the instruments $z_1$ suggest model 2 can be rejected in favor of superior fit of model 1 while $z_2$ suggest the exact opposite. As $M^*_1 \not\subseteq M^*_2$ and $M^*_2 \not\subseteq M^*_1$ we say the evidence is in conflict. Hence, misspecification is severe and the researcher should let  $M^* = \{1,2\}$, in line with the conservative spirit of the procedure.

Suppose now that $M^*_1 = \{1\}$ and $M^*_2 = \{1,2\}$.  Because $M^*_1 \subset M^*_2$, there is no conflicting evidence found in step 1. In step 2 we initialize $M^* = \{1\}$, the smallest MCS.  By doing so, we use the information that $z_1$ reject model 2 regardless of the power potential diagnosed by the $F$-statistic.  We then only add model 2 to  $M^*$ if the $F$-statistic suggests that instruments $z_2$ are strong for power.  If $z_2$ are weak, then not rejecting the  null is likely a consequence of low power and not informative about firm conduct.

\vspace{0.5em}
\noindent  \textbf{Extension to More than Two Models:} In many settings, including our application, a researcher may want to test a set $M$ of more than two models.  To accumulate evidence across sets of instruments using the procedure in Figure \ref{fig:flowchart}, we need to define $M^*_\ell$ for each of the $L$ instrument sets. We adopt the  procedure of \cite{hln11} to construct each $M^*_\ell$.  This procedure initializes the $M^*_\ell$ to $M$, and then checks in each iteration whether the model of worst fit according to MSE of predicted markups can be excluded.  This occurs if the largest RV test statistic in magnitude across all pairs of models in $M^*_\ell$ exceeds the $(1-\alpha)$-th quantile of its asymptotic null distribution.\footnote{This quantile can be simulated by drawing from the asymptotic null distribution, see Appendix \ref{sect:MCSDetail}.}  When no model can be excluded, the procedure stops. If there are only two models, this procedure coincides with the RV test as discussed above. As shown in \cite{hln11}, ${M}^*_\ell$ controls the familywise error rate as it  contains the model(s) with the best fit with probability at least $1-\alpha$ in large samples. Moreover, every other model with strictly worse fit is excluded from $M^*_\ell$ with probability approaching one.\footnote{Under no degeneracy, $M^*$ is guaranteed to contain the true model with probability at least $1-\alpha$, as each $M^*_\ell$ has the same property and $M^*$ is the union of these model confidence sets.} \looseness=-1
    
To illustrate the construction of  ${M}^*_\ell$, suppose a researcher wants to test candidate models $m=1,2,3$.  For a given set of instruments $z_\ell$, the MCS procedure computes three RV test statistics $T^\text{RV}_{m,m'} = {\sqrt{n}(\hat{Q}_{m\phantom{\!'}}-\hat{Q}_{m'})}/{\hat\sigma_{\text{RV},mm'}}$, one for each distinct pair of models.  Suppose $T^\text{RV}_{1,2} = 5.34$, $T^\text{RV}_{1,3} = 4.35$, and $T^\text{RV}_{2,3} = 0.32$.  If $T^\text{RV}_{1,2}$, the largest test statistic in magnitude, exceeds the  critical value for the max of three RV test statistics, then model 1 is excluded from $M^*_\ell$.  In the next iteration, only models 2 and 3 remain, so the only relevant RV test statistic is $T^\text{RV}_{2,3} = 0.32$.  As the null of equal fit cannot be rejected, $M^*_\ell = \{2,3\}$.

\section{Application: Testing Vertical Conduct}
\label{sec:empiricaleg}
We revisit the empirical setting of \cite{v07}.  She investigates the vertical relationship of yogurt manufacturers and supermarkets by testing different models of vertical conduct.\footnote{\cite{v07} uses a Cox test which is a model assessment procedure with similar properties to AR, as shown in Appendix \ref{sect:EBCox}.}  This setting is ideal to illustrate our results as theory suggests a rich set of models and the data is used in many applications. 

\noindent \subsection{Data}\label{sec:data}

Our main source of data is the IRI Academic Dataset for 2010 \citep[see][for a description]{bkm08}. 
This dataset contains weekly price and quantity data for UPCs sold in a sample of stores in the United States. 
We define a market as a retail store-quarter and approximate the market size with a measure of the traffic in each store,  derived from the store-level revenue information from IRI. We drop the 5\% of stores for which this approximation results in an unrealistic outside share below 50\%. \looseness=-1 

We further restrict attention to UPCs labelled as ``yogurt'' in the IRI data and focus on the most commonly purchased sizes: 6, 16, 24 and 32 ounces. 
Similar to \cite{v07}, we define a product as a brand-fat content-flavor-size combination, where flavor is either plain or other and fat content is either light (less than 4.5\% fat content) or whole. We further standardize package sizes by measuring quantity in six ounce servings. Based on market shares, we exclude niche firms for which their total inside share in every market is below five percent.  We drop products from markets for which their inside share is below 0.1 percent. 
Our final dataset has  205,123 observations for 5,034 markets corresponding to  1,309 stores.\looseness=-1

We supplement our main dataset with county level demographics from the Census Bureau's PUMS database which we match to the DMAs in the IRI data.  We draw 1,000 households for each DMA and record standardized household income and age of the head of the household. We exclude households with income lower than \$12,000 or larger than \$1 million. We also obtain quarterly data on regional diesel prices from the US Energy Information Administration. With these prices, we measure transportation costs as average fuel cost times distance between a store and manufacturing plant.\footnote{We thank Xinrong Zhu for generously sharing manufacturer plant locations used in \cite{z21}.}  We summarize the main variables for our analysis in  Table \ref{Tab:summarystats}.\looseness=-1

\begin{table}[!ht]
\footnotesize
\caption{Summary Statistics}
\label{Tab:summarystats}
\centering
\begin{threeparttable}

\begin{widetable}{.98\columnwidth}{lrrrrr} 
\toprule
Statistic & \multicolumn{1}{c}{Mean} & \multicolumn{1}{c}{St. Dev.} & \multicolumn{1}{c}{Median}  & \multicolumn{1}{c}{Pctl(25)} & \multicolumn{1}{c}{Pctl(75)} \\ 
\midrule
Price (\$) & 0.76 & 0.30 & 0.68 &  0.55 & 0.91 \\ 
Sales (6 oz. units) & 1,461 & 3,199 & 503  & 213 & 1,301 \\ 
Shares & 0.007 & 0.012 & 0.003 & 0.001 & 0.007 \\ 
Outside Share & 0.710 & 0.111 & 0.708 & 0.631 & 0.788 \\ 
Size (oz.) & 17.82 & 10.57 & 16  & 6 & 32 \\ 
Frac. Light & 0.93 & 0.26 & 1 &  1 & 1 \\ 
Number Flavors & 5.39 & 5.81 & 3  & 1 & 8 \\ 
Frac. Private Label & 0.09 & 0.28 & 0  & 0 & 0 \\ 
Distance to Plant (mi.) & 493 & 477 & 392  & 199 & 546 \\ 
Freight Cost (\$) & 212 & 242 & 164  & 52 & 271 \\ 
\bottomrule
\end{widetable} 

\begin{tablenotes}[flushleft]
    \setlength\labelsep{0pt}
    \footnotesize
    \item Source: IRI Academic Dataset for 2010 \citep[][]{bkm08}.
\end{tablenotes}
\end{threeparttable}
\end{table}    

\subsection{Demand: Model, Estimation, and Results}

To perform testing, we need to estimate demand and construct the markups implied by each candidate model of conduct.

\vspace{0.5em}
\noindent \textbf{Demand Model:}
Our model of demand follows \cite{v07} in adopting the framework from  \cite{blp95}.  Each consumer $i$ receives utility from product $j$ in market $t$ according to the indirect utility:\looseness=-1
\begin{align}\label{eq:utility}
    \boldsymbol{u}_{ijt} = \beta^x_i \boldsymbol{x}_j + \beta^p_i \boldsymbol{p}_{jt} + \boldsymbol{\xi}_{t} + \boldsymbol{\xi}_{s} + \boldsymbol{\xi}_{b(j)} + \boldsymbol{\xi}_{jt} + \boldsymbol{\epsilon}_{ijt} 
\end{align}
where $\boldsymbol{x}_j$ includes package size, dummy variables for low fat yogurt and for plain yogurt, and the log of the number of flavors offered in the market to capture differences in shelf space across stores.  $\boldsymbol{p}_{jt}$ is the price of product $j$ in market $t$,  and $\boldsymbol{\xi}_{t}$, $\boldsymbol{\xi}_{s}$, and $\boldsymbol{\xi}_{b(j)}$   
denote fixed effects for the quarter, store, and brand producing product $j$ respectively.  $\boldsymbol{\xi}_{jt}$ and $\boldsymbol{\epsilon}_{ijt}$ are unobservable shocks at the product-market and the individual product market level, respectively.  Finally, consumer preferences for characteristics ($\beta^x_i$) and price ($\beta^p_i$) vary with individual level income and age of the head of household:
\begin{align}
    \beta^p_i = \bar\beta^p + \tilde\beta^p \boldsymbol{D}_i, \qquad
    \beta^x_i = \bar\beta^x + \tilde\beta^x \boldsymbol{D}_i,
\end{align}    
where $\bar\beta^p$ and $\bar\beta^x$ represent the mean taste, $\boldsymbol{D}_i$ denotes demographics, while $\tilde\beta^p$ and $\tilde \beta^x$ measure how preferences change with $\boldsymbol{D}_i$.

To close the model we make additional standard assumptions.  We normalize consumer $i$'s utility from the outside option as $\boldsymbol{u}_{i0t} = \boldsymbol{\epsilon}_{i0t}$.  The shocks $\boldsymbol{\epsilon}_{ijt}$ and $\boldsymbol{\epsilon}_{i0t}$ are assumed to be distributed i.i.d. Type I extreme value. Assuming that each consumer purchases one unit of the good that gives her the highest utility from the set of available products $\mathcal{J}_{t}$, the market share of product $j$ in market $t$ takes the following form:\looseness=-1      
\begin{align}\label{eq:shares}
\boldsymbol{s}_{jt} = \int \frac{\exp({\beta^x_{i} \boldsymbol{x}_{j}  + \beta^p_i \boldsymbol{p}_{jt} +\boldsymbol{\xi}_{t} + \boldsymbol{\xi}_{s} + \boldsymbol{\xi}_{b(j)}  + \boldsymbol{\xi}_{jt}})}{1 +\sum_{l\in{\mathcal{J}_{t}}} \exp({\beta^x_{i} \boldsymbol{x}_{l}  + \beta^p_i \boldsymbol{p}_{lt} +\boldsymbol{\xi}_{t} + \boldsymbol{\xi}_{s} + \boldsymbol{\xi}_{b(l)}  +  \boldsymbol{\xi}_{lt}})}f(\beta^p_i,\beta^x_i) \ \text{d}\beta^p_i \ \text{d}\beta^x_i.
\end{align}

\vspace{0.5em}
\noindent \textbf{Identification and Estimation:}   Demand estimation and testing can either be performed \textit{sequentially}, in which demand estimation is a preliminary step, or \textit{simultaneously} by stacking the demand and supply moments.  Following \cite{v07}, we adopt a sequential approach which is simpler computationally while illustrating the empirical relevance of the findings in Sections \ref{sec:hypo_formulation},
 \ref{sec:instruments}, and \ref{sec:accumulate}.\footnote{We compare the sequential and simultaneous approach for testing conduct in Appendix \ref{app:simul}.}\looseness=-1  

The demand model is identified under the assumption that demand shocks $\boldsymbol{\xi}_{jt}$ are orthogonal to a vector of demand instruments.  By shifting supply, transportation costs help to identify the parameters $\bar \beta^p$, $\tilde \beta^p$, and $\tilde \beta^x$.
Following \cite{gh19}, we use variation in mean demographics across DMAs as a source of identifying variation by interacting them with both fuel cost and product characteristics.\footnote{This class of instruments is relevant for demand insofar as the true markups depend on own and rival product characteristics. This is true for all the models we test below.}  We estimate demand as in \cite{blp95} using \texttt{PyBLP} \citep{cg19}.\looseness=-1

\begin{table}[htb]
\footnotesize
\caption{Demand Estimates}
\label{Tab:demand}
\centering
\begin{threeparttable}
\begin{widetable}{.98\columnwidth}{lrrrrrr}
\toprule
& \multicolumn{2}{c}{(1) Logit-OLS}  & \multicolumn{2}{c}{(2) Logit-2SLS}  & \multicolumn{2}{c}{(3) BLP}   \\
\cmidrule(lr){2-3} \cmidrule(lr){4-5} \cmidrule(lr){6-7}
& coef. & \multicolumn{1}{c}{s.e.}
& coef. & \multicolumn{1}{c}{s.e.}
& coef. & \multicolumn{1}{c}{s.e.}    \\
\midrule
 Prices         &    $-1.750$ &   (0.019)    &    $-6.519$  &   (0.209)    &  $-12.001$  & (0.777)\\
 Size           &    0.037 &    (0.001)   &    0.018 &   (0.001)    &  $-0.060$  & (0.013) \\
 Light          &    0.259  &   (0.010)  &    0.413 &   (0.014)    &  $-0.270$  & (0.144) \\
 Plain   &    0.508 &   (0.007)   &    0.423  &   (0.009)   &  0.439 & (0.012) \\
 log(\#Flavors)   &    1.127  &   (0.004)  &    1.106 &   (0.005)    &  1.135    & (0.007) \\
 Income $\times$ price &      &       &        &      &  4.333   & (0.378)\\
 Income $\times$ light &     &        &        &      &  0.215  & (0.069)\\
Age $\times$ light    &      &       &      &        & $-0.565$  & (0.113)  \\
 Age $\times$ size     &    &         &   &           & $-0.067$& (0.008)  \\
\midrule 
 Own price elasticity-mean       & \multicolumn{2}{c}{-1.32}  & \multicolumn{2}{c}{-4.917} & \multicolumn{2}{c}{-6.306} \\
 Own price elasticity-median     & \multicolumn{2}{c}{-1.177} & \multicolumn{2}{c}{-4.384} & \multicolumn{2}{c}{-6.187} \\
 Diversion outside option-mean   &  \multicolumn{2}{c}{0.72}  &  \multicolumn{2}{c}{0.72}  &  \multicolumn{2}{c}{0.39}  \\
 Diversion outside option-median &  \multicolumn{2}{c}{0.71}  &  \multicolumn{2}{c}{0.71}  &  \multicolumn{2}{c}{0.38}  \\
\bottomrule
\end{widetable}

\begin{tablenotes}[flushleft]
    \setlength\labelsep{0pt}
    \footnotesize
    \item We report demand estimates for a logit model of demand obtained from OLS estimation in column 1 and 2SLS  estimation in column 2. Column 3 corresponds to the full BLP model. All specifications have fixed effects for quarter, store, and brand. $n=205,123$. 
\end{tablenotes}
\end{threeparttable}
\end{table}

\vspace{0.5em}
\noindent \textbf{Results:} Results for demand estimation are reported in Table \ref{Tab:demand}.  As a reference, we report estimates of a standard logit model of demand in Columns 1 and 2.  In Column 1, the logit model is estimated via OLS. In Column 2, we use transportation cost as an instrument for price and estimate the model via 2SLS.  When comparing OLS and 2SLS estimates, we see a large reduction in the price coefficient, indicative of endogenity not controlled for by the fixed effects.  Column 3 reports estimates of the full demand model which generates elasticities comparable to those obtained in \cite{v07}. Although our model is simpler than the one she uses, the implied diversion to the outside option is far from logit.

\subsection{Test for Conduct}

\noindent \textbf{Models of Conduct:}  We consider five models of vertical conduct from \cite{v07}. A full description of the models is in Appendix \ref{sect:Robustness}.\footnote{\cite{v07} also considers retailer collusion and vertically integrated monopoly. As we do not observe all retailers in a geographic market, we cannot test those models.}
\begin{enumerate}
\item \textit{Zero wholesale margin:}  Retailers choose retail prices, wholesale price is set to marginal cost and retailers pay manufacturers a fixed fee.  
\item \textit{Zero retail margin:} Manufacturers choose retail prices, and pay retailers a fixed fee.
\item \textit{Linear pricing:} Manufacturers, then retailers, set prices.
\item \textit{Hybrid model:} Retailers are vertically integrated with their  private labels.
\item \textit{Wholesale Collusion:}  Manufacturers act to maximize joint profit.
\end{enumerate}

Given our demand estimates, we compute implied markups $\boldsymbol{\Delta}_m$ for each model $m$. We specify marginal cost as a linear function of observed shifters and an unobserved shock. We include in $\textbf{w}$ an estimate of the transportation cost for each manufacturer-store pair and dummies for quarter, brand and city.

\vspace{0.5em}
\noindent \textbf{Inspection of Implied Markups and Costs:}  Economic restrictions on price-cost margins ${\boldsymbol{\Delta}_m}/{\boldsymbol{p}}$ (PCM) and estimates of cost parameters $\tau$ may be used to learn about conduct, and are complementary to formal testing. For every model, we estimate $\tau$ by regressing implied marginal cost on the transportation cost and fixed effects. The coefficient of transportation cost is positive for all models, consistent with intuition. Thus, no model can be ruled out based on estimates of $\tau$.\looseness=-1

Figure \ref{fig:mkups_plot} reports the distributions of PCM for all models. Compared to Table 7 in \cite{v07}, our PCM are qualitatively similar both in terms of median and standard deviations, and have the same ranking across models. While distributions of PCM are reasonable for models 1 to 4, 
model 5 implies PCM that are greater than 1 (and thus negative marginal cost) for $32$ percent of observations. We rule out model 5 based on the figure alone. However, discriminating between models 1 to 4 requires our more rigorous procedure.\footnote{Including model 5 in our testing procedure does not change our results as it is always rejected.}\looseness=-1 

\begin{figure}[ht] 
    \centering
    \caption{Distributions of PCM}
    \label{fig:mkups_plot}
    \vspace{-0.25cm}
    \includegraphics[scale=0.60]{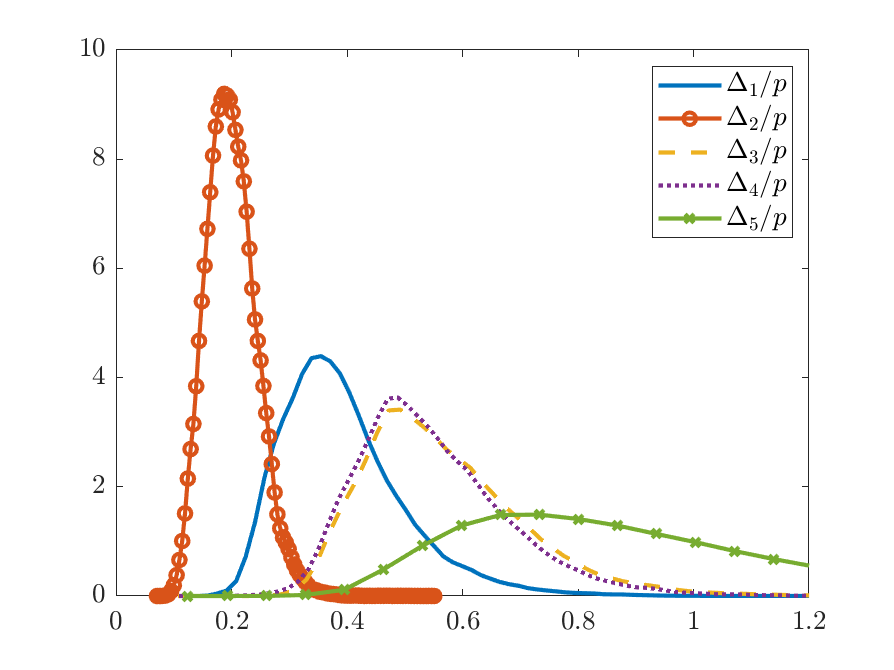}
    \caption*{\footnotesize{We report the distribution of ${\boldsymbol{\Delta}_{mi}}/{\boldsymbol {p}_i}$, the  unresidualized PCM, implied by each model.}}
    \vspace{-0.5cm}
\end{figure}

\noindent \textbf{Model Falsification and Instruments:} Instruments must first be exogenous for testing.  Following  \cite{bh14},  several sources of variation may be used to construct exogenous instruments. These include: (i) both observed and unobserved characteristics of other products, (ii) own observed product characteristics (excluded from cost), (iii) the number of other firms and products, (iv) rival cost shifters, and (v) market level demographics.  Instruments must also be relevant for testing.  Lemma \ref{as:falsify} shows that differences in predicted markups across models distinguish conduct. To distinguish models 1 and 2 we thus need to differentially move downstream markups, while to distinguish 1, 3, 4, and 5 we need to differentially move upstream markups. 

Theoretically, for every pair of models, variation in sources (i)--(v) move upstream and downstream markups for at least one model, making them plausibly relevant. \cite{mqsw22} show that whether instruments differentially move markups depends on the passthrough matrices of the two models, interacted with how instruments move equilibrium prices. To provide a concrete example of the economic determinants of falsification, consider models 1 and 2. In an environment with a simpler demand system than we consider, \cite{mqsw22} derive passthrough matrices under models 1 and 2, and show that instruments related to sources (i)--(v) will falsify either model when the other one is the truth. Because a more flexible form of demand makes it generically easier to falsify models, we have good reason ex-ante to believe that sources (i)--(v) may generate relevant instruments.  

We then need to form instruments from the exogenous and plausibly relevant sources of variation. We consider four instrument choices constructed from these sources that are standard in estimating demand \citep{gh19} and have been used in testing conduct \citep[see e.g.,][]{bcs20}. 

We first leverage sources of variation (i)-(iii) by considering two sets of BLP instruments: the number of own and rival products in a market (NoProd) and the differentiation instruments proposed in \cite{gh19} (Diff).  These instruments have been shown to perform well in applications of demand estimation.  As they leverage variation in the products firms offer and move markups, they are appropriate choices in our setting.  For product-market $jt$, let $O_{jt}$ be the set of products other than $j$ sold by the firm that produces $j$, and let $R_{jt}$ be the set of products produced by rival firms. For product characteristics $\boldsymbol{x}$, the instruments are:\looseness=-1
\[\boldsymbol{z}^\text{NoProd}_{jt}=\left[\begin{matrix} \sum\limits_{k\in O_{jt}} 1[k\in O_{jt}] & \sum\limits_{k\in R_{jt}} 1[k\in R_{jt}]  \end{matrix}\right]\]
\[\boldsymbol{z}^\text{Diff}_{jt}=\left[\begin{matrix}\sum\limits_{k\in O_{jt}} 1[|\boldsymbol{d}_{jkt}|<sd(\boldsymbol{d})] & \sum\limits_{k\in R_{jt}} 1[|\boldsymbol{d}_{jkt}|<sd(\boldsymbol{d})] \end{matrix}\right]\]
where  $\boldsymbol{d}_{jkt}\equiv \boldsymbol{x}_{kt}-  \boldsymbol{x}_{jt}$ and $sd(\boldsymbol{d})$ is the vector of standard deviations of the pairwise differences across markets for each characteristic.\footnote{Following  \cite{c12}, \cite{c17}, and \cite{bcs20}, we perform RV testing with the leading principal components of the Diff instruments.  We choose the number of principal components corresponding to 95\% of the total variance, yielding five Diff instruments. The results below do not qualitatively depend on our choice of principal components.} To form instruments from rival cost shifters, we average transportation costs of rival firms' products (Cost). Finally, \cite{gh19} suggests that variation in demographics can be leveraged for demand estimation by interacting market level moments with product characteristics.  Given the heterogeneity in consumer preferences in our demand system, we interact mean income with light and mean age with size and light to construct our fourth set of instruments (Demo).\looseness = -1

\vspace{0.5em}
\noindent \textbf{AR Test:} We first perform the AR test with the NoProd instruments. Table \ref{Tab:Model_assessment_all} reports test statistics obtained for each pair of models. The results illustrate Propositions \ref{prop:nulls} and \ref{prop:distrib_TE2}: AR rejects all models when testing with a large sample.\footnote{In Appendix \ref{sect:Robustness} we show EB and Cox tests also reject all models.} 

\begin{table}[ht]
\footnotesize
\caption{AR Test Results}
\label{Tab:Model_assessment_all}
\centering
\begin{threeparttable}
\begin{widetable}{.98\columnwidth}{lccc}
\toprule
\textbf{NoProd IVs} & 2 & 3 & 4  \\
\cmidrule(lr){2-4}
1. Zero wholesale margin&315.34, 575.29&315.34, 398.27&315.34, 396.08\\
2. Zero retail margin &&575.29, 398.27&575.29, 396.08\\
3. Linear pricing&&&398.27, 396.08\\
4. Hybrid model&&&\\
\bottomrule
\end{widetable}

\begin{tablenotes}[flushleft]
    \setlength\labelsep{0pt}
    \footnotesize
    \item Each cell reports $T^\text{AR}_i,T^\text{AR}_j$ for row model $i$ and column model $j$, with NoProd instruments. For 95\% confidence, the critical value is 5.99. Standard errors account for two-step estimation error  and clustering at the market level; see Appendix \ref{sect:TwoStep}.\looseness=-1
\end{tablenotes}
\end{threeparttable}
\end{table}   

\noindent \textbf{RV Test:} We perform RV tests using  NoProd, Diff, Cost, and Demo instruments.  Following  Section \ref{sec:accumulate}, we keep the instrument sets separate and construct model confidence sets using the procedure of \cite{hln11}.  We report the results in Table \ref{Tab:RV_app}.\footnote{The results are computed with the Python package \texttt{pyRVtest} available on GitHub \citep{dmss_code}. The package, portable to a wide range of applications, seamlessly integrates with \texttt{PyBLP} \citep{cg19} to import results of demand estimation. A researcher needs only specify the models they want to test, the instruments and the cost shifters, and the package outputs all the information in  Table \ref{Tab:RV_app}.  The variance estimators developed in this paper enable fast computation of all elements in that table, even in large datasets and with flexible demand systems.} \looseness=-1

To aid the reader, we begin by explaining Panel A.  The first three columns give the pairwise RV test statistics for all pairs of models.  For each pair, a value above $1.96$ indicates rejection of the null of equal fit in favor of the column model.  Instead, a value below $-1.96$ corresponds to rejection in favor of the row model.  The second three columns give all the pairwise $F$-statistics.  Finally, the last column reports the MCS $p$-values.  In Panel A, the MCS contains only model 2 corresponding to zero retail margins; the MCS $p$-value for the other three models is below 0.05, our chosen level.  The NoProd instruments are strong for testing: there are no size distortions above 0.025 with two instruments and each pairwise $F$-statistic exceeds the critical value for target best-case power of 0.95.\footnote{Across all values of $\rho^2$, the largest critical value for best-case power of 0.95 when testing with two instruments is 18.9. The lookup table of critical values is part of \texttt{pyRVtest}.} 
If a researcher precommitted to the NoProd instruments for testing, Panel A shows the results that would obtain.\looseness=-1

\begin{table}[h!]
\footnotesize
\caption{RV Test Results}
\label{Tab:RV_app}
\centering
\begin{threeparttable}
\begin{widetable}{.98\columnwidth}{lrrrrrrc}
\toprule
& \multicolumn{3}{c}{$T^\text{RV}$}& \multicolumn{3}{c}{$F$-statistics}&  MCS $p$-values\\
\cmidrule(lr){2-4}  \cmidrule(lr){5-7} \cmidrule(lr){8-8}  
Models & 2 & 3 &  4 & 2 & 3 & 4 & \\ 
\midrule
\textbf{Panel A: NoProd IVs} $(d_z = 2)$ \\ 
1. Zero wholesale margin& 4.33&$-$8.54& $-$8.56&143.9&126.2&126.8&0.00\\
2. Zero retail margin&&$-$7.35&$-$7.36&&100.2&100.8&{1.00}\\
3. Linear pricing&&&$-$3.10&&&204.9&0.00\\
4. Hybrid model&&&&&&&0.00\\
\cmidrule(lr){1-8}
\textbf{Panel B: Demo IVs} $(d_z = 3)$ \\
1. Zero wholesale margin& 2.42&$-$6.38& $-$6.41&30.7&52.4&53.1&0.02\\
2. Zero retail margin&&$-$5.36&$-$5.41&&37.8&38.1&{1.00}\\
3. Linear pricing&&&0.38&&&47.5&0.00\\
4. Hybrid model&&&&&&&0.00\\
\cmidrule(lr){1-8}
\textbf{Panel C: Cost IVs} $(d_z = 1)$\\
1. Zero wholesale margin& $-$1.99&$-$1.51& $-$1.80&106.8&6.0$^\dagger$&7.1$^\dagger$&{1.00}\\
2. Zero retail margin&&1.83&1.77&&86.7&87.7&{0.10}\\
3. Linear pricing&&&$-$2.97&&&91.7&{0.13}\\
4. Hybrid model&&&&&&&0.01\\
\cmidrule(lr){1-8}
\textbf{Panel D: Diff IVs} $(d_z = 5)$ \\
1. Zero wholesale margin& 0.81&0.52&0.51&6.2$^\dagger$&3.3$^\dagger$&3.3$^\dagger$&{0.67}\\
2. Zero retail margin&&0.37&0.36&&2.9$^\dagger$&2.9$^\dagger$&{0.71}\\
3. Linear pricing&&&$-$1.18&&&1.8$^\dagger$&{1.00}\\
4. Hybrid model&&&&&&&{0.56}\\
\midrule
\midrule
\textbf{Aggregating Evidence:} &\multicolumn{7}{c}{ $\boldsymbol{M^* = \{2\}}$}\\
\multicolumn{8}{c}{Step 0: $M^*_A = \{2\}$,  $M^*_B = \{2\}$, $M^*_C = \{1,2,3\}$, $ M^*_D = \{1,2,3,4\}$ \hfill Step 1: No conflicting evidence}\\
\multicolumn{8}{l}{Step 2:  Smallest MCS is $M^*_A= \{2\}$, no additional models supported by strong instruments}\\
\bottomrule
\end{widetable}

\begin{tablenotes}[flushleft]
    \setlength\labelsep{0pt}
    \footnotesize
    \item Panels A--D report the RV test statistics $T^\text{RV}$ and the effective $F$-statistic for all pairs of models, and the MCS $p$-values (details on their computation are in Appendix \ref{sect:MCSDetail}). A negative RV test statistic suggests better fit of the row model. $F$-statistics indicated with $\dagger$ are below the appropriate critical value for best-case power above 0.95.  With MCS $p$-values below 0.05 a row model is rejected from the model confidence set.   Steps in the aggregating evidence panel correspond to Figure \ref{fig:flowchart}. Both $T^\text{RV}$ and the $F$-statistics account for two-step estimation error and clustering at the market level; see Appendix \ref{sect:TwoStep} for details.  \looseness=-1
\end{tablenotes}
\end{threeparttable}
\end{table}

Panels B-D report test results in the same format as Panel A for the other three sets of instruments. Results vary markedly across panels. While the MCS in Panel B contains only model 2, coinciding with the MCS in Panel A, the MCS in Panels C and D contain additional models. Inspection of the pairwise $F$-statistics shows that the failure to reject models in Panels C and D is due to the Cost and Diff instruments having low power. For instance, the five Diff instruments in Panel D, while strong for size, are weak for testing at a target best-case power of 0.95 for all pairs of models. Given that the diagnostic is based on best-case power, the realized power could be considerably lower than 0.95. Similarly, the single rival Cost instrument in Panel C is weak for testing: for several pairs of models, the instrument is weak for power at a target of 0.50. Given the null is not rejected in these cases, power is the salient concern. \looseness=-1

The diagnostic enhances the interpretation of the RV test results in Table \ref{Tab:RV_app}.  Had the researcher precommitted to Diff or Cost instruments, the conclusions one could draw on firm conduct would not be informative. Because it is hard, in this context, to precommit to any one set of instruments, we suggest the researcher accumulates evidence across instrument sets.\footnote{Alternatively, we could pool all instrument sets. Appendix \ref{sect:Robustness} shows that doing so dilutes instrument power, resulting in lower $F$-statistics and a larger MCS.} To do so, we implement the procedure in Figure \ref{fig:flowchart}. In step 1, we check for conflicting evidence.  As all MCS for each set of instruments are nested, there is no conflicting evidence in this setting.  Thus, in step 2 we initially set $M^*=\{2\},$ which is the smallest MCS arising from NoProd and Demo instruments. As Diff IVs and Cost IVs are not strong for all pairs of models, there is no addition to be made to $M^*$. Thus, the evidence accumulated across the four sets of instruments supports concluding for model 2.\looseness=-1

\vspace{0.5em}
\noindent \textbf{Main Findings:} This application highlights the practical importance  of allowing for misspecification and degeneracy when testing conduct.  First,  by formulating hypotheses to perform model selection, RV offers interpretable results in the presence of misspecification.  Instead, AR rejects all models in our large sample. Second, instruments are weak in a standard testing environment, affecting inference.  When  RV is run with the Diff or Cost instruments, it has little to no power in this application.\footnote{However, these instruments could  be strong in other applications.} Thus, assuming at least one of the models is testable is not innocuous.  Our diagnostic distinguishes between weak and strong instruments, allowing the researcher to assess whether inference is valid.  Finally, by not having to precommit to a choice of instruments, our procedure for accumulating evidence allows researchers to draw sharp conclusions on firm conduct in this setting.\looseness=-1

In addition to illustrating our results, this application speaks to how prices are set in consumer packaged goods industries. Unlike \cite{v07} who concludes for the zero wholesale margin model, only a model where manufacturers set retail prices is supported by our testing procedure.   Our finding  is important for the broader literature studying conduct in markets for consumer packaged goods as it supports the common assumption that manufacturers set retail prices \citep[e.g.,][]{n01,mw17}.\looseness=-1

\section{Conclusion}\label{sec:conclusion}

In this paper, we discuss inference in an empirical environment encountered often by IO economists: testing models of firm conduct.  Starting from the falsifiable restriction in  \cite{bh14}, we study the effect of formulating  hypotheses and choosing instruments on inference.  Formulating hypotheses to perform model selection may allow the researcher to learn the true model of firm conduct in the presence of demand or cost misspecification.  Alternative approaches based on model assessment instead will reject the true model of conduct if noise is sufficiently low.  Given that misspecification is likely in practice, we focus on the RV test.\looseness=-1  

However, the RV test suffers from degeneracy when instruments are weak for testing.  Based on this characterization, we outline the inferential problems caused by degeneracy and provide a diagnostic.  The diagnostic relies on an $F$-statistic which is easy to compute, and can inform the researcher about the presence of size distortions or a lack of power.  We also show how to aggregate evidence across different sets of instruments, while using the $F$-statistic to draw sharp conclusions.\looseness=-1

An empirical application testing vertical models of conduct \citep{v07} highlights the importance of our results.  We find that AR rejects all models of conduct. This illustrates the advantage of adopting a model selection approach.  Four sets of exogenous and plausibly relevant instruments exist in this setting. Two of these are weak, as diagnosed by our $F$-statistic. Adopting our procedure for accumulating evidence across RV tests with separate instrument sets, we conclude for a single model in which manufacturers set retail prices.

\small{
\bibliography{biblio_RV}}

\normalsize{
\pagebreak
\begin{appendices}
\begin{center}
    \Large{\textbf{Online Appendix}}
\end{center}
\section{Proofs}\label{sect:Proofs}

As a service to the reader, we collect here the key notational conventions and give a formulaic description of the asymptotic RV variance $\sigma^2_\text{RV}$. Additionally, we present a lemma that serves as the foundation for multiple subsequent proofs. 

For any variable $\boldsymbol{y}$, we let $y = \boldsymbol y - \textbf{w}E[\textbf{w}'\textbf{w}]^{-1}E[\textbf{w}'\boldsymbol y]$,  $\hat y = \boldsymbol y - \textbf{w}(\textbf{w}'\textbf{w})^{-1}\textbf{w}'\boldsymbol y$. Predicted markups are $\Delta^z_m = z \Gamma_m$ where $\Gamma_m = E[z'z]^{-1}E[z'\Delta_{m}]$. The GMM objective functions are ${Q}_m  =  {g}_m' W {g}_m$ where ${g}_m= E[z_i(p_i-\Delta_{mi})]$ and $ W=E[z_i z_i']^{-1}$ with sample analogs of $\hat Q_m=\hat g_m'\hat W \hat g_m$ where $\hat g_m = n^{-1} \hat z'(\hat p - \hat \Delta_m)$ and $\hat W = n (\hat z'\hat z)^{-1}$. The RV test statistic is
$T^\text{RV} = {\sqrt{n}(\hat Q_1 - \hat Q_2)}/{\hat \sigma_\text{RV}}$ where
\begin{align}
	\hat \sigma^2_\text{RV} = 4\!\left[ \hat g_1'\hat W^{1/2} \hat V^\text{RV}_{11} \hat W^{1/2} \hat g_1 + \hat g_2'\hat W^{1/2} \hat V_{22}^\text{RV} \hat W^{1/2} \hat g_2- 2\hat g_1'\hat W^{1/2} \hat V_{12}^\text{RV} \hat W^{1/2} \hat g_2  \right]
\end{align}
and the variance estimators are $\hat V_{\ell k}^\text{RV} = \frac{1}{n} \sum_{i=1}^n \hat \psi_{\ell i} \hat \psi_{ki}'$ for the influence function
\begin{align}
      \hat \psi_{mi} = \hat W^{1/2} \left( \hat z_i \big(\hat p_i - \hat \Delta_{mi}\big) - \hat g_m \right)- \tfrac{1}{2} \hat W^{3/4} \left( \hat z_i \hat z_i' -\hat W^{-1}\right)\hat W^{3/4} \hat g_m.
\end{align}
The AR statistic is $T^\text{AR}_m = n\hat \pi_m' (\hat V^\text{AR}_{mm})^{-1} \hat \pi_m$ for $\hat \pi_m = \hat W \hat g_m$, $\hat V^\text{AR}_{\ell k} = \frac{1}{n} \sum_{i=1}^n \hat \phi_{\ell i} \hat \phi_{ki}'$,
\begin{align}
    \hat \phi_{mi} = \hat W \left( \hat z_i \big(\hat p_i - \hat \Delta_{mi}\big) - \hat g_m \right)- \hat W \left( \hat z_i \hat z_i' -\hat W^{-1}\right)\hat W \hat g_m.
\end{align}
Also, $\pi_m = Wg_m$. $\hat V^\text{AR}_{mm}$ is the White heteroskedasticity-robust variance estimator, since we also have $\hat \phi_{mi} = \hat W \hat z_i\big( \hat p_i - \hat \Delta_{mi} - \hat z_i'\hat \pi_m\big)$.

To state the following lemma and give a formulation of $\sigma^2_\text{RV}$, we introduce population versions of $\hat \psi_{mi}$ and $\hat \phi_{mi}$ along with notation for their variances. Let $\psi_{mi} = W^{1/2}  z_i ( p_i -  \Delta_{mi}) - \tfrac{1}{2}  W^{3/4} z_i  z_i'  W^{3/4}  g_m - \tfrac{1}{2}W^{1/2} g_m$ and $\phi_{mi} = W z_i e_{mi}$. Also, let $V^\text{RV}_{\ell k} = E[\psi_{\ell i}\psi_{ki}']$, $V^\text{AR}_{\ell k} = E[\phi_{\ell i}\phi_{ki}']$, and $V^\text{RV} = E[(\psi_{1i}',\psi_{2i}')'(\psi_{1i}',\psi_{2i}')]$, which is a matrix with $V^\text{RV}_{11}$, $V^\text{RV}_{12}$, and $V^\text{RV}_{22}$ as its entries. Finally,
\begin{align}\label{eq:sig2rv}
    \sigma^2_\text{RV} = 4\!\left[  g_1' W^{1/2}  V^\text{RV}_{11}  W^{1/2}  g_1 +  g_2' W^{1/2}  V^\text{RV}_{22}  W^{1/2}  g_2- 2 g_1' W^{1/2}  V^\text{RV}_{12}  W^{1/2} g_2 \right]\!. \,\,
\end{align}

\begin{lem}\label{lem:asymp norm}
   Suppose Assumptions \ref{as:momcond} and \ref{as:regcond} hold. For $\ell,k,m \in \{1,2\}$, we have
   \begin{align}
      (i)& \ \ \sqrt{n}\begin{pmatrix} \hat W^{1/2} \hat g_1 -  W^{1/2}  g_1  \\ \hat W^{1/2} \hat g_2 -  W^{1/2}  g_2  \end{pmatrix} \xrightarrow{d} N\big(0,V^\emph{RV}\big), & (ii)& \ \ \hat V^\emph{RV}_{\ell k} \xrightarrow{p} V^\emph{RV}_{\ell k}, \\
       (iii)& \ \  \sqrt{n}\left( \hat \pi_m - \pi _m \right) \xrightarrow{d} N\big(0,V^\emph{AR}_m\big), 
       &(iv)& \ \ \hat V_{m}^\emph{AR} \xrightarrow{p} V_{m}^\emph{AR}.
   \end{align} 
  
\end{lem}

\begin{proof}
See Appendix \ref{sec:suppl}.
\end{proof}

\begin{remark}\label{rem:validity}
    From parts $(iii)$ and $(iv)$, it immediately follows that $T^\text{AR}_m \xrightarrow{d} \chi^2_{d_z}$ under $H_{0,m}^\text{AR}$ so that the AR tests are asymptotically valid when Assumptions \ref{as:momcond} and \ref{as:regcond} hold. When Assumption \ref{as:nodegen} also holds, it follows from parts $(i)$, $(ii)$, and a first order Taylor approximation that $T^\text{RV} \xrightarrow{d} N(0,1)$ under $H_{0}^\text{RV}$ so that the RV test is asymptotically valid. Details of this step can be found in \citeA{rv02,hp11} and are omitted. When Assumption \ref{as:nodegen} fails to hold, a first order Taylor approximation does not capture the behavior of $T^\text{RV}$ as discussed in Section \ref{sec:instruments}.
\end{remark}


\noindent For the next two proofs, we rely on the following sequence of equalities:
\begin{align}
    E[(\Delta_{0i}^{z}-\Delta_{mi}^z)^2] &= E[(z_i'(\Gamma_0-\Gamma_m))^2] \label{eq:line1}\\
    &= (\Gamma_0-\Gamma_m)'E[z_i z_i'](\Gamma_0-\Gamma_m) \label{eq:line2}\\
    &= E[z_i(\Delta_{0i}-\Delta_{mi})]'E[z_i z_i']^{-1}E[z_i(\Delta_{0i}-\Delta_{mi})] \label{eq:line3}\\
    &= E[z_i(p_i-\Delta_{mi})]'E[z_iz_i']^{-1}E[z_i(p_i-\Delta_{mi})] \label{eq:line4}\\
    &= \pi_m E[z_iz_i'] \pi_m \label{eq:line5} = Q_m.
\end{align}
The first equality follows from $\Delta_{mi}^z = z_i\Gamma_m$, the third is implied by $\Gamma_m = E[z'z]^{-1}E[z'\Delta_{m}] = E[z_iz_i']^{-1} E[z_i\Delta_{mi}]$, the fourth is a consequence of $E[z_i \omega_{0i}]=0$, $W=E[z_i z_i']^{-1}$, and $\Delta_{0i} = p_i - \omega_{0i}$, and the fifth and final equalities follow from $\pi_m = W g_m$ and $g_m= E[z_i(p_i-\Delta_{mi})].$

\begin{proof}[Proof of Lemma \ref{as:falsify}.]

In our parametric framework, the falsifiable condition in Equation (28) of \cite{bh14} is\footnote{See Section 6, Case 2 in \cite{bh14} for a discussion of their non-parametric environment.}
\begin{align}
    E[ \boldsymbol{p}_i-\boldsymbol{\Delta}_{mi} \mid z_i,\textbf{w}_i ] = \textbf{w}_i \tau + E[ \omega_{0i} \mid z_i,\textbf{w}_i ] \quad \text{a.s.}
\end{align}
This equation, together with residualization against $\textbf{w}_i$ and integration over the distribution of  $\textbf{w}_i$, then implies that
\begin{align}
    E[ {p}_i-{\Delta}_{mi} \mid z_i ] = E[ \omega_{0i} \mid z_i] \quad \text{a.s.}
\end{align}
Integrating the above against $z_i$, using the assumption $E[z_i\omega_{0i} ]=0$, and relying on the law of iterated expectations, then yields
\begin{align}
    E[z_i(p_i-\Delta_{mi})] = 0
\end{align}
Finally, we note the equivalence
\begin{align}\label{eq:equivalence}
    E[(\Delta_{0i}^{z}-\Delta_{mi}^z)^2] = 0 \quad  &\Leftrightarrow \quad E[z_i(p_i-\Delta_{mi})] = 0 
\end{align}
which follows from \eqref{eq:line1}, \eqref{eq:line4}, and $E[z_i z_i']$ being positive definite. Thus we have shown that Equation (28) of \cite{bh14} implies \eqref{eq:testability}.
%
%
%
\end{proof}

\begin{proof}[Proof of Proposition \ref{prop:nulls}]
For $(i)$, we need to show the equivalence
\begin{align}\label{eq:equiv_prop1_ar}
\pi_m = 0 \qquad \Leftrightarrow \qquad   E[(\Delta^z_{0i}-\Delta^z_{mi})^2] = 0  
\end{align}
This equivalence is a consequence of Equations \eqref{eq:line1} and \eqref{eq:line5} in addition to $E[z_i z_i']$ being positive definite.
For $(ii)$, we we need to show the equivalence
\begin{align}
    {Q}_1-{Q}_2 =0  \qquad \Leftrightarrow \qquad  E[(\Delta^z_{0i}-\Delta^z_{1i})^2]-E[(\Delta^z_{0i}-\Delta^z_{2i})^2]=0. 
\end{align}
This equivalence is a consequence of Equations \eqref{eq:line1} and \eqref{eq:line5}.
\end{proof}

\begin{proof}[Proof of Proposition \ref{prop:distrib_TE2}.]
For $(i)$, we use Equations \eqref{eq:line2} and \eqref{eq:line5} in addition to the definition of the local alternative in Equation \eqref{eq:local selection} to write
\begin{align}
\sqrt{n}\left( Q_1 -Q_2 \right) &= \sqrt{n}\left( (\Gamma_0-\Gamma_1) - (\Gamma_0-\Gamma_2) \right)\!' E[z_iz_i']\!\left( (\Gamma_0-\Gamma_1) + (\Gamma_0-\Gamma_2) \right)\! \\
&= q'E[z_iz_i']\!\left( (\Gamma_0-\Gamma_1) + (\Gamma_0-\Gamma_2) \right)\!.
\end{align}
Assumption \ref{as:regcond}, part (iii), implies that $(\Gamma_0-\Gamma_1) + (\Gamma_0-\Gamma_2) $ is bounded. We therefore assume essentially without loss of generality that $\sqrt{n}\left( Q_1 -Q_2 \right)$ is a constant, say $c$.\footnote{This assumption is essentially without loss of generality since the boundedness of $(\Gamma_0-\Gamma_1) + (\Gamma_0-\Gamma_2)$, and therefore of $\sqrt{n}\left( Q_1 -Q_2 \right)$, allows us to otherwise argue along subsequences where $\sqrt{n}\left( Q_1 -Q_2 \right)$ converges to a constant.} From Equations \eqref{eq:line1} and \eqref{eq:line5}, we can also write 
\begin{align}
    c &= \sqrt{n} \!\left( E\!\left[(\Delta_{0i}^{z}-\Delta_{1i}^z)^2\right] - E\!\left[(\Delta_{0i}^{z}-\Delta_{2i}^z)^2\right] \right) \\
    &= E\!\big[(\Delta_{0i}^{\text{RV},z}-\Delta_{1i}^{\text{RV},z})^2\big] - E\!\big[(\Delta_{0i}^{\text{RV},z}-\Delta_{2i}^{\text{RV},z})^2\big].
\end{align}

As in Remark \ref{rem:validity}, a first order Taylor expansion, Lemma \ref{lem:asymp norm}, parts $(i)$ and $(ii)$, together with consistency of $\hat \sigma^2_\text{RV}$ now leads to
\begin{align}
    T^\text{RV} = \frac{c}{\sigma_\text{RV}} + \frac{g_1'W^{1/2} \hat W^{1/2}\hat g_1 - g_2'W^{1/2} \hat W^{1/2}\hat g_2}{\sigma_\text{RV}} + o_p(1) \xrightarrow{d} N(c,1).
\end{align}

For $(ii)$, we first note that Assumption \ref{as:homsked} implies that
\begin{align}
    V^\text{AR}_{mm} = E[z_iz_i']\inverse E[e_{mi}^2 z_iz_i'] E[z_iz_i']\inverse = \sigma_m^2 E[z_iz_i'].
\end{align}
Thus we have from Equations \eqref{eq:line2} and \eqref{eq:line5} that 
\begin{align}
    \sigma_m^2 n\pi_m'\big( V^\text{AR}_{mm}\big)\inverse\pi_m = n(\Gamma_0-\Gamma_m)'E[z_i z_i'](\Gamma_0-\Gamma_m) = q_m'E[z_i z_i']q_m
\end{align}
which in turn yields that $n\pi_m'\big( V^\text{AR}_{mm}\big)\inverse\pi_m = E\big[(\Delta^{\text{AR},z}_{0i}-\Delta^{\text{AR},z}_{mi}) ^2\big]/\sigma_m^2$ is a constant, say $c_m$. From Lemma \ref{lem:asymp norm}, parts $(iii)$ and $(iv)$, and continuous mapping we then have
\begin{align}
    T_m^\text{AR} &= n \hat \pi_m'\big( \hat V^\text{AR}_{mm}\big)\inverse \hat \pi_m = n \hat \pi_m'\big( V^\text{AR}_{mm}\big)\inverse \hat \pi_m + o_p(1) \xrightarrow{d} \chi^2_{d_z}(c_m). \qedhere
\end{align}
\end{proof}

\begin{proof}[Proof of Proposition \ref{prop:degen} and Corollary \ref{cor:degen}]

From Equations \eqref{eq:line1} and \eqref{eq:line2}, we have equivalence between $E[(\Delta^z_{0i}-\Delta^z_{mi})^2] = 0$ and $\Gamma_0-\Gamma_m=0$. Furthermore, we recall that $\pi_m=\Gamma_0-\Gamma_m$. Thus, we only need to show that $\sigma^2_\text{RV}=0$ if and only if $\pi_m=0$ for all $m\in\{1,2\}$. Rewriting Equation \eqref{eq:sig2rv} in matrix notation, we have
\begin{align}
    \sigma^2_\text{RV} = 4 
    \begin{pmatrix} W^{-1/2}\pi_1 \\ W^{-1/2}\pi_2 \end{pmatrix}'
    \begin{bmatrix}    V^\text{RV}_{11} & - V^\text{RV}_{12} \\ - V^\text{RV}_{12} & V^\text{RV}_{22}\end{bmatrix}
    \begin{pmatrix} W^{-1/2}\pi_1 \\ W^{-1/2}\pi_2 \end{pmatrix}.
\end{align}
Therefore, the claims to be proven follow from positive definiteness of the variance matrix $V^\text{RV} = E[(\psi_{1i}',\psi_{2i}')'(\psi_{1i}',\psi_{2i}')]$, which is the matrix with $V^\text{RV}_{11}$, $V^\text{RV}_{12}$, and $V^\text{RV}_{22}$ as its entries.

To show that $V^\text{RV}$ is positive definite, we take an arbitrary non-zero, non-random vector $v=(v_1',v_2')' \in \mathbb{R}^{2d_z}$. $V^\text{RV}$ is positive definite if $E(v_1'\psi_{1i}+v_2'\psi_{2i})^2 > 0$ for any such $v$. For certain implied non-random $u = (u_1',u_2')'$ and $t=(t_1',t_2')'$ where $t$ is non-zero, we have $E(v_1'\psi_{1i}+v_2'\psi_{2i})^2 = E(u_1'z_i \cdot u_2'z_i + t_1'z_i \cdot e_{1i} + t_2'z_i \cdot e_{2i} )^2$. Because $\sigma_{12}^2 < \sigma_1^2\sigma_2^2$ (Assumption \ref{as:homsked}), we have that $E(t_1'z_i \cdot e_{1i} + t_2'z_i \cdot e_{2i} )^2 > 0$. Therefore, we can only have $E(v_1'\psi_{1i}+v_2'\psi_{2i})^2 = 0$ if $e_{1i}$ or $e_{2i}$ is a linear function of $z_i$ almost surely. However, because $E[z_ie_{mi}]=0$, such dependence is ruled out by $\sigma_m^2>0$ (Assumption \ref{as:homsked}).
\end{proof}

\begin{proof}[Proof of Proposition \ref{prop:wkinstrF}.]
The proof proceeds in three steps. Step (1) provides a function of the data $(\tilde \Psi_-',\tilde \Psi_+')'$ and two constants $\mu_-$, $\mu_+$ such that $(\tilde \Psi_-',\tilde \Psi_+')' \xrightarrow{d} (\Psi_-',  \Psi_+')'$. All of $(\tilde \Psi_-',\tilde \Psi_+')'$, $\mu_-$, and $\mu_+$ are also functions of the DGP. Step (2) establishes parts $(ii)$--$(iv)$. Step (3) shows that $\abs{T^\text{RV}} = \abs{\tilde \Psi_-'\tilde \Psi_+}/\big(\norm{\tilde \Psi_-}^2 + \norm{\tilde \Psi_+}^2 + 2 \rho \tilde \Psi_-'\tilde \Psi_+ + o_p(1) \big)^{1/2}$ and $F=\big( \norm{\tilde \Psi_-}^2 + \norm{\tilde \Psi_+}^2 - 2\rho \tilde \Psi_-'\tilde \Psi_+ \big)/\big(2d_z \big) + o_p(1)$. Part $(i)$ follows from continuous mapping, step (1), and step (3).

\noindent \textbf{Step (1)} We first provide definitions of $\mu_-$ and $\mu_+$. Let $\tau_+ = 2(\sigma_1^2 + \sigma_2^2 + 2\sigma_{12})^{1/2}$ and $\tau_- = 2(\sigma_1^2 + \sigma_2^2 - 2\sigma_{12})^{1/2}$ and use these two positive constants to define
\begin{align}
		\begin{pmatrix}
			 \mu_1 \\  \mu_2
		\end{pmatrix}  =
		\left(\frac{1}{{\tau_-}} + \frac{1}{{\tau_+}}\right)\! \begin{pmatrix} \sqrt{n} W^{1/2} g_1    \\	 	-\sqrt{n} W^{1/2}  g_2		\end{pmatrix}
		+
		\left(\frac{1}{{\tau_-}} - \frac{1}{{\tau_+}}\right)\!\begin{pmatrix}	-\sqrt{n} W^{1/2}  g_2  \\ 		\sqrt{n} W^{1/2}  g_1 		\end{pmatrix}\!.
\end{align}
Defining $\kappa = 1 + \mathbf{1}\{ \norm{\mu_1} \le \norm{\mu_2}  \}$, we then let $\mu_- = \norm{\mu_\kappa} - \norm{\mu_{3-\kappa}}$ and $\mu_+ = \norm{\mu_1} + \norm{\mu_2}$. Since the instruments are weak for testing, we have that $\sqrt{n} W^{1/2} g_m = E[z_i z_i']^{1/2} q_m$ so $\mu_-$ and $\mu_+$ do not depend on $n$.

To introduce $\tilde \Psi_-$ and $\tilde \Psi_+$, we let ${\cal Q}_m \in \mathbb{R}^{d_z \times d_z}$ be a non-random orthogonal matrix $({\cal Q}_m{\cal Q}_m'={\cal Q}_m'{\cal Q}_m=I_{d_z})$ such that ${\cal Q}_m \mu_m = \norm{\mu_m} \boldsymbol e_1$. With these matrices in hand, we define $\tilde \Psi_- = \tilde \mu_\kappa - \tilde \mu_{3-\kappa}$ and $\tilde \Psi_+ = \tilde \mu_1 + \tilde \mu_2$ where
\begin{align}
		\begin{pmatrix}
			\tilde \mu_1 \\ \tilde \mu_2
		\end{pmatrix} = 
		\left(\frac{1}{{\tau_-}} + \frac{1}{{\tau_+}}\right)\! \begin{pmatrix} \sqrt{n} {\cal Q}_1\hat W^{1/2} \hat g_1    \\	 	-\sqrt{n} {\cal Q}_2\hat W^{1/2}  \hat g_2		\end{pmatrix}
		+
		\left(\frac{1}{{\tau_-}} - \frac{1}{{\tau_+}}\right)\!\begin{pmatrix}	-\sqrt{n} {\cal Q}_1\hat W^{1/2}  \hat g_2  \\ 		\sqrt{n} {\cal Q}_2\hat W^{1/2}  \hat g_1 		\end{pmatrix}\!.
\end{align}

The preceding definitions imply that $(\tilde \Psi_-',\tilde \Psi_+')' = \sqrt{n} A (\hat g_1'\hat W^{1/2},\hat g_2'\hat W^{1/2} )'$ for a particular non-random matrix $A \in \mathbb{R}^{2d_z \times 2d_z}$ with $\sqrt{n} A ( g_1' W^{1/2}, g_2' W^{1/2} )' = (\mu_-\boldsymbol e_1',\mu_+\boldsymbol e_1')'$. Since $\mu_-$ and $\mu_+$ do not depend on $n$, it therefore follows from Lemma \ref{lem:asymp norm}, part $(i)$, that $(\tilde \Psi_-',\tilde \Psi_+')' \xrightarrow{d} (\Psi_-',  \Psi_+')'$ provided that $AV^\text{RV}A' = {\tiny \begin{bmatrix} 1 & \rho \\ \rho & 1 \end{bmatrix} }\otimes I_{d_z}$ where $\rho$ is the correlation between $e_{\kappa i}-e_{3-\kappa, i}$ and $e_{1i}+e_{2i}$. To see why this convergence occurs, note first that weak instruments for testing (Assumption \ref{as:wkinstr}) implies that $g_m = O(n^{-1/2})$ for $m=1,2$, which in turn yields that the second part of $\psi_{mi}$ is $O_p(n^{-1/2})$. From Assumption \ref{as:homsked}, we then have
\begin{align}\label{eq:varconv}
		V^\text{RV}_{\ell k} = W^{1/2} E[ e_{\ell i} e_{k i} z_iz_i'] W^{1/2} + O(n^{-1/2}) = \sigma_{\ell k} I_{d_z} + O(n^{-1/2})
\end{align}
where we write $\sigma_{mm}$ for $\sigma_m^2$. Thus, we have that $V^\text{RV} = {\tiny \begin{bmatrix} \sigma_1^2 & \sigma_{12} \\ \sigma_{12} & \sigma_{2}^2 \end{bmatrix} }\otimes I_{d_z} .$ Furthermore, we note that $A$ takes the form
\begin{align}
			A = 
			\begin{bmatrix}	(-1)^{3-\kappa}I & (-1)^\kappa I \\ I & I	\end{bmatrix}
			\begin{bmatrix}	{\cal Q}_1 & 0 \\ 0 & {\cal Q}_2	\end{bmatrix}
			\begin{bmatrix}	I & I \\ I & -I	\end{bmatrix}
			\begin{bmatrix}	I\tau_-^{-1} & 0 \\ 0 & I\tau_+^{-1}	\end{bmatrix}
			\begin{bmatrix}	I & -I \\ I & I	\end{bmatrix}
\end{align}
and leave the verification of $A\big( {\tiny \begin{bmatrix} \sigma_1^2 & \sigma_{12} \\ \sigma_{12} & \sigma_{2}^2 \end{bmatrix} }\otimes I_{d_z} \big)A' = {\tiny \begin{bmatrix} 1 & \rho \\ \rho & 1 \end{bmatrix} }\otimes I_{d_z}$ to the reader.

\noindent \textbf{Step (2)} Part $(iv)$ is an immediate implication of the triangle inequality and the definitions $\mu_- = \norm{\mu_\kappa} - \norm{\mu_{3-\kappa}}$ and $\mu_+ = \norm{\mu_1} + \norm{\mu_2}$. For part $(ii)$, we have that $\mu_-=0$ if and only if $\norm{\mu_1}^2 - \norm{\mu_2}^2 =0$. In turn, we have that
\begin{align}
		\norm{\mu_1}^2 - \norm{\mu_2}^2 &= n\!\left( \!\left(\tau_-^{-1} + \tau_+^{-1} \right)^2  - \!\left(\tau_-^{-1} - \tau_+^{-1} \right)^2 \right)\! \left( Q_1 - Q_2\right) \\
		&=4n (\tau_+ \tau_-)^{-1} \left( Q_1 - Q_2\right),
\end{align}
from which part $(ii)$ is immediate. For part $(iii)$, we have that $\mu_+=0$ if and only if $\norm{\mu_1}^2 + \norm{\mu_2}^2 =0$. We now have,
\begin{align}
			\norm{\mu_1}^2 \!+\! \norm{\mu_2}^2
			&\!=\! n\!\begin{pmatrix} W^{1/2}g_1  \\ W^{1/2}g_2	\end{pmatrix}'\!
			A'A 
			\!\begin{pmatrix} W^{1/2}g_1  \\ W^{1/2}g_2	\end{pmatrix} \!
			=\! n\!\begin{pmatrix} W^{-1/2}\pi_1  \\ W^{-1/2}\pi_2	\end{pmatrix}'\!
			A'A 
			\!\begin{pmatrix} W^{-1/2}\pi_1  \\ W^{-1/2}\pi_2	\end{pmatrix} \!
\end{align}
so part $(iii)$ follows from the positive definiteness of both $W$ and the matrix $A'A = 4 \!{\tiny \begin{bmatrix} \tau_+^{-2} + \tau_-^{-2} & \tau_+^{-2} - \tau_-^{-2} \\ \tau_+^{-2} - \tau_-^{-2} & \tau_+^{-2} + \tau_-^{-2} \end{bmatrix} }\!\otimes I_{d_z}$.

\noindent \textbf{Step (3)} For $T^\text{RV}$ we first consider the numerator. Here, we observe that
\begin{align}
		\tilde \Psi_-' \tilde\Psi_+ = \tilde \mu_\kappa'\tilde \mu_\kappa - \tilde \mu_{3-\kappa}'\tilde \mu_{3-\kappa} = 4n (\tau_+ \tau_-)^{-1} \!\left( \hat Q_\kappa - \hat Q_{3-\kappa}\right)
\end{align}
so that $\sqrt{n}\abs{\hat Q_1-\hat Q_2} = (\tau_+\tau_-/4)\abs{\tilde\Psi_-' \tilde\Psi_+}/\sqrt{n}$. For the denominator, we initially note that Lemma \ref{lem:asymp norm}, part $(ii)$, and Equation \eqref{eq:varconv} yields that
\begin{align}
		(\tau_+ \tau_-/4)^{-2} n \hat \sigma^2_\text{RV} = \tfrac{4^3n}{\tau_+^2\tau_-^2}\!\left[ \sigma_1^2 \hat Q_1 + \sigma_2^2 \hat Q_2 - 2\sigma_{12} \hat g_1' \hat W \hat g_2  \right] + o_p(1).
\end{align}	
Similarly, we can calculate that
\begin{align}
		\norm{\tilde \Psi_-}^2  &+ \norm{\tilde \Psi_+}^2  + 2\rho \tilde \Psi_-'\tilde \Psi_+ 
		= n \begin{pmatrix}  \hat W^{1/2}\hat g_1 \\ \hat W^{1/2} \hat g_2		\end{pmatrix}
		'A' \bigg( {\tiny \begin{bmatrix} 1 & \rho \\ \rho & 1 \end{bmatrix} }\otimes  I_{d_z}  \bigg)A
		\begin{pmatrix}  \hat W^{1/2}\hat g_1 \\ \hat W^{1/2} \hat g_2		\end{pmatrix}\\
		&= {4^2n}{} \begin{pmatrix}  \hat W^{1/2}\hat g_1 \\ \hat W^{1/2} \hat g_2		\end{pmatrix}'
		\bigg( \! \tfrac{4}{\tau_+^2\tau_-^2}{\tiny \begin{bmatrix} \sigma_1^2 \! & \! -\sigma_{12} \\ -\sigma_{12} \! & \! \sigma_2^2 \end{bmatrix}} \!\otimes I_{d_z}  \bigg)
		\begin{pmatrix}  \hat W^{1/2}\hat g_1 \\ \hat W^{1/2} \hat g_2		\end{pmatrix} \\
		&= \tfrac{4^3n}{\tau_+^2\tau_-^2}\!\left[ \sigma_1^2 \hat Q_1 + \sigma_2^2 \hat Q_2 - 2\sigma_{12} \hat g_1' \hat W \hat g_2  \right]
\end{align}
where the second equality follows from the last sentences of steps (1) and (2) together with
\begin{align}
		{\tiny \begin{bmatrix} \tau_+^{-2} + \tau_-^{-2} & \tau_+^{-2} - \tau_-^{-2} \\ \tau_+^{-2} - \tau_-^{-2} & \tau_+^{-2} + \tau_-^{-2} \end{bmatrix} \begin{bmatrix} \sigma_1^2 & \sigma_{12} \\ \sigma_{12} & \sigma_{2}^2 \end{bmatrix}  \begin{bmatrix} \tau_+^{-2} + \tau_-^{-2} & \tau_+^{-2} - \tau_-^{-2} \\ \tau_+^{-2} - \tau_-^{-2} & \tau_+^{-2} + \tau_-^{-2} \end{bmatrix} = \begin{bmatrix} \sigma_1^2 & -\sigma_{12} \\ -\sigma_{12} & \sigma_2^2 \end{bmatrix} } \tfrac{4}{\tau_+^2\tau_-^2}.
\end{align}
Thus we have the desired conclusion
		\begin{align}
		\abs{T^\text{RV}} &= \sqrt{n}\frac{\abs{\hat Q_1 - \hat Q_2}}{\hat \sigma_\text{RV}} = \frac{4n (\tau_+ \tau_-)^{-1/2} \!\abs{ \hat Q_1 - \hat Q_2}}{\left( (\tau_+ \tau_-/4^2)\inverse n \hat \sigma^2_\text{RV}  \right)^{1/2}} \\
		&=
		 \abs{\tilde \Psi_-'\tilde \Psi_+}/\big(\norm{\tilde \Psi_-}^2 + \norm{\tilde \Psi_+}^2 + 2\kappa \tilde \Psi_-'\tilde \Psi_+  + o_p(1) \big)^{1/2}
\end{align}

For $F$, we first note that standard arguments lead to
\begin{align}
		2d_zF &=n(1-\hat \rho^2)\frac{\hat \sigma_2^2 \hat g_1' \hat W \hat g_1 + \hat \sigma_1^2 \hat g_2' \hat W \hat g_2 - 2\hat \sigma_{12} \hat g_1' \hat W \hat g_2}{\hat \sigma_1^2 \hat \sigma_2^2 - \hat \sigma_{12}^2 } \\
		&=n(1-\rho^2)\frac{\sigma_2^2 \hat Q_1 +  \sigma_1^2 \hat Q_2 - 2 \sigma_{12} \hat g_1' \hat W \hat g_2}{ \sigma_1^2  \sigma_2^2 -  \sigma_{12}^2 } + o_p(1).
\end{align}
Similarly, we can calculate that
\begin{align}
		\norm{\tilde \Psi_-}^2  &+ \norm{\tilde \Psi_+}^2  - 2\rho \tilde \Psi_-'\tilde \Psi_+ 
		= n \begin{pmatrix}  \hat W^{1/2}\hat g_1 \\ \hat W^{1/2} \hat g_2		\end{pmatrix}'
		A' \bigg( {\tiny \begin{bmatrix} 1 & -\rho \\ -\rho & 1 \end{bmatrix} }\otimes  I_{d_z}   \bigg)A
		\begin{pmatrix}  \hat W^{1/2}\hat g_1 \\ \hat W^{1/2} \hat g_2		\end{pmatrix} \\
		&= n(1-\rho^2) \begin{pmatrix}  \hat W^{1/2}\hat g_1 \\ \hat W^{1/2} \hat g_2		\end{pmatrix}'
		\bigg(  {\tiny \begin{bmatrix} \sigma_1^2  &  \sigma_{12} \\ \sigma_{12}  &  \sigma_2^2 \end{bmatrix}}\inverse \otimes I_{d_z}  \bigg)
		\begin{pmatrix}  \hat W^{1/2}\hat g_1 \\ \hat W^{1/2} \hat g_2		\end{pmatrix} \\
		&= n(1-\rho^2)\frac{\sigma_2^2 \hat Q_1 +  \sigma_1^2 \hat Q_2 - 2 \sigma_{12} \hat g_1' \hat W \hat g_2}{ \sigma_1^2  \sigma_2^2 -  \sigma_{12}^2 }
\end{align}
which follows from inverting the equality in the last sentence of step (1).
\end{proof}

\section{Testing with Non-constant Marginal Cost}\label{sect:AltCost}
 In this appendix we discuss how the results of the paper are preserved in a more general setting where marginal cost depends on quantity sold. While in the paper, we derive results assuming constant marginal cost, this is an important extension. It is well known going back to \citeA{b82} and \citeA{l82} that the requirements for testing conduct are greater when marginal cost is non-constant.  Thus, the reader may wonder if the results of our paper would be qualitatively different in a more general setting.

Let $\boldsymbol{q}_i$ denote quantity, and consider a separable cost function:
\begin{align}
    \boldsymbol{c}_i = \bar{\boldsymbol{c}}(\boldsymbol{q}_i,\textbf{w}_i;\tau) + \omega_i,
\end{align}
where $\bar c$ is specified up to some cost parameters $\tau$.  The specification of marginal cost that we adopt in the paper is a special case where $\bar{\boldsymbol{c}}(\boldsymbol{q}_i,\textbf{w}_i;\tau) = \textbf{w}_i \tau$.

There are two additional cases to consider. The first concerns the researcher knowing the true value of $\tau$. We can define $\bar{\boldsymbol{\Delta}}_{mi} = \boldsymbol{\Delta}_{mi} + \bar{\boldsymbol{c}}(\boldsymbol{q}_i,\textbf{w}_i;\tau)$: this term is pinned down by a model of conduct $m$, and a cost function $\bar{\boldsymbol{c}}$. In this case, we can test alternative pairs of models of conduct and cost functions with the methods described in the paper. In particular, that can be done by replacing ${{\Delta}}_{m}$ in the paper with $\bar{\boldsymbol{\Delta}}_{m}$. The set of instruments $\boldsymbol{z}$ must include $\textbf{w}$ in this case, but since $\bar{\boldsymbol{c}}$ is fully specified, there is no additional requirement on instruments. 

Alternatively, when $\tau$ is unknown to the researcher, she can estimate it under a model of conduct $m$ either as a preliminary step, or simultaneously with testing, as she would in the case of constant marginal cost. However, excluded instruments are now needed to estimate $\tau$, and thus must be sufficient to identify $\tau$ under the true model of conduct. If a researcher pursues a sequential approach, the researcher can construct $\boldsymbol{\omega}_{mi} = \boldsymbol{p}_i - \boldsymbol{\Delta}_{mi} - \bar{\boldsymbol{c}}(\boldsymbol{q}_i,\textbf{w}_i;{\tau}_m)$ and perform testing after having estimated ${\tau}_m$ under each model. 

The results in the paper are still applicable as long as the researcher adjusts the standard errors of the RV test statistic and the effective $F$-statistic. As there may exist models, falsified by a set of instruments under constant marginal cost, which are not falsified by those same instruments with a more flexible marginal cost specification, degeneracy is more likely to occur when cost is non-constant. This is akin to the classic example of demand shifters in \cite{b82}. More specifically, Corollary 3 in \cite{mqsw22} shows that falsification does not necessarily require rotators, but rather two sets of economically distinct variables -- intuitively, after using some exogenous variation to pin down the slope of marginal cost implied by a given model, we need to have some residual exogenous variation to falsify models of conduct. For instance, two cost shifters of the same rivals will not be enough, but cost shifters of two different firms may be sufficient to falsify models of conduct. While this is an additional requirement as compared to falsification in the constant marginal cost case, conditional on having this additional variation, our discussion of falsification in the constant marginal cost case goes through. As this setup places additional requirements on the instruments, evaluating the quality of the inference on conduct using our diagnostic is even more important in this setting. However, beyond the extra burden, the testing problem is fundamentally unchanged, therefore the results in this paper apply.   

This discussion makes it explicit that testing firm conduct also jointly tests models of marginal cost. In fact,  $\bar{\boldsymbol{\Delta}}_{m} $ generalizes the term $\Breve{\boldsymbol{\Delta}}_{m}$ defined in Appendix \ref{sect:MCMisp} below. In that Appendix, we show that cost misspecification can be incorporated in  $\Breve{\boldsymbol{\Delta}}_{m}$, and thus be understood as markup misspecification. Here, misspecification of $\bar c$ is manifested as misspecification of  $\bar{\boldsymbol{\Delta}}_{m} $. If one is flexible in specifying $\bar c$, misspecification of $\bar{\boldsymbol{\Delta}}_{m} $ largely concerns misspecification of ${\boldsymbol{\Delta}}_{m} $, and therefore conduct. Finally, this formulation shows that the methods described in the paper can be used to test models of cost, even when conduct is known.\looseness=-1

\section{Standard Error Adjustments}\label{sect:TwoStep}
 
This appendix extends all our previously introduced statistics to take into account uncertainty stemming from preliminary demand estimation as well as dependence across observations. We suppose that $\Delta_m$ is a function of demand parameters $\theta^D_0$ that are estimated using a GMM estimator $\hat \theta^D$. We therefore let $W^D$ denote the GMM weight matrix and $h(\theta^D)=\frac{1}{n} \sum_{i=1}^n h_i(\theta^D)$ the GMM sample moment function used. Furthermore, we let $ H=\nabla_\theta h(\hat\theta^D)$ be the gradient of the sample moment function $ h$ and let $ G_m=-\frac{1}{n} \hat z' \nabla_{\theta} \hat \Delta_m(\hat\theta^D)$ be the gradient of $\hat g_m$. Both gradients are with respect to $\theta^D$.

\noindent\textbf{RV test:}
The RV statistic with a two-step adjustment replaces $\hat V_{\ell k}^\text{RV}$ in the definition of $\hat\sigma^2_\text{RV}$ with $\tilde V_{\ell k}^\text{RV} = \frac{1}{n} \sum_{i=1}^n \tilde \psi_{\ell i} \tilde \psi_{k i}'$. Here, the influence function $\tilde \psi_{mi}$ adjusts $\hat \psi_{mi}$ to account for preliminary demand estimation:
\begin{align}
      \tilde \psi_{mi} &= \hat \psi_{mi} - \hat W^{1/2} G_m \Phi \!\left(h_i(\hat\theta^D)- h(\hat\theta^D)\right)
\end{align}
where $\Phi=(H'W^DH)^{-1}H'W^D$. This is a standard adjustment for first-step estimation based on the asymptotic approximation $\hat \theta^D - \theta^D_0 = -\Phi h(\theta^D_0) + o_p\big(n^{-1/2}\big).$

\noindent\textbf{AR test:}
Analogously to the above, the AR statistics with a two-step adjustment replaces $\hat V_{mm}^\text{AR}$ with $\tilde V_{mm}^\text{AR}$ where $\tilde V_{\ell k}^\text{AR} = \frac{1}{n} \sum_{i=1}^n \tilde \phi_{\ell i} \tilde \phi_{k i}'$. Here, the influence function $\tilde \phi_{mi}$ adjusts $\hat \phi_{mi} = \hat W \hat z_i(\hat p_i- \hat \Delta_{mi} - \hat z_i'\hat \pi_i)$ to account for preliminary demand estimation using the same approximation to $\hat \theta^D$ as above:
\begin{align}
      \tilde \phi_{mi} &= \hat \phi_{mi} - \hat W G_m \Phi_m \!\left(h_i(\hat\theta^D)- h(\hat\theta^D)\right).
\end{align}



\noindent\textbf{$F$-statistic:}
The $F$-statistic with a two-step adjustment replaces $\hat \sigma_1^2$, $\hat \sigma_2^2$, and $\hat \sigma_{12}$ in the definition of $F$ and $\hat \rho^2$ with $\tilde \sigma_m^2 = d_z\inverse \text{trace}\big( \tilde V_{mm}^\text{AR} \hat W^{-1} \big)$ for $m\in\{1,2\}$ and $\tilde \sigma_{12} = d_z\inverse  \text{trace}\big(  \tilde V_{\ell k}^\text{AR}\hat W^{-1} )$. Here, $ \tilde V_{\ell k}^\text{AR}$ and $\tilde \phi_{mi}$ were introduced in the preceding paragraph.

An extension of Proposition \ref{prop:wkinstrF} that accounts for two-step estimation can be established under homoskedasticity, i.e., when $\hat W^{-1/2} \tilde V_{\ell k}^\text{AR} \hat W^{-1/2}$ for $\ell,k \in \{1,2\}$ converge in probability to diagonal matrices. In the absence of homoskedasticity, $F$ is still informative about the strength of the instruments, but the exact thresholds for size control reported in Table \ref{tab:Tab_StockYogo_Combined_new} may only be approximations to the true thresholds.



\noindent\textbf{Dependence:}
Dependent data, e.g., cluster sampling, is easily accommodated by adjustments to $\hat V_{\ell k}^\text{RV}$ and $\hat V_{\ell k}^\text{AR}$. If we let $c_{ij}$ take the value one if observations $i$ and $j$ are deemed dependent and zero otherwise, then the variance estimators used in the paper can be replaced by 
\begin{align}
    \Breve V_{\ell k}^\text{RV} = \frac{1}{n} \sum_{i=1}^n\sum_{j=1}^n c_{ij} \hat\psi_{\ell i} \hat\psi_{k j}'
    \quad \text{and} \quad 
    \Breve V_{\ell k}^\text{AR} = \frac{1}{n} \sum_{i=1}^n\sum_{j=1}^n c_{ij} \hat\phi_{\ell i} \hat\phi_{k j}'.
\end{align}
Combinations of two-step estimation and dependence are also handled by simply replacing $\hat \psi$ and $\hat \phi$ by $\tilde \psi$ and $\tilde \phi$ in the definitions of $\Breve V_{\ell k}^\text{RV}$ and $\Breve V_{\ell k}^\text{AR}$. Provided that suitable central limit theorem and laws of large numbers can be alluded to under the type of dependence considered, any claims on asymptotic validity under strong instruments made in the paper continue to hold. For clustered data, we refer to \citeA{ha19} for such results.



\section{Other Model Assessment Tests}\label{sect:EBCox}
In this section, we discuss the two other model assessment procedures used in the empirical IO literature. Although the details of test performance differ across the three model assessment procedures, EB and Cox share with AR the undesirable property that inference on conduct is not valid under misspecification.\looseness=-1

 \noindent \textbf{Estimation Based Test (EB):} A test of Equation \eqref{eq:testability} can be constructed  by viewing the problem as one of inference about a regression parameter. We refer to this approach as estimation based, or EB.  One way to implement an estimation based approach, proposed in \citeA{p17}, is to consider the equation\footnote{Alternatively, the procedure could be based on the regression $p =  \Delta \theta + \omega$,
 where ${\Delta} = [{\Delta}_{1}, {\Delta}_{2}]$ is a $n$-by-$2$ vector of the implied markups for each of the two models. The analysis of this procedure is substantively identical, except for the fact that this procedure requires at least two valid instruments.}  \begin{align}\label{eq:test}
     p & = \Delta_m \theta_m + \omega_m,
 \end{align} 
 For each model $m$, the null and alternative hypotheses for model assessment are
 \begin{align}
 H_{0,m}^\text{EB} \ : \ \theta_m = 1 \qquad \text{and} \qquad H_{A,m}^\text{EB} \ : \ \theta_m \neq 1.
 \end{align}
 Note also that under the null we have that $\omega_m=\omega_0$ so that $E[z_i\omega_{mi}]=0$.\looseness=-1

 With this formulation, a natural testing procedure to consider is then based on the Wald statistic:
 \begin{align}
     T^\text{EB}_m &= (\hat \theta_m - 1)' \hat V_{\hat \theta_m}^{-1} (\hat \theta_m - 1)
 \end{align}
 where $\hat \theta_m$ is the 2SLS estimator applied to the sample counterpart of Equation \eqref{eq:test} and $\hat V_{\hat \theta_m}$ is a consistent estimator of the variance of $\hat \theta_m$.  The asymptotic null distribution of $T_m^\text{EB}$ is a $\chi^2_{1}$ distribution and the EB test at level $\alpha$ therefore rejects if $T^\text{EB}_m$ exceeds the $\alpha$-th quantile of that null distribution. 

The EB test is similar to AR. We can, in general, show that if markups are misspecified, EB rejects the true model of conduct. To see this, note that
 \begin{align}
     \plim n^{-1}T^\text{EB}_m &= (\theta_m - 1)'  V_{\theta_m}^{-1} ( \theta_m - 1)
 \end{align}
 where $\theta_m = \plim\hat\theta_m$ is given as:    
 \begin{align}
     \theta_m = 1 + E[\Delta^z_{mi}\Delta^z_{mi}]^{-1}E[\Delta^z_{mi}(\Delta^z_{0i}-\Delta^z_{mi})]
 \end{align}
Since $V^{-1}_{\theta_m}$ is strictly positive, $\plim n^{-1}T^\text{EB} = 0$ if and only if $\theta_m = 1$.  Thus, EB asymptotically rejects any model $m$ for which     $E[\Delta^z_{mi}(\Delta^z_{0i}-\Delta^z_{mi})] \neq 0$ as $\plim n^{-1} T^\text{EB}_m = \infty$ in that case.  Generically with misspecification, $E[\Delta^z_{mi}(\Delta^z_{0i}-\Delta^z_{mi})] \neq 0$ for $m=1,2$ and EB rejects both models. In the presence of misspecification a researcher is not guaranteed to learn the true nature of conduct with this model assessment procedure.

 \noindent \textbf{Cox Test (Cox):} The next testing procedure we consider is inspired by the \citeA{c61} approach to testing non-nested hypotheses. 
 To perform a Cox test, we  formulate two different pairs of null and alternative hypotheses for each model  $m$, based on the same moment conditions defined for RV.  Specifically, for model $m$ we formulate the null and alternative hypotheses:
  \begin{align}
 H_{0,m}^\text{Cox}\ : \ g_m = 0 \qquad \text{and} \qquad H_{A,m}^\text{Cox}\ : \  g_{-m} =  0
 \end{align}
 where $-m$ denotes the opposite of model $m$. 
 To implement the Cox test in our environment, one can follow \citeA{s92}. With $\hat g_m$ as the finite sample analogue of the moment conditions, the test statistic for model $m$ is
 \begin{align}
 T_m^\text{Cox} &= \frac{\sqrt{n}}{\hat \sigma_\text{Cox}}\bigg(\hat g_{-m}'\hat W  \hat g_{-m} -  \hat g_{m}'\hat W \hat g_m -(\hat g_{-m}-\hat g_m)'\hat W(\hat g_{-m}-\hat g_m)\bigg) \\
 &= \frac{2\sqrt{n}\hat g_m'\hat W(\hat g_{-m} -\hat g_m)}{\hat \sigma_\text{Cox}},
 \end{align}
 where $\hat \sigma_\text{Cox}^2 = 4\hat g_{-m}' \hat V_{mm}^\text{AR}  \hat g_{-m}$ is a consistent estimator of the asymptotic variance of the numerator of $T_m^\text{Cox}$ under the null. As shown in \citeA{s92}, this statistic is asymptotically distributed according to a standard normal distribution under the null hypothesis. Under the alternative, the mean of $T_m^\text{Cox}$ is negative, so the test rejects for values of $T_m^\text{Cox}$ below the $\alpha$-th quantile of a standard normal distribution. As for the case of RV, the asymptotic normal limit distribution requires that Assumption \ref{as:nodegen} is satisfied. 
 
 The Cox test maintains -- under the null and the alternative -- that either of the two candidate models is correctly specified. Thus, in the presence of misspecification, one is neither under the null nor the alternative making the properties of the test hard to characterize. 
 
 In practice, as $n \rightarrow \infty$, the Cox test statistic diverges. To see this, note that the $\plim$ of $T^\text{Cox}_m$ is given as:
\begin{align}
    \plim T^\text{Cox}_m &= \lim_{n\to\infty} \frac{2 \sqrt{n} g_{m}' W(g_{-m}-g_m)}{\sigma_{Cox}}\\
    &=\lim_{n\to\infty}\frac{2\sqrt{n} (\norm{W^{1/2}g_1} \norm{W^{1/2}g_2} \cos(\vartheta) - \norm{ W^{1/2}g_1}^2)}{\sigma_{Cox}}
\end{align}
where $\sigma^2_\text{Cox} = 4g_{-m}V_{mm}^\text{AR}g_{-m}$    and $\vartheta$ is the angle between $W^{1/2}g_1$ and $W^{1/2}g_2$.  Suppose now that model 1 is the true model, both models are misspecified, and model 1 has the better fit, i.e., $0<\norm{ W^{1/2}g_2}\inverse\norm{ W^{1/2}g_1}<1$. While the RV test will select in favor of model 1 in this case, the behavior of the Cox test depends on the angle $\vartheta$:
\begin{align}
  \plim T^\text{Cox}_1 = \begin{cases}
        +\infty, & \text{if } \cos(\vartheta) > \frac{\norm{ W^{1/2}g_1}}{\norm{ W^{1/2}g_2}} ,
        \\
        -\infty, & \text{if } \cos(\vartheta) < \frac{\norm{ W^{1/2}g_1}}{\norm{W^{1/2}g_2}} .
        \end{cases}
 \end{align}
If treated as a two sided test, Cox therefore rejects the true model with probability approaching one for all $g_1$ and $g_2$ except in the knife edge case of $\cos(\vartheta) = \norm{ W^{1/2}g_2}\inverse\norm{ W^{1/2}g_1}$. If treated as a one-sided test, Cox may still reject the true model if $\cos(\vartheta)$ is sufficiently small. By similar derivations and the ordering $\norm{ W^{1/2}g_1}\inverse\norm{ W^{1/2}g_2} > 1 > \cos(\vartheta)$, it follows that $\plim T^\text{Cox}_2 = -\infty$, i.e., the worse fitting model, model 2, is rejected with probability approaching one in large samples.
In summary, if considered as a two-sided test the Cox test will reject both models in large samples, while as a one-sided test, it can also lead the researcher to reject the true model of conduct even when the true model has better asymptotic fit than the wrong model.

\section{Marginal Cost Misspecification}\label{sect:MCMisp}
 \noindent  

In Section \ref{sec:hypo}, we consider the consequences of misspecification of markups, e.g., because of misspecification of the demand system. It is also possible that a researcher misspecifies marginal cost.  Here we show that testing with misspecified marginal costs can be reexpressed as testing with misspecifed markups so that the results in the previous section apply. As a leading example, we consider the case where the researcher specifies $\textbf{w}_\textbf{a}$ which are a subset of $\textbf{w}$.  This could happen in practice because the researcher does not observe all the variables that determine marginal cost or does not specify those variables flexibly enough in constructing $\textbf{w}_\textbf{a}$.\footnote{Under a mild exogeneity condition, the results here extend to the broader case where $\textbf{w}\neq \textbf{w}_\textbf{a}$.}

To perform testing with misspecified costs, the researcher would residualize $\boldsymbol p$, $\boldsymbol{\Delta}_1$, $\boldsymbol{\Delta}_2$ and $\boldsymbol z$ with respect to $\textbf{w}_\textbf{a}$ instead of $\textbf{w}$.  Let $y^\textbf{a}$ denote a generic variable $\boldsymbol {y}$  residualized with respect to $\textbf{w}_\textbf{a}$.  Thus, with cost misspecification, both  RV and AR depend on the moment\looseness=-1
\begin{align}\label{eq:omega_costmisc}
    E\!\left[z_i^\textbf{a}(p^\textbf{a}_i-\Delta^\textbf{a}_{mi})\right] =E\!\left[z_i^\textbf{a}(\Delta^\textbf{a}_{0i} - \Breve{\Delta}^{\textbf{a}}_{mi})\right],
\end{align} 
where $\Breve{\Delta}^\textbf{a}_{mi}={\Delta}^\textbf{a}_{mi}-\text{w}^\textbf{a}\tau$ and $\text{w}^\textbf{a}$ is $\textbf{w}$ residualized with respect to $\textbf{w}_\textbf{a}$. 
When price is residualized with respect to cost shifters $\textbf{w}_\textbf{a}$, the true cost shifters are not fully controlled for and the fit of model $m$ depends on the distance between $\Delta^{\textbf{a}z}_0$ and $\Breve\Delta^{\textbf{a}z}_m$.  For example, suppose model 1 is the true model and demand is correctly specified so that $\boldsymbol\Delta_0 = \boldsymbol\Delta_1$.  Model 1 is still falsified as $\Breve{\Delta}^\textbf{a}_{1}={\Delta}^\textbf{a}_{0} - \text{w}^\textbf{a}\tau$.  Thus, when performing testing with misspecified cost, it is as if the researcher performs testing with markups that have been misspecified by $-\text{w}^\textbf{a}\tau$.  We formalize the implications of cost misspecification on testing in the following lemma:

\begin{lemma}\label{cor:cost_misc}
   Suppose  $\emph{\textbf{w}}_\emph{\textbf{a}}$ is a subset of $\emph{\textbf{w}}$, all variables $y^\emph{\textbf{a}}$ are residualized with respect to $\emph{\textbf{w}}_\emph{\textbf{a}}$, and Assumptions \ref{as:momcond}, \ref{as:regcond}, and \ref{as:nodegen} are satisfied.  Then, with probability approaching one as $n\rightarrow \infty$,
    \begin{itemize}
        
        \item[(i)] \emph{RV} rejects the null of equal fit in favor of model 1 if $E\!\big[({\Delta}^{\emph{\textbf{a}}z}_{0i}-\Breve{\Delta}^{\emph{\textbf{a}}z}_{1i})^2\big] \!<\! E\!\big[({\Delta}^{\emph{\textbf{a}}z}_{0i}-\Breve{\Delta}^{\emph{\textbf{a}}z}_{2i})^2\big]$,
        
        \item[(ii)] \emph{AR} rejects the null of perfect fit for any model $m$ with $E\!\big[({\Delta}^{\emph{\textbf{a}}z}_{0i}-\Breve{\Delta}^{\emph{\textbf{a}}z}_{mi})^2 \big]\neq0$.
    
    \end{itemize}
\end{lemma} 

\noindent
Thus the effects of cost misspecification can be fully understood as markup misspecification. \looseness=-1

\section{Simulations}\label{app:simul}

We construct three Monte Carlo experiments that illustrate the finite sample performance of the RV test and the ability of the $F$-statistic to detect low power and size distortions generated by weak instruments.  

\paragraph{Monte Carlo environment:} We consider an environment that mirrors the market for yogurt in our empirical application, importing the same market structure, product characteristics, and cost shifters. We also adopt the estimated demand model of Table \ref{Tab:demand}, column 3, and the estimated cost parameters implied by the zero retail margin model as the true demand system and cost function in the data generating process. The cost shocks, $\omega$, are drawn from a normal distribution with the same mean and standard deviation as the empirical distribution of the estimated shocks. For each of the three experiments, we perform 1,000 simulations; in each simulation, we sample 500 markets uniformly at random from the set of 5,034 yogurt markets to limit computational burden, and we take draws from the distribution of the cost unobservables. We solve for prices and quantities according to a true model of conduct, which differs across the three experiments as discussed below. Note that in all three experiments, we assume that the researcher knows the true demand model, which allows us to avoid reestimating demand for thousands of simulations.

\paragraph{Experiment 1:} In this experiment we illustrate how weak instruments for power affect the RV test, and how the $F$-statistic diagnoses weak instruments for power in finite samples. We consider as the true model of conduct Nash price setting, and we test the Nash price setting hypothesis against a model of joint profit maximization. To vary the power of the test, we consider instruments that are found to be strong in our application, the ``NoProd'' (number of own and rivals' products). We then inject noise by adding a vector of random values $\boldsymbol{u}$ drawn from a uniform distribution to construct a sequence of instruments with varying strength: $\boldsymbol{z}'(\lambda_1)=(1-\lambda_1) \boldsymbol{z} + \lambda_1 \boldsymbol{u}$ with values of $\lambda_1$ ranging from 0 to 0.9. As $\lambda_1$ increases, the instruments become more noisy.


\begin{center}
\begin{figure}[!ht]
    \centering
        \caption{Power reduction from weak instruments}
\includegraphics[width=0.9\textwidth]{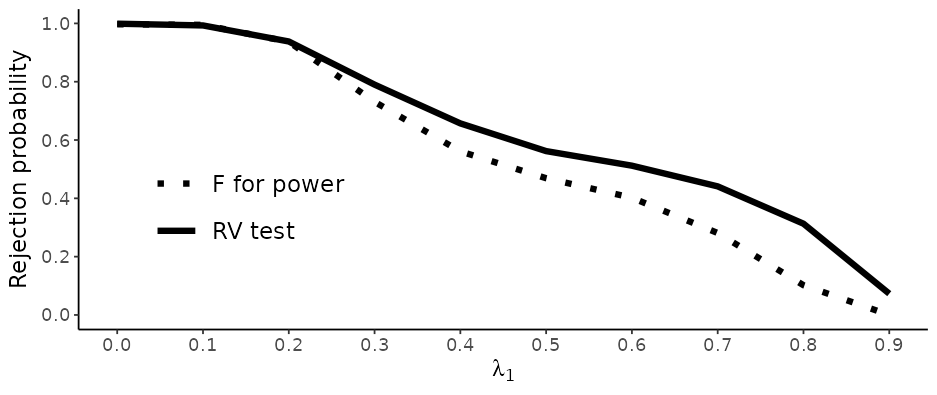}
\label{mc_power}
    \begin{tablenotes}[flushleft]\scriptsize
        \item 
        The figure plots the rejection probabilities of the RV test (solid line) and the fraction of simulations when the $F$-statistic exceeds the critical value for best-case power of 0.95 (dashed line) as a function of $\lambda_1$. As $\lambda_1$ increases, we inject more noise into the NoProd instruments. 
    \end{tablenotes}
\end{figure}
\end{center}

Figure \ref{mc_power} shows the simulated rejection probability of the RV test and the fraction of simulations for which the $F$-statistic diagnoses that the instruments are strong for best-case power of 0.95. As we increase the noise in the instruments, the power of the test reduces monotonically, from always rejecting to power essentially equal to size. Appropriately, our proposed $F$-statistic detects this reduction in power by increasingly diagnosing instruments as weak for power.

\paragraph{Experiment 2:} To perform the second experiment, we need an environment where we can approach the region of degeneracy while staying in the null space. Example 3 from the paper provides an economically meaningful environment in which we can accomplish this goal. In particular, we consider ``rule-of-thumb" pricing models, and we assume that prices are set by firms as two times cost, or $\boldsymbol{p}=2\boldsymbol{c}_0=2(\textbf{w}\gamma+\boldsymbol{\omega})$. We then test between two  alternative rule-of-thumb models: $\boldsymbol{\Delta}_1=0.5 \boldsymbol{c}_0+\lambda_2 \boldsymbol{z}+0.01\boldsymbol{\Tilde{e}}_1$ and $\boldsymbol{\Delta}_2=1.5\boldsymbol{c}_0+\lambda_2 \boldsymbol{z}+0.01\boldsymbol{\Tilde{e}}_2$, where the instrument $\boldsymbol{z}$ is the sum of rivals' products, $\lambda_2$ determines the correlation between the instrument and markups for models 1 and 2, and $(\boldsymbol{\Tilde{e}}_1,\boldsymbol{\Tilde{e}}_2)$ are random draws from a standard normal distribution. By considering alternatives where markups are specified as multipliers of costs plus the instruments (rescaled by a parameter), we can stay under the null (as $\Gamma_1=\Gamma_2$) while varying the strength of the instruments for size: as $\lambda_2$ gets closer to zero, we move towards the space of degeneracy (as $\Gamma_1 - \Gamma_0$ and $\Gamma_2 - \Gamma_0$ approach zero). Random components $(\boldsymbol{\Tilde{e}}_1,\boldsymbol{\Tilde{e}}_2)$ are included to satisfy Assumption 3.\looseness=-1

\begin{center}
\begin{figure}[ht]\caption{Size distortions from weak instruments}\label{mc_size}
\begin{threeparttable}
    \centering
\includegraphics[width=0.9\textwidth]{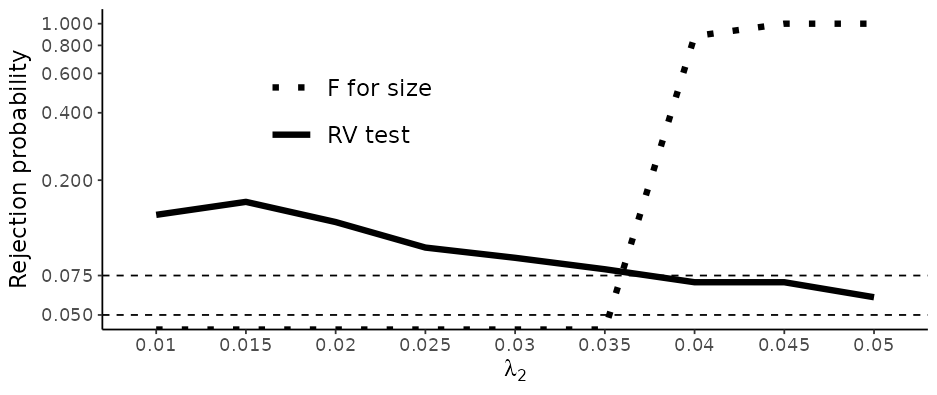}\\
    
    \label{fig:enter-label}
    \begin{tablenotes}[flushleft]\scriptsize
        \item 
        The figure plots the rejection probabilities of the RV test (solid line) and the fraction of simulations when the $F$-statistic exceeds the critical value for worst-case size of 0.075 (dashed line) as a function of $\lambda_2$. As $\lambda_2$ increases, we move away from the space of degeneracy, while remaining under the null.
    \end{tablenotes}
\end{threeparttable}
    
\end{figure}
\end{center}

Figure \ref{mc_size} shows the simulated rejection probabilities of the RV test and the fraction of simulations for which the $F$-statistic diagnoses that the instruments are strong for size. Near degeneracy, i.e. for small values of $\lambda_2$, we find large size distortions for the RV test (around 0.15). As $\lambda_2$ increases, and we move away from the space of degeneracy, size approaches the level of the test. Our diagnostic appropriately detects when instruments are weak for size.

\paragraph{Experiment 3:} In the third experiment, we explore the degree to which the RV test is affected by misspecification of demand. To do so, we test between different models of conduct while using a misspecified demand model to construct markups. Prices and shares are simulated from a joint profit maximization model, using the estimated demand model of Table 3, column 3. To approximate the effects of misspecification of demand for testing conduct, we construct markups for joint profit maximization, Nash price setting, and Perfect Competition, while altering the value of the coefficient $\Tilde{\beta}^p_{I}$ of the interaction between prices and income, in the spirit of the misspecification introduced in Example 1 of the paper. We vary $\Tilde{\beta}^p_{I}$ by adding values that correspond to $\pm 0.5, \pm 1, \pm 2$ times the standard error for the coefficient (0.378). 

Table \ref{mc_miss} reports the fraction of simulations for which the RV test has the correct sign, i.e. the true model of joint profit maximization (JPM) with potentially misspecified demand fits the data better than the wrong model with misspecified demand.  We find that the RV statistic has the correct sign in virtually all simulations, showing that the RV test may still be able to conclude for the true JPM model with misspecified demand.

\begin{table}[ht]\caption{Fraction of Simulations with Correctly Signed RV Test Statistics}\label{mc_miss}
\centering
\begin{threeparttable}
\begin{widetable}{.95\columnwidth}{l  c c}\toprule 
  &  \multicolumn{2}{c}{$\Pr(T^{\text{RV}}<0)$}\\ 
  \cmidrule{2-3}    &  JPM vs Nash price setting & JPM vs Perfect Competition\\
  $\Tilde{\beta}^p_{I}$ & (1) & (2) \\
  \midrule
$5.09$ &0.88 &0.99 \\
 $4.71$ &0.96& 1 \\
 $4.52$ &0.98 & 1  \\
 $\mathbf{4.33}$ & 1 & 1 \\
 $4.14$ & 1 & 1 \\
 $3.95$ &1& 1   \\
 $3.57$ &1 &0.97 \\
\bottomrule
\end{widetable}

\begin{tablenotes}[flushleft] 
\setlength\labelsep{0pt}
    \footnotesize
    \item The table reports the fraction of simulations for which $T^{RV}$ is correctly signed, i.e. suggests better fit of the joint profit maximization (JPM) model relative to Nash price setting (column 1) and perfect competition (column 2). Different rows correspond to different values of the demand parameter $\Tilde{\beta}^p_{I}$ 
 used in testing. Data are generated for the for the true value $\Tilde{\beta}^p_{I} = 4.33$.
\end{tablenotes}
\end{threeparttable}

\end{table}

\section{Additional Empirical Details}\label{sect:Robustness}
 
This appendix provides additional details for our empirical application. 

\noindent \textbf{Code Details:} Below is the definition of the demand estimation problem in \texttt{PyBLP} and the testing problem in \texttt{pyRVtest}.
\begin{figure}[H]
    \centering
    \caption{Code for Demand Estimation and Testing}
    \begin{subfigure}[t]{0.47\textwidth}
        \centering
        \caption*{Panel A. Demand Estimation Code}
        \includegraphics[scale=0.24]{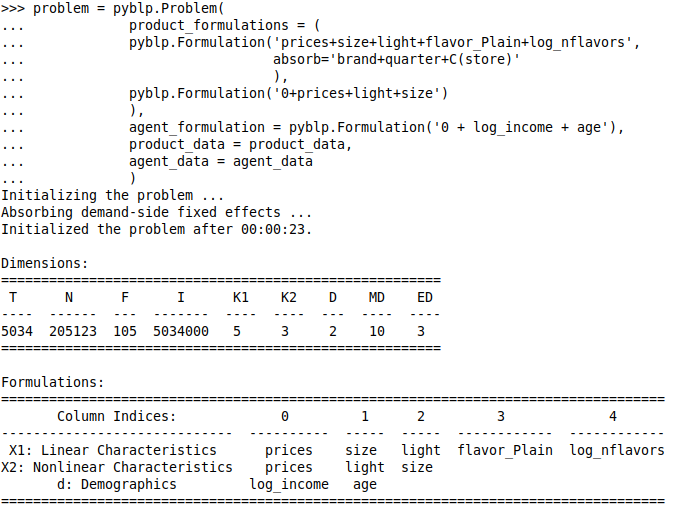}
    \end{subfigure}
    \begin{subfigure}[t]{0.52\textwidth}
        \centering
        \caption*{Panel B. Testing Code}
        \includegraphics[scale=0.24]{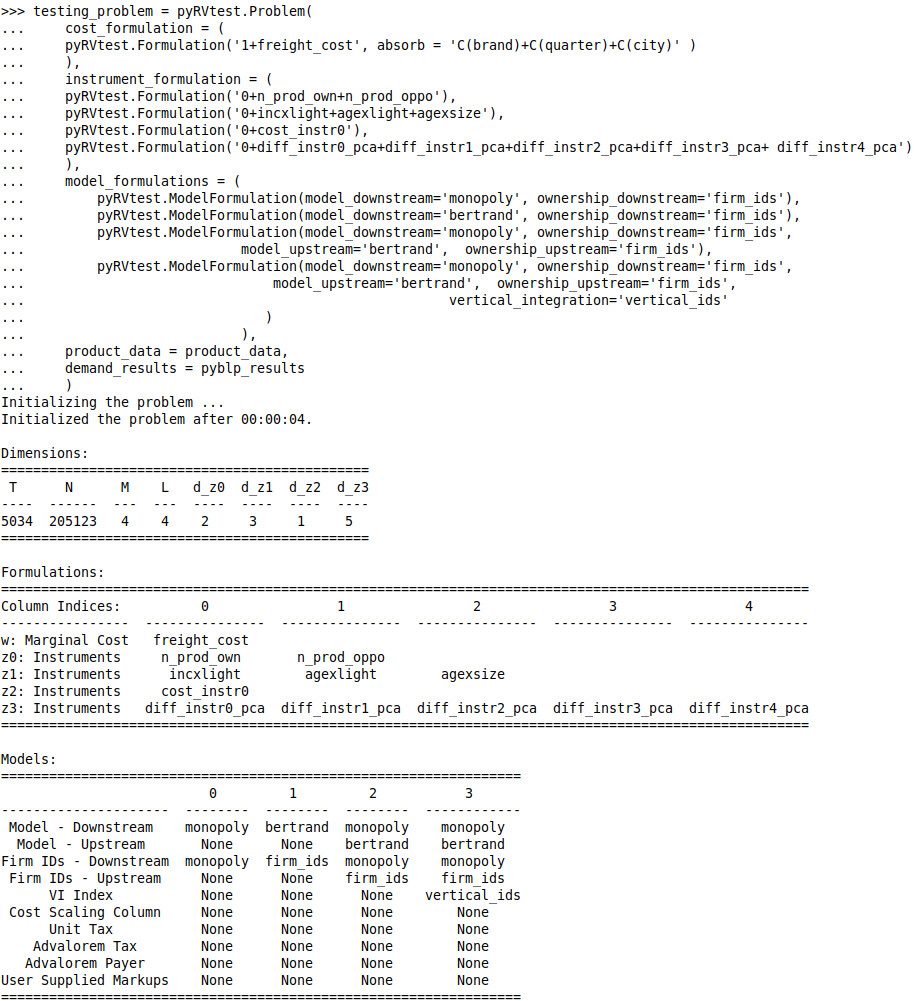}
    \end{subfigure}
    \caption*{\footnotesize{Panels A and B report the definitions of the demand estimation problem in \texttt{PyBLP} and the testing problem in \texttt{pyRVtest}, respectively.}}
\end{figure}


\noindent \textbf{Description of the Models of Conduct:} Following \cite{v07}, the markups in market $t$ for a model $m$ among those we consider can be written in the following form:
\[\boldsymbol{\Delta}_{mt}=\underbrace{-(\boldsymbol{\Omega}^r_{mt} \odot \boldsymbol{D}^r_t)^{-1} \boldsymbol{s}_t}_{=\boldsymbol{\Delta}^\text{downstream}_{mt}}\underbrace{-(\boldsymbol{\Omega}^w_{mt} \odot \boldsymbol{D}^w_t )^{-1} \boldsymbol{s}_t}_{=\boldsymbol{\Delta}^\text{upstream}_{mt}}\]
where $\boldsymbol{\Omega}^r_{mt}$ and $\boldsymbol{\Omega}^w_{mt}$ are ownership matrices, $\boldsymbol{D}^w_t$ is the jacobian of retail share $\boldsymbol{s}_t$ with respect to wholesale price, and $\boldsymbol{D}^r_t$ is the jacobian of retail share with respect to retail price. The markup $\boldsymbol{\Delta}_{mt}$ implied by each model is the sum of downstream markups $\boldsymbol{\Delta}^\text{downstream}_{mt}$ and upstream markups $\boldsymbol{\Delta}^\text{upstream}_{mt}$. We can derive each model by using different assumptions on the ownership matrices:
\begin{enumerate}
    \item \textit{Zero wholesale margin}: Set $\boldsymbol{\Omega}^w_{mt}$ to a matrix of zeros, set $\boldsymbol{\Omega}^r_{mt}$ to a matrix of ones.
    \item \textit{Zero retail margin}: Set $\boldsymbol{\Omega}^w_{mt}$ to a matrix of zeros, and set $\boldsymbol{\Omega}^r_{mt}$ to a matrix  with element $(i,j)$ that is equal to one if products $i$ and $j$ are produced by the same manufacturer, and to zero otherwise. 
    \item \textit{Linear pricing}: Set $\boldsymbol{\Omega}^r_{mt}$ to a matrix of ones, and set $\boldsymbol{\Omega}^w_{mt}$ to a matrix  with element $(i,j)$ that is equal to one if products $i$ and $j$ are produced by the same manufacturer, and to zero otherwise. 
    \item \textit{Hybrid model}: Set $\boldsymbol{\Omega}^r_{mt}$ to a matrix of ones, and set $\boldsymbol{\Omega}^w_{mt}$ to a matrix  with element $(i,j)$ that is equal to one if products $i$ and $j$ are produced by the same manufacturer and $i$ is not a private label, and to zero otherwise. 
    \item \textit{Wholesale Collusion}: Set $\boldsymbol{\Omega}^r_{mt}$ and $\boldsymbol{\Omega}^w_{mt}$ to matrices of ones. \looseness=-1
    
\end{enumerate}

\noindent \textbf{Assessment Tests Results:}
We present in Table \ref{Tab:Cox_EB} the test results from the two other model assessment procedures defined in Appendix \ref{sect:EBCox}. With NoProd instruments, the EB and Cox tests reject all models. Similar to the AR results in Table \ref{Tab:Model_assessment_all}, the Cox and EB results illustrate the pitfalls of using model assessment tests for conduct, formally discussed in Section \ref{sec:hypo}.

\begin{table}[ht]
\footnotesize
\caption{Cox and EB Test Results}
\label{Tab:Cox_EB}
\centering
\begin{threeparttable}
\begin{widetable}{.98\columnwidth}{lccc}
\toprule
\textbf{Cox Test} & 2 & 3 & 4  \\
\cmidrule(lr){2-4}
1. Zero wholesale margin&-4.72, 3.99&4.34, -6.53&4.35, -6.55\\
2. Zero retail margin &&9.47, -12.42&9.47, -12.46\\
3. Linear pricing&&&0.01, -0.01\\
4. Hybrid model&&&\\
\midrule
\textbf{EB Test} & 2 & 3 & 4  \\
\cmidrule(lr){2-4}
1. Zero wholesale margin&253.16, 319.1&253.16, 338.14&253.16, 337.97\\
2. Zero retail margin &&319.1, 338.14&319.1, 337.97\\
3. Linear pricing&&&338.14, 337.97\\
4. Hybrid model&&&\\
\bottomrule
\end{widetable}
\begin{tablenotes}[flushleft]
    \setlength\labelsep{0pt}
    \footnotesize
    \item Each cell reports the respective pair of test statistics $T_i,T_j$ for row model $i$ and column model $j$, with NoProd instruments. For 95\% confidence, the critical value for Cox is $\pm1.95$, and EB is $5.98$. Standard errors account for two-step estimation error and clustering at the market level; see Appendix \ref{sect:TwoStep}.\looseness=-1
\end{tablenotes}
\end{threeparttable}
\end{table}

\noindent \textbf{Pooling all Instrument Sets:}  We present in Table \ref{Tab:Instr_Rob} the test results obtained by pooling all instrument sets used in Section \ref{sec:empiricaleg}.  Two observations emerge from the table. First, the $F$-statistics with the pooled instruments, although above the corresponding critical values for 95\% best-case power, are below those for the NoProd instrument set. This confirms our observation that pooling instruments may result in lower $F$-statistics, and potentially weak instruments. Second, the MCS for both sets of pooled instruments now includes not only model 2, but also model 1 at a confidence level $\alpha = 0.05$. Hence, pooled instrument sets lead to less sharp conclusions in this example.

\begin{table}[ht]
\footnotesize
\caption{RV Test Results with Alternative Instruments}
\label{Tab:Instr_Rob}
\centering
\begin{threeparttable}

\begin{widetable}{.98\columnwidth}{lrrrrrrc}
\toprule
& \multicolumn{3}{c}{$T^\text{RV}$}& \multicolumn{3}{c}{$F$-statistics}&  MCS $p$-values\\
\cmidrule(lr){2-4}  \cmidrule(lr){5-7} \cmidrule(lr){8-8}  
Models & 2 & 3 &  4 & 2 & 3 & 4 & \\ 
\midrule
\multicolumn{8}{l}{\textbf{IVs: No. Products + Demo + Cost + Diff} ($d_z=11$)}  \\ 
1. Zero wholesale margin&1.664&  -8.827&  -8.838&105.3 &   {70.7}   &  {70.8}&0.096 \\
2. Zero retail margin&&-5.655 & -5.653&&{85.6}  &   {86.2} &1.0      \\
3. Linear pricing&&&-1.781&&&{101.1}&0.00 \\
4. Hybrid model&&&&&&&0.00 \\
\bottomrule
\end{widetable}

\begin{tablenotes}[flushleft]
    \setlength\labelsep{0pt}
    \footnotesize
    \item Table reports the RV test statistics $T^\text{RV}$ and the effective $F$-statistic for all pairs of models, and the MCS $p$-values (details on their computation are in Appendix F). A negative RV test statistic suggests better fit of the row model. With MCS $p$-values below 0.05 a row model is rejected from the model confidence set. All $F$-statistics exceed the appropriate critical values for a worst-case size of  0.075 and best-case power of 0.95. Both $T^\text{RV}$ and the $F$-statistics account for two-step estimation error and clustering at the market level; see Appendix \ref{sect:TwoStep} for details.  \looseness=-1
\end{tablenotes}
\end{threeparttable}
\end{table}

 \section{MCS Details}\label{sect:MCSDetail}

In this appendix, we provide further details on the construction of the model confidence set.  The construction is iterative where in each step a $p$-value for a model of worst fit is computed.  For a set of instruments $z_\ell$, the MCS $M^*_\ell$ is the collection of models with a $p$-value above the significance level.

For a researcher testing a set of candidate models $M$, we now describe the construction of the MCS $p$-values at each step of the algorithm.  
At iteration $k$, there are $M_k\subset M$ models under current consideration.  Denoting $\abs{M_k}$ as the cardinality of $M_k$, this results in $K = \binom{\abs{M_k}}{2}$ distinct pairs of models and therefore RV test statistics. To make notation compact, we define $\Pi$ as a one-to-one mapping from unique model pairs to $\{1,\hdots,K\}$, and let $T^\text{RV}_{\Pi(m_1,m_2)}=  {\sqrt{n}(\hat{Q}_{m_1}-\hat{Q}_{m_2})}/{\hat\sigma_{\text{RV},\Pi(m_1,m_2)}}$.  In this iteration, we find the pair of models $(m_1,m_2)$ associated with the largest test statistic in magnitude, denoted $\bar{T}^\text{RV}$.  If $\bar{T}^\text{RV}$ is positive (negative), then $m_1$ ($m_2$) is the model of worst fit for which we compute the MCS $p$-value.  This model will be dropped from $M_{k+1}$ in the next iteration.

For computation of the MCS $p$-values, we utilize that $    \big(T^\text{RV}_1,\hdots,T^\text{RV}_K \big)'$ is asymptotically normal with zero mean and a variance $\Sigma$ under no degeneracy and a null of equal fit.  This observation follows by extending Lemma \ref{lem:asymp norm} to $\abs{M_k}$ models and alluding to a first order Taylor approximation as in Remark \ref{rem:validity}.  We estimate $\Sigma$ using $\hat\Sigma$.  The diagonal entries of $\hat\Sigma$ are all one.  For any off-diagonal element $\hat\Sigma_{ij}$ with $(i,j) = (\Pi(m_1,m_2),\Pi(m_3,m_4)) $ we have:\looseness=-1
\begin{align}
    \hat\Sigma_{ij} = 4\hat\sigma^{-1}_{\text{RV},i}\hat\sigma^{-1}_{\text{RV},j}&\bigg(\hat g_{m_1}'\hat W^{1/2}\hat V^{\text{RV}}_{m_1m_3}\hat W^{1/2}\hat g_{m_3} -\hat g_{m_2}'\hat W^{1/2}\hat V^{\text{RV}}_{m_2m_3}\hat W^{1/2}\hat g_{m_3} \\&\qquad- \hat g_{m_1}'\hat W^{1/2}\hat V^{\text{RV}}_{m_1m_4}\hat W^{1/2}\hat g_{m_4} +  \hat g_{m_2}'\hat W^{1/2}\hat V^{\text{RV}}_{m_2m_4}\hat W^{1/2}\hat g_{m_4}\bigg).
\end{align}

We take 99,999 draws from this distribution and for each draw compute the max of the absolute value of the $K$ simulated test statistics denoted $\bar{T}^{\text{RV},sim}$.  The $p$-value is then computed as the fraction of draws for which  $\bar{T}^{\text{RV},sim} > \abs{\bar T^\text{RV}}$.

This procedure is computationally expedient.  In particular, one does not need to bootstrap both demand estimation and the testing procedure, which can be prohibitively costly.  Instead, using the adjustments derived in Appendix \ref{sect:TwoStep}, $\hat\Sigma$ can incorporate adjustments for demand estimation and clustering.

\section{Simultaneous versus Sequential Demand Estimation: Implications for Testing}\label{sect:SimSec}
 Consider the sequential setting in the paper, where demand is estimated in a preliminary step, and then testing is performed with the RV test. The lack of fit for a model of conduct $m$ is measured by the GMM objective function formed with the supply moments: $\hat g_m\left(\theta^{D}\right) = \frac{1}{n}\hat z'(\hat p-\hat \Delta_m(\theta^{D}))$, and the hat notation denotes that these are finite sample objects, residualized against the observed cost shifters $\textbf{w}$. We augment the framework in the paper by allowing the supply moments to depend on demand parameters $\theta^D$ through markups $\Delta_m=\Delta_m(\theta^D)$. Demand moments are instead $\hat h\left(\theta^{D}\right) = \frac{1}{n}\textbf{z}^{D\prime}\boldsymbol{\xi}(\theta^{D})$,
where $\boldsymbol{\xi}(\theta^{D})$ is the value of the demand unobservable implied by $\theta^D$, and $\textbf{z}^{D}$ is a vector of demand instruments.\footnote{Demand moments could include micro-moments; our discussion in this section would still apply.} We assume that the dimension of $\textbf{z}^D$ is sufficient to identify $\theta^{D}$ from demand moments alone. \looseness=-1  

We define the plim of the demand parameter estimates obtained from sequential estimation as:
\begin{align*}
    \theta^{Dseq} = \arg\min_{\theta^D} {Q}^D(\theta^D),
\end{align*}
where ${Q}^D(\theta^D)= h(\theta^{D})' W^{D} h(\theta^{D}),$ $h(\theta^{D}) = \text{plim} \hat h(\theta^{D})$, and $ W^{D} $ is the plim of the weight matrix used in demand estimation. We construct the RV test statistic in a second step from a measure of fit based on supply moments.  The asymptotic value of the measure of fit for each model $m$ is:
\[
{Q}_{m}(\theta^{Dseq})={g}_m(\theta^{Dseq})'W {g}_m(\theta^{Dseq}).
\]

Alternatively, we could pursue a simultaneous approach, where we use the supply moments to estimate demand under each model of conduct. Stacking the demand and supply moments, the simultaneous GMM objective function is:
\begin{align*}
Q_m^{sim}(\theta^D) = (g_m(\theta^D)',h(\theta^D)')W^{sim}(g_m(\theta^D)',h(\theta^D)')'
\end{align*}
We maintain the following assumption on $W^{sim}$ for tractability:
\begin{assumption}\label{as:blockdiag}
    $W^{sim}$ is block diagonal, so that $W^{sim}=diag(W, W^D)$.  
\end{assumption}
\noindent The plim of the demand parameters obtained from simultaneous estimation is: $$\theta_m^{Dsim} = \arg\min_{\theta^{D}} {Q}_m(\theta^{D}) + {Q}^{D}(\theta^{D}).$$ 
\noindent While the value of $\theta_m^{Dsim}$ is model specific, $\theta^{Dseq}$ is not. We now establish the following:
\begin{proposition}\label{prop:demand}
Let Assumption \ref{as:blockdiag} hold, $\normalfont \textbf{w}$ be correctly specified, $\theta_m^{Dsim}$ be unique, and $\theta^{Dseq}\neq\theta_m^{Dsim}$ for model $m$. Then, $E[( \Delta^z_0 -  \Delta^z_m(\theta_m^{Dsim}))^2] <         E[( \Delta^z_0 -  \Delta^z_m(\theta^{Dseq}))^2]$.
\end{proposition}

The consequences of the proposition depend on the model of conduct $m$ and the testing environment. Suppose $m=1$ is the true model and there is no misspecification.  Then $\theta^{Dseq}=\theta_1^{Dsim}$, and $ E[( \Delta^z_0 - \Delta^z_1(\theta_1^{Dsim}))^2] =         E[( \Delta^z_0 - \Delta^z_1(\theta^{Dseq}))^2] = 0.$  However, predicted markups for $m=2$ also get closer to the true ones.  It follows that  ${Q}_1 - {Q}_2$ is smaller in magnitude under simultaneous estimation.  Instead, when demand is misspecified, the proposition implies that the MSE in predicted markups is smaller with a simultaneous approach for each model.  We summarize these findings as follows:
\begin{corollary}\label{cor:demand}
Let Assumption \ref{as:blockdiag} hold, $\theta_m^{Dsim}$ be unique, and $\theta^{Dseq}\neq\theta_m^{Dsim}$ for model $m$. 
\begin{itemize}
    \item[(i)] Without misspecification, the plim of the RV numerator is smaller in magnitude under simultaneous estimation, but has the same sign as under sequential estimation.
    \item[(ii)]  With misspecification, the sign and magnitude of plim of the RV numerator could differ under simultaneous and sequential estimation.
\end{itemize}
\end{corollary}

\noindent While we focus on the numerator, the mode of demand estimation will also affect the denominator of $T^{RV}$. How these forces affect the power of RV is ambiguous even under our simplifying assumptions and we leave it to future research.

The demand estimation literature (see \cite{cg19}) has discussed how simultaneous and sequential estimation affect estimates and standard errors for the demand parameters. Our results show that simultaneous estimation does not guarantee demand parameter estimates closer to the true demand parameters. Instead, the supply moments pull the demand parameters towards values which improve the fit of the predicted markups under the assumed conduct model. Moreover, efficiency considerations that apply to estimation do not immediately apply to our setup as one of the models we consider is not correctly specified. \looseness=-1

\section{Proof of Appendix results}\label{sec:suppl}
\begin{remark}\label{rem:ignore w}
As a prologue to the proof of Lemma \ref{lem:asymp norm}, we remind the reader that the first order properties of $\hat W\inverse = n\inverse \hat z'\hat z$ and the infeasible $\check W\inverse =n\inverse z' z$ are the same. This follows from the equality $n\inverse \hat z'\hat z = n\inverse z'\hat z$, which in turn leads to
\begin{align}
	\hat W\inverse \!=\! \check W\inverse \!+\! \underbrace{n\inverse z' \textbf{w} \phantom{\big)}}_{=O_p(n^{-1/2})}\underbrace{\big( E[\textbf{w}'\textbf{w}]\inverse E[\textbf{w}'\boldsymbol{z}] \!-\! (\textbf{w}'\textbf{w})\inverse \textbf{w}'\boldsymbol{z} \big)}_{=O_p(n^{-1/2})} \!=\! \check W\inverse \!+\! O_p(n\inverse).
\end{align}	\vspace{-12pt}\\
The same argument applied to $\hat g_m = n\inverse \hat z'(\hat p - \hat \Delta_m)$ yields $\hat g_m = \check g_m  + O_p(n\inverse)$ where $\check g_m = n\inverse z'(p - \Delta_m)$.
\end{remark}

\begin{remark}\label{rem:binomial}
For a matrix $A$ with all singular values strictly below 1, the proof of Lemma \ref{lem:asymp norm} relies on the binomial series expansion $(I+A)^{-1/2} = \sum_{j=0}^\infty \binom{-1/2}{j} A^j = I - \frac{1}{2}A + \frac{3}{8} A^2 - \frac{5}{16}A^3 + \dots$, where the generalized binomial coefficient is $\binom{\alpha}{j}  = (j!)\inverse \prod_{k=1}^j(\alpha-k+1)$. 
\end{remark}

\begin{proof}[Proof of Lemma \ref{lem:asymp norm}]
	We prove $(i)$ and $(ii)$ in three steps and then comment on the modifications to these steps that derive $(iii)$ and $(iv)$. Step (1) shows that $\frac{1}{\sqrt{n}} \sum_{i=1}^n (\psi_{1i}',\psi_{2i}')'\xrightarrow{d} N(0,V^\text{RV})$ and $\check V^\text{RV}_{\ell k} := \frac{1}{n} \sum_{i=1}^n \psi_{\ell i} \psi_{ki}' \xrightarrow{p} V^\text{RV}_{\ell k}$ for $\ell,k \in \{1,2\}$, step (2) establishes that $\sqrt{n}\big(\hat W^{1/2} \hat g_m - W^{1/2}g_m\big) - \frac{1}{\sqrt{n}}  \sum_{i=1}^n \psi_{mi} = o_p(1)$ for $m \in\{1,2\}$, and step (3) proofs that $\text{trace}\big( (\hat V^\text{RV}_{\ell k}-\check V^\text{RV}_{\ell k})'(\hat V^\text{RV}_{\ell k}-\check V^\text{RV}_{\ell k}) \big) = o_p(1)$ for $\ell,k \in \{1,2\}$. The combination of steps (1) and (2) establishes $(i)$ while steps (1) and (3) yields $(ii)$.

	\noindent
	\textbf{Step (1)} From Assumption \ref{as:regcond}, part (i) and (iii), it follows from the standard central limit theorem for iid data that $\frac{1}{\sqrt{n}} \sum_{i=1}^n (\psi_{1i}',\psi_{2i}')u \xrightarrow{d} N(0,u'V^\text{RV}u)$ for any non-random $u \in \mathbb{R}^{2d_z}$ with $\norm{u}=1$. The Cram\'er-Wold device therefore yields  $\frac{1}{\sqrt{n}} \sum_{i=1}^n (\psi_{1i}',\psi_{2i}')'\xrightarrow{d} N(0,V^\text{RV})$. Additionally, a standard law of large numbers applied element-wise implies $\check V_{\ell k} \xrightarrow{p} V_{\ell k}$ for $\ell,k \in \{1,2\}$.

	\noindent
	\textbf{Step (2)} From standard variance calculations it follows that
	\begin{align}\label{eq:rates}
		\check W - W = O_p(n^{-1/2}) \qquad \text{and} \qquad \check g_m - g_m = O_p(n^{-1/2}).
	\end{align}
	In turn, Equation \eqref{eq:rates} together with Remark \ref{rem:ignore w} and \ref{rem:binomial} implies that
	\begin{align}\label{eq:Taylor}
		\hat W^{1/2} - W^{1/2} &= W^{1/4} \left( \big(I + W^{1/2}(\hat W \inverse - W\inverse) W^{1/2} \big)^{-1/2} - I \right)W^{1/4} \\
		&=-\tfrac{1}{2}W^{1/4} \left( W^{1/2}(\hat W \inverse - W\inverse) W^{1/2}\right)W^{1/4} + O_p(n\inverse)
	\end{align}
	Combining Equations \eqref{eq:rates} and \eqref{eq:Taylor}, we then arrive at %
	\begin{small}%
	\begin{align}
		\sqrt{n}\big(\hat W^{1/2}\hat g_m \!-\! W^{1/2} g_m \big)\! &= W^{1/2} \big( \hat g_m \!-\! g_m\big)\!+\!\big( \hat W^{1/2} \!-\! W^{1/2} \big)g_m \!+\! O_p\big( n^{-1} \big) \\
		&=W^{1/2} \big( \hat g_m \!-\! g_m\big) \!-\! \tfrac{1}{2}W^{3/4}\big( \hat W^{-1} \!-\! W^{-1} \big)W^{3/4} g_m \!+\! O_p\big( n^{-1} \big) \\
		&=W^{1/2} \big( \check g_m \!-\! g_m\big) \!-\! \tfrac{1}{2}W^{3/4}\big( \check W^{-1} \!-\! W^{-1} \big)W^{3/4} g_m \!+\! O_p\big( n^{-1} \big) \\
		&=\tfrac{1}{n} \sum\nolimits_{i=1}^n \psi_{mi} \!+\! O_p\big( n^{-1} \big).
	\end{align}%
	\end{small}%
	\noindent
	\textbf{Step (3)} Letting $R_{m} = \frac{1}{n} \sum_{i=1}^n (\hat \psi_{mi} - \psi_{mi})'(\hat \psi_{mi} - \psi_{mi})$, it follows from matrix analogs of the Cauchy-Schwarz inequality that	
	\begin{align}
		\text{trace}\!\left( \!(\hat V^\text{RV}_{\ell k}\!-\!\check V^\text{RV}_{\ell k})'(\hat V^\text{RV}_{\ell k}\!-\!\check V^\text{RV}_{\ell k})\! \right)\! \le \! 4\!\left( \text{trace}(\check V^\text{RV}_{\ell \ell} )R_{k} \!+\! \text{trace}(\check V^\text{RV}_{kk} )R_{\ell} \!+\! R_{\ell} R_{k} \right)\!.
	\end{align}
	Therefore, it suffices to show that $R_{m} = o_p(1)$ for $m \in \{1,2\}$, since we have from step (1) that $\check V^\text{RV}_{mm} =O_p(1)$ for $m \in \{1,2\}$. To further compartmentalize the problem, we note that
	\begin{align}
		R_{m} \!&\le\! \tfrac{3}{n} \sum\nolimits_{i=1}^n \norm{\hat W^{1/2} \hat z_i(\hat p_i \!-\!\hat \Delta_{mi})\!-\!W^{1/2} z_i ( p_i \!-\!\Delta_{mi})}^2 \\
		\!&+\! \tfrac{3}{4n} \sum\nolimits_{i} \norm{\hat W^{3/4} \hat z_i \hat z_i'\hat W^{3/4} \hat g_m \!-\! W^{3/4}  z_i z_i' W^{3/4} g_m}^2   
		\!+\! \tfrac{3}{4} \norm{\hat W^{1/2}\hat g_m \!-\! W^{1/2} g_m}^2.
	\end{align}
	Argumentation analogous to Remark \ref{rem:ignore w} combined with Equation \eqref{eq:rates} yields $R_{m} = o_p(1)$. 
	
	Finally, note that to derive $(iii)$ and $(iv)$, one can follow the same line of argument. The only real difference is that Equation \eqref{eq:Taylor} gets replaced by
	\begin{align}
		\hat W - W &= -W\inverse (\hat W\inverse - W\inverse)W\inverse + O_p(n\inverse). \qedhere
	\end{align} 
\end{proof}
	
As a prelude to the proof of Lemma \ref{cor:cost_misc} we establish the following analogous result.

\begin{lemma}\label{cor:modelmisc}
   Suppose that Assumptions \ref{as:momcond}, \ref{as:regcond}, and \ref{as:nodegen}  are  satisfied. Then, with probability approaching one as $n\rightarrow \infty$,
    
    \begin{itemize}
        
        \item[(i)] \emph{RV} rejects the null of equal fit in favor of model 1 if $E\!\big[(\Delta^z_{0i}-\Delta^z_{1i})^2\big]< E\!\big[(\Delta^z_{0i}-\Delta^z_{2i})^2\big]$,

        \item[(ii)] \emph{AR} rejects the null of perfect fit for model $m$ with $E\!\big[(\Delta^z_{0i}-\Delta^z_{mi})^2\big] \neq 0$.\looseness=-1
        
    \end{itemize}
\end{lemma}

\begin{proof}[Proof of Lemma \ref{cor:modelmisc}.] For $(i)$, suppose for concreteness that $E[(\Delta^z_{0i}-\Delta^z_{1i})^2]<E[(\Delta^z_{0i}-\Delta^z_{2i})^2]$. We need to show that $\Pr\big( T^\text{RV} < -q_{1-\alpha/2}(N) \big) \rightarrow 1$ where $q_{1-\alpha/2}(N)$ is the $(1-\alpha/2)$-th quantile of a standard normal distribution. From Proposition \ref{prop:nulls}, part $(ii)$, we have $Q_1 < Q_2$. Lemma \ref{lem:asymp norm}, parts $(i)$ and $(ii)$, together with Remark \ref{rem:validity} yields $T^\text{RV}/\sqrt{n} = \frac{ \hat Q_1-\hat Q_2}{\hat \sigma_\text{RV}} \xrightarrow{p}  \frac{Q_1-Q_2}{\sigma_\text{RV}} < 0$. Therefore, $\Pr\big( T^\text{RV}/\sqrt{n} < -q_{1-\alpha/2}(N_{0,1})/\sqrt{n} \big) \rightarrow 1$.

For $(ii)$, suppose that $E[(\Delta^z_{0i}-\Delta^z_{mi})^2]\neq 0$ for some $m \in \{1,2\}$.  We need to show that $\Pr\big( \text{reject } H^\text{AR}_{0,m} \big) \rightarrow 1$. From Proposition \ref{prop:nulls}, part $(i)$, it follows that $\pi_m \neq 0$ and since $V^\text{AR}_{mm}$ is positive definite we have $\pi_m'(V^\text{AR}_{mm})\inverse \pi_m > 0$. Using Lemma \ref{lem:asymp norm}, parts $(iii)$ and $(iv)$, we have $(\hat \pi_m,\hat V^\text{AR}_{mm}) \xrightarrow{p} (\pi_m,V^\text{AR}_{mm})$ so that $T^\text{AR}_m/n = \hat\pi_m'\big(\hat V^\text{AR}_{mm}\big)\inverse \hat\pi_m \xrightarrow{p} \pi_m'\big( V^\text{AR}_{mm}\big)\inverse \pi_m > 0$. Therefore, $\Pr\big( \text{reject } H^\text{AR}_{0,m} \big) = \Pr\big( T^\text{AR}_m/n > q_{1-\alpha}(\chi^2_{d_z})/n \big) \rightarrow 1$ where $q_{1-\alpha}(\chi^2_{d_z})/n ) \rightarrow 1$ where $q_{d_z}(1-\alpha)$ denotes the $(1-\alpha)$-th quantile of a $\chi^2_{d_z}$ distribution.
\end{proof}

\begin{proof}[Proof of Lemma \ref{cor:cost_misc}.] 
Since $\textbf{w}_\textbf{a}$ is a subset of $\textbf{w}$, the proof follows immediately by replacing any variable $y$ residualized with respect to $\textbf{w}$ in the Proof of Lemma \ref{cor:modelmisc}, with $y^\textbf{a}$ which has been residualized with respect to $\textbf{w}_\textbf{a}$.
\end{proof}

\begin{proof}[Proof of Proposition \ref{prop:demand}.] 
By definition of $\theta_m^{Dsim}$ and because $\theta^{Dseq}$ is the minimizer of ${Q}^{dem}$, it must be that ${Q}^{sim}_m(\theta_m^{Dsim}) < {Q}^{sim}_m(\theta^{Dseq})$. In turn, this implies that ${Q}_m(\theta_m^{Dsim}) < {Q}_m(\theta^{Dseq})$ - if this inequality did not hold, we would need to have ${Q}^{D}(\theta^{Dseq})> {Q}^{D}(\theta^{Dsim}_m)$. Since each objective function is assumed to have a unique minimizer, this last inequality is a contradiction. Because we assumed that cost shifters are correctly specified, we have  ${Q}_{m}\left(\Delta(\theta^D)\right) = E[( \Delta^z_0 -  \Delta^z_m(\theta^D))^2]$. Therefore, 
    \begin{align*}
        E[( \Delta^z_0 -  \Delta^z_m(\theta_m^{Dsim}))^2] &<         E[( \Delta^z_0 -  \Delta^z_m(\theta^{Dseq}))^2]. \qedhere
    \end{align*}
\end{proof}

\bibliographystyleA{ecta}

\bibliographyA{biblio_RV}

\end{appendices}
}

\end{document}